%% file: current_draft.tex
\documentclass[a4paper,12pt,reqno,oneside,dvipsnames]{amsart}
\usepackage{graphicx}
\usepackage{subcaption}
\usepackage{adjustbox}
\usepackage{multicol}
\usepackage{adjmulticol}
\usepackage{caption}
\usepackage[bbgreekl]{mathbbol}

\usepackage{tikz}
\usetikzlibrary{arrows.meta, decorations.pathreplacing}

\input{supporting_files/preamble}

\allowdisplaybreaks

\begin{document}

\begin{titlepage}
\title{Robust Market Interventions}

\date{\today}
\author[]{Andrea Galeotti \and Benjamin Golub \and Sanjeev Goyal  \and \\  Eduard Talam\`{a}s \and Omer Tamuz}\thanks{%
Galeotti is at the London Business School, Golub is at Northwestern University, Goyal is at the University of Cambridge, Talam\`{a}s is at IESE Business School, and Tamuz is at Caltech. \\ \hspace*{\parindent}  We thank  Yann Calv\'{o} L\'{o}pez and Yu-Chi Hsieh for excellent research assistance. \\ \hspace*{\parindent}This work was supported by the ERC grant 724356 (Galeotti), the Pershing Square Fund for Research on the Foundations of Human Behavior (Golub), the Keynes Fund for Applied Economics (Goyal), the National Science Foundation (SES-1658940, SES-1629446, Golub \& DMS-1944153, Tamuz), the Sloan Foundation (Tamuz) and the US-Israel Binational Science Foundation (2018397, Tamuz). \\ \hspace*{\parindent}  We thank Bruno Pellegrino for valuable exchanges and for sharing data with us, and Volker Nocke
for a helpful discussion on a related project at the 2023 AEA meetings. We thank Pol Antr\`{a}s, Itai Ashlagi, Modibo Camara, Yeon-Koo Che, Kevin Chen, Rebecca Diamond, Jeff Ely, Joey Feffer,  Alex Frankel,  Matthew Gentzkow, Tim Gowers, Matthew O. Jackson, Zi Yang Kang, Ernest Liu, Margaret Meyer, Jose Luis Moraga, Stephen Morris, Martin Peitz, Ilya Segal, Andy Skrzypacz,  Ran Spiegler, Ludwig Straub, Xavier Vives, Kuang Xu, and Frank Yang for helpful comments. \\
\hspace*{\parindent} We thank the editor and four excellent referees for very helpful comments and guidance on previous versions that greatly improved the manuscript. \\
\hspace*{\parindent} The paper subsumes the working paper \citet{galeotti2022taxesmarketpowerprincipal}.}

\begin{abstract} 
When can interventions in markets be designed to increase surplus \emph{robustly}---i.e., with high probability---accounting for uncertainty due to imprecise information about economic primitives? In a setting with many strategic firms, each possessing some market power, we present conditions for such interventions to exist. The key condition, \emph{significant structure}, requires large-scale complementarities among families of products.  The analysis works by decomposing the incidence of interventions in terms of principal components of a Slutsky matrix. Under significant structure, a noisy signal of this matrix reveals enough about these principal components to design robust interventions. Our results demonstrate the usefulness of spectral methods for analyzing imperfectly observed strategic interactions with many agents. \end{abstract}
\newpage

\maketitle

\thispagestyle{empty}

\end{titlepage}

\setcounter{page}{1}

\section{Introduction}

Market power has recently attracted renewed attention and is thought to have significant and growing welfare implications \citep[see, e.g.,][]{syverson2019macroeconomics}.  While many applied studies of competition confine the study of market power to tightly defined product markets, it is becoming clear that important welfare-relevant spillovers operate \emph{across} such markets \citep{baqaee2020productivity,
azarvives2021,pelligrino2021,edererpelligrino2021}. Our theoretical understanding of such spillovers in environments with many firms interacting via general demand systems remains limited. This paper is about the welfare theory of such interactions.

Consider firms that set prices in a setting where  demands exhibit arbitrary complementarities and substitutabilities. For instance, one firm's product---e.g., a Samsung smartphone---may be a substitute for some products---e.g., Apple smartphones---and a complement to others---e.g., compatible accessories such as earbuds, watches, and smart home appliances---which may, in turn, be substitutes or complements to one another. %

We are interested in the nature of inefficiencies in such an environment and interventions to respond to them. We consider these issues from the perspective of an authority that recognizes the possibility of inefficiency due to market power and can intervene through taxes and subsidies on firms' sales. For concreteness, we can think of this authority as the operator of a large marketplace, such as Amazon, that mediates retail sales, or as a governmental institution. The authority's objective is to increase the equilibrium economic surplus generated by the marketplace. (Once this is achieved, the authority can reclaim and redistribute this additional surplus via transfers.) What principles should guide the design of such interventions? When can such policies be implemented under realistic uncertainty about market parameters?

What makes the problem challenging is that, once we broaden our perspective beyond one traditionally defined market (e.g., a set of imperfectly substitutable products such as smartphones) and consider spillovers to a variety of other complements and substitutes, there is a kind of curse of dimensionality. A marketplace with numerous and changing goods is described by a demand system that is high-dimensional and a priori unstructured. As any firm's cost changes, the number of potential spillover effects to consider is equal to the number of products; thus, the number of interactions scales quadratically in this number. Realistic signals leave substantial uncertainty about many aspects of the environment, so that the authority will have nothing close to a precise estimate of the entire demand system. This raises the question of whether there exist policies that can reliably improve surplus despite this uncertainty.

Our main result is that if the economic environment and the statistical signal about it jointly satisfy a property that we call \emph{significant structure}, then there are feasible intervention rules that robustly (i.e., with high probability) increase equilibrium total surplus. They do so despite the fact that the authority faces large errors in observing every detail of the system. Moreover, within a natural class of interventions---those whose direct (pre-transfer) incidence on consumer surplus is nonnegative---our feasible interventions achieve the largest gain in surplus that is possible for a given level of subsidy expenditure. Hence, within this class, these interventions are as good as those that could be designed by an authority with perfect information.

The key condition in the paper is significant structure. We now explain what it means and how it is used.

The demand structure is encoded in a matrix $\bm{D}$ of demand derivatives (which in our setting is equal to the Slutsky matrix). A given cell $D_{ij}$ in this matrix is the derivative of product $i$'s demand with respect to product $j$'s price. Thus, the matrix specifies the complementarity and substitutability relationships across products. Mathematically, the significant structure property requires that  a normalized version of $\bm{D}$ possesses a low-dimensional subspace of eigenvectors with large eigenvalues. In other words, there exist one or more principal components of the demand system that account for a significant amount of demand behavior. 

We show that significant structure entails an economic property---substantial large-scale complementarities in the economy. These large-scale complementarities give rise to cross-market double marginalization problems, where many goods exert externalities on one another. In turn, these externalities create the potential for small subsidies to have large spillover effects, raising the consumption of many underproduced goods and thereby substantially improving welfare.

Significant structure is also central to realizing this potential. In particular, it permits uncovering the large-scale structure of demand relationships when the Slutsky matrix is observed imperfectly.  The authority's signal consists of noisy estimates of the entries of this matrix.  In principle, this noise can create large uncertainty in the operation of a given intervention. The key step in dealing with this is an application of a theorem of \citet{davis1970rotation}. We use this theorem to show that when the economy satisfies significant structure, the  noisy observation of $\bm{D}$ suffices to precisely predict the effects of \emph{some} well-chosen interventions---specifically, those operating in the space of eigenvectors associated with the largest eigenvalues of $\bm{D}$.

Combining the economic and statistical implications of significant structure allows us to establish our main result. In markets possessing such structure, the authority can recover precise information about large-eigenvalue components of $\bm{D}$, and this information suffices to construct interventions with highly predictable surplus implications.

The robust interventions  constructed in this way keep consumer surplus essentially unchanged while generating total-surplus gains. Though the direct incidence of these gains is on producers, they can be shared with consumers through additional transfer instruments.  \Cref{Prop:LackOfAS} shows that, in some environments with significant structure, the authority \emph{cannot} robustly guarantee any fixed positive fraction of expenditure accruing to consumer surplus. Thus, there are some intrinsic constraints on the direct incidence of robust interventions.

At a technical level, to perform our analysis we develop a new spectral description of the pass-through of an intervention. That is, we diagonalize the Slutsky matrix to obtain a specific orthonormal basis in which we can express the implications of any intervention as a linear combination of orthogonal effects.  These effects correspond to the projection of the intervention onto each eigenvector of a (normalized) Slutsky matrix $\bm{D}$. By characterizing the pass-throughs of subsidies to prices, quantities, and welfare separately across these principal components, we are able to prove that targeting the high-eigenvalue principal components yields precisely predictable results achieving our claimed welfare properties.

\subsection{Structure of the paper}
We introduce the framework in \Cref{sec:model} and, as a benchmark, \Cref{Sec:FullInfo} characterizes the outcome an authority can implement with complete information (\Cref{prop:Welfare CI}).

\Cref{sec:SBM} introduces an important running example of a demand system, consistent with a representative consumer with a quadratic utility function. There are $K$ product categories; products within a category are substitutes, while across categories they are complements. We illustrate some of the main ideas of the paper in the context of this example and present a special case of our main result, \Cref{prop:example}, in this setting: as the number of product categories grows large, the authority can design robust interventions that are welfare-improving. This example brings out the role of large-scale complementarities and a tractable decomposition of effects of taxes/subsidies in the design of robust intervention policies.

\Cref{sec:spt} introduces a new spectral description of the pass-through of an intervention. The spectral decomposition is of independent interest, yielding a useful basis in which price and welfare pass-throughs of cost shocks behave intuitively despite arbitrary spillovers across firms. \Cref{sec:rs} formally defines  the property of significant structure and shows how to check it within the product-category model. 

\Cref{sec:mainresult} presents the main result of the paper, \Cref{Th:Main}. 

\Cref{Sec:Illustration} develops a Monte Carlo experiment to demonstrate how significant structure permits the recovery of welfare-relevant structure from noisy observations of the demand system, and to  illustrate the limitations of this recovery.

\Cref{concluding remarks} further explores the scope of our analysis. \Cref{S:Tight} clarifies the limits of our robust interventions and shows that it may be impossible to robustly improve welfare when significant structure fails (\Cref{Prop:LackOfAS}). \Cref{concluding remarks_sampling} discusses foundations for the property of significant structure and develops an empirical diagnostic for it, which we illustrate using the Monte Carlo experiment. \Cref{sec:hedonic} discusses how the linear--quadratic utility model and the hedonic model of demand relate to our theory. \Cref{concluding remarks_NL} discusses the implications of our assumption of local linear demand and how to generalize our insights once we dispense with that assumption.

\subsection{Related literature}

Our paper contributes to the literature on the structure and theoretical properties of market power. For an early theoretical paper, see \citet{dixit1986comparative}; more recent studies include, for example, \citet{vives1999oligopoly}, \citet{azarvives2021}, \citet{nocke2018Multiproduct}, and \citet{nocke2022Merger}. A recent strand of research in macroeconomics and industrial organization uses differentiated oligopoly network models---similar to the one we consider here---to provide empirical estimates of efficiency losses due to market power  \citep[e.g., ][]{pelligrino2021,edererpelligrino2021}.\footnote{See also \citet{elliott2019role} for related arguments about how network methods can be useful for competition authorities in developing antitrust investigations.}

Given these estimates of inefficiencies, a natural theoretical question is: What feasible interventions can improve welfare? Our main contribution is to analyze interventions from the perspective of an authority uncertain about the demand structure. Our analysis combines new spectral pass-through formulas with results building on the statistical theory of large matrices, and we identify conditions on the demand structure that ensure the robust achievement of welfare improvements even when many aspects of the demand structure cannot be accurately estimated.\footnote{Our focus on pass-through builds on work emphasizing the value of pass-through as a conceptual tool, e.g., \citet{marshall1890}, \citet{Pigou1920}, \citet{dixit1979price} and, more recently, \citet{weyl2013pass}, \citet{miklos2021pass} and \citet{JN2024}.} This approach has significant implications for understanding which kinds of empirical models are needed to design interventions in large markets with many goods. In particular, we note that the significant structure condition entails widespread complementarities in the economy, and this is distinctive relative to the focus on substitutabilities in the models of hedonic utility used in the literature. We elaborate on these issues in \Cref{concluding remarks}.

Methods in high-dimensional statistics are currently attracting considerable interest in econometric settings \citep*{athey2018impact,athey2021matrix, chernozhukov2023inference}, including work applying related statistical models to informational or behavioral spillovers in social networks and marketplaces \citep*{golub2012homophily,bajari2023experimental,cai2022recommender,dasaratha2020distributions, cai2022,parise2023graphon, chandrasekhar2024non,wager2021experimenting}.

However, we know little about when noisy data can be effectively used in order to implement desirable interventions in the presence of strategic spillovers, particularly in market settings. We show that, in a large oligopoly market, spectral methods for high-dimensional statistics developed in the literature on large network recovery can be useful for designing socially desirable interventions.\footnote{See the monograph \citet{chen2021spectral} for a detailed treatment of the mathematical background and some marquee applications.} 

Our paper contributes to the theory of network interventions. Early contributions include \cite{Borgatti2006}, \citet*{Ballesteretal2006}, and \cite{goyal1996interaction}.\footnote{The literature on this subject is very large.  Intervention design has been studied in models of information diffusion, advertising, finance, security, and pricing, among other topics---see, e.g., \citet*{banerjee2013diffusion},  \citet{BlochQuerou2013}, \citet*{Candoganetal2012}, %
\citet{BelhajDeroian}, \citet{GDemange}, \cite{DziubinskiGoyal2017}, \citet{GaleottiGoyal2009}, %
and \citet*{leduc2017pricing}.}  Spectral methods have recently been applied to optimal intervention problems when spillovers are known \citep*{galeotti2020targeting, gaitonde2021polarization,liu2024dynamic}.\footnote{Some recent work uses spectral analysis to derive conditions for core-selecting re-allocative auctions \citep{MarzenaR2023}, and robust implementation \citep{ollar2023network}. See also \citet{aguiar2017slutsky} on spectral methods to study Slutsky matrices in a consumer theory setting.} By contrast, in the present paper, the authority observes strategic spillovers with significant noise.
The methods we develop for robust interventions can be applied to other network games and we briefly discuss this in \Cref{concluding remarks_NG}. Our analysis of perturbations of taxes and subsidies is related to the classic ``tax reform approach'' in public finance \citep{feldstein1976theory,tirole1981tax}; the study of uncertain spillovers distinguishes our work.

Our approach to robustness is conceptually related to, but methodologically distinct from, an extensive literature in economic theory. That literature focuses on understanding the design of mechanisms and contracts that achieve desired outcomes even when assumptions about the environment (e.g., agents' preferences, beliefs, and rationality) are relaxed; see \cite{carroll2019robustness} for a survey. Our definition of robustness aligns with the spirit of this literature. However, in our context, the motivation for analyzing robust interventions arises from the high-dimensional nature of the market state, and we use methods that align with statistical work in this type of setting.

\section{Framework} \label{sec:model}
The foundation of our framework is a differentiated oligopoly game. Within this game, we introduce a statistical framework describing the signals available to the authority about the oligopoly, a class of interventions available to the authority, and a notion of rules that use these signals to achieve good outcomes robustly.

\subsection{Demand side}
 There is a set $N=\{1,\ldots,n\}$ of distinct products, with the number of products $n\geq1$; a typical product is denoted by $i$. The demand for these products arises from the consumption choices of a finite number of optimizing households. Each household $h\in \{1, \dots, H\}$ takes prices as given and has a \emph{choice utility} that is quasilinear in a numeraire $m$,   
$$  {U}^h(\tilde{\bm{q}}^h,m) =  V^h(\tilde{\bm{q}}^h) +  m,$$ 
 where $V^h$ is a twice-differentiable and strictly concave function of the consumption profile $\tilde{\bm{q}}^h \in \mathbb{R}^n$ and $m$ is a numeraire (``money''), in which all prices are denominated.  Given a price profile $\tilde{\bm{p}}$, the household's problem is to choose a bundle $\tilde{\bm{q}}^h$ to maximize $V^h(\tilde{\bm{q}}^h)-\tilde{\bm{p}}\cdot \tilde{\bm{q}}^h$.  Letting $\bm{q}^h(\tilde{\bm{p}})$ be the solution to household $h$'s problem (unique by strict concavity of $V^h$), total market demand is\footnote{Because the households' utilities are quasilinear in money, one can derive the same aggregate demand from a single representative consumer.} 
$$ \bm{q}(\tilde{\bm{p}}) = \sum_{h=1}^H \bm{q}^h(\tilde{\bm{p}}).$$ We use tilde notation for an arbitrary price or quantity, and then drop the tilde for these variables to indicate optimal or equilibrium solutions to be introduced later.

\subsection{Supply side}

There is a firm associated with each product: Firm $i$ produces product $i$. Firms play a simultaneous pricing game; each firm chooses $\tilde{p}_i \geq 0$. For any profile of prices $\tilde{\bm{p}}$, firm $i$'s profit is 
\begin{equation}\label{payoffs} 
  q_i(\tilde{\bm{p}})(\tilde{p}_i - c_i+\sigma_i),
\end{equation}
where $c_i$ is firm $i$'s (constant) marginal cost of production and $\sigma_i$ is a per-unit subsidy on firm $i$. In particular, when $\sigma_i$ is positive, the authority subsidizes firm $i$'s production and transfers $\sigma_i \tilde{q}_i$ to firm $i$, whereas a negative $\sigma_i$ corresponds to a tax.

We call $\bm{\sigma}=(\sigma_1,\ldots,\sigma_n)$ an intervention. We denote by $\bm{p}(\bm{\sigma})$ a pure strategy Nash equilibrium given $\bm{\sigma}$, and we refer to $\bm{p}^0=\bm{p}(\bm{0})$ and $\bm{q}^0=\bm{q}(\bm{p}^0)$ as the \emph{status quo} vectors of equilibrium price and quantity. Throughout, $\|\cdot\|$ denotes the Euclidean norm for vectors and the induced operator norm for matrices. To facilitate unambiguous local comparative statics, we assume that $\bm{p}(\bm{\sigma})$ is unique and continuous in $\bm{\sigma}$ around the status quo $\bm{\sigma}=\bm{0},$ as the following assumption, maintained throughout the analysis, formalizes:

\begin{assumption}[Local equilibrium uniqueness]
\label{ass:unique} There exist $\nu>0$ and $\rho > 0$ such that, for all $\bm{\sigma}$ with $\Vert \bm{\sigma}\Vert < \nu$, there is a  unique pure-strategy Nash equilibrium $\bm{p}(\bm{\sigma})$ in a $\rho$-neighborhood of $\bm{p}^0$.
 \end{assumption}

For a discussion of conditions for existence and uniqueness of Nash equilibrium in oligopoly models, see \citet{fudenberg1991game}, \citet{vives1999oligopoly}, and \citet{cumbul2018multilateral}.\footnote{Assuming compact action sets, existence of a pure strategy equilibrium is guaranteed if profits are quasi-concave in prices. }
 A sufficient condition for local uniqueness is non-singularity of the Jacobian of best responses at equilibrium; this holds generically in our setting.\footnote{For a discussion of these issues, see, e.g., \citet{mclennan2018advanced}.}

 We confine attention to small interventions (i.e., $\Vert \bm{\sigma}\Vert < \nu$) and evaluate the effect of an intervention $\bm{\sigma}$ on any outcome variable $Y$ (e.g., producer or consumer surplus) by its first derivative evaluated at the status quo intervention $\bm{\sigma}=\bm{0}$:\begin{equation}\dot{Y}_{\bm{\sigma}}=\sum_{i}\frac{d Y}{d \sigma_i}(\bm{0})\sigma_i.  \label{eq:dot_definition}\end{equation}  
Throughout, we consider the effects of an intervention starting from the status quo equilibrium, which implies that the expenditure $S = \bm{q}\cdot \bm{\sigma}$ associated with intervention $\bm{\sigma}$ satisfies \begin{equation}\label{eq:dotS}
\dot{S}_{\bm{\sigma}}=\bm{q}^0\cdot \bm{\sigma}. 
 \end{equation}

\subsection{Interventions and their effects}

 We start by focusing on environments in which each good's demand is linear in its own price in a neighborhood of the status quo equilibrium, as captured by the following assumption.   
 
\begin{assumption}[Each good's demand is linear in its own price, locally]\label{A1} 
There exists $\rho > 0$ such that $\frac{\partial q_i(\bm{p})}{ \partial p_i}=\frac{\partial q_i(\bm{p}^0)}{ \partial p_i}$ for all $\bm{p}$ satisfying $\Vert \bm{p}-\bm{p}^0\Vert<\rho$.

\end{assumption}

This assumption facilitates our analysis of interventions, yielding simple formulas for comparative statics. We discuss the content of this assumption and the extension of our methods to the case in which this assumption does not hold in  \Cref{concluding remarks_NL}.

The firms' first-order conditions imply that equilibrium prices $\bm{p}(\bm{\sigma})$ satisfy 

\begin{equation}\label{eq:FOC} 
q_i(\bm{p}(\bm{\sigma}))=-\frac{\partial q_i(\bm{p}(\bm{\sigma}))}{ \partial p_i} ({p}_i(\bm{\sigma})-c_i+\sigma_i) \quad \text{ for each firm $i$.}
\end{equation}
 \Cref{A1} implies that we can replace the $\bm{p}(\bm{\sigma})$-dependent partial derivative in the above equation with the constant $\frac{\partial q_i(\bm{p}^0)}{ \partial p_i}$.  Implicitly differentiating the first-order-condition system with respect to $\bm{\sigma}$ gives the price equation below. The quantity equation follows from applying the chain rule to $\bm{q}(\bm{p}(\bm{\sigma}))$ at $\bm{\sigma}=\bm{0}$:
\begin{equation}\label{eq:PricePass}
[-\text{diag}(\bm{D})-\bm{D}]\dot{\bm{p}}_{\bm{\sigma}}=\text{diag}(\bm{D}) \bm{\sigma} \quad \text{ and } \quad   \dot{\bm{q}}_{\bm{\sigma}} = \bm{D} \dot{\bm{p}}_{\bm{\sigma}},
\end{equation}
 where $\bm{D}=\bm{D}(\bm{p}^0)$ is the Slutsky matrix:\footnote{In general, the  matrix of  derivatives of Marshallian demand need not be the same as the Slutsky matrix (which works with compensated demand).  However, in this demand system, the wealth effect is zero due to the fact that the goods utility and money are additively separable. Thus, the two matrices coincide \citep{nocke2017quasi}, and so we use the term ``Slutsky matrix'' throughout.} $$ D_{ij}= {\frac{\partial q_i(\bm{p}^0) }{\partial p_j}},$$ and $\text{diag}(\bm{D})$ is a diagonal matrix defined by $\text{diag}(\bm{D})_{ii} = D_{ii}$. Note that $D_{ii}<0$ by strict concavity of the consumers' utility functions. For $i\neq j$, if $D_{ij}>0$ (resp. $D_{ij}<0$) then, around the equilibrium, products $i$ and $j$ are substitutes (resp. complements).  

\Cref{A1} implies that comparative statics of prices and quantities are fully determined by the Slutsky matrix $\bm{D}$. We note that $\bm{D}$ satisfies the following property \citep{nocke2017quasi}.
\begin{PropertyA} \label{Ass:psd} The Slutsky matrix $\bm{D}$ is negative semidefinite.
\end{PropertyA} This property holds because the demand function can be taken to arise from a representative household (with a twice-differentiable utility function for goods equal to the sum of the consumers' utilities, $V^h$).

\subsection{Surpluses} \label{sec:Interventions}
The authority cares about the surplus that different market participants obtain in equilibrium. We focus on three canonical metrics: consumer surplus $C$, producer surplus $P$, and total surplus $W=C+P-S$, where $S$ is the authority's expenditure. Let $C^h$ denote the consumer surplus accruing to consumer $h$. Given an intervention $\bm{\sigma}$, equilibrium price $\bm{p}(\bm{\sigma})$ and equilibrium quantity profiles
$\{\bm{q}^h(\bm{p}(\bm{\sigma}))\}_{h=1, \dots, H}$, these are: 
\[
C=\sum_{h} {C}^h \quad \text{where} \quad  {C}^h= V^h(\bm{q}^h(\bm{p}(\bm{\sigma}))) - \bm{q}^h(\bm{p}(\bm{\sigma})) \cdot \bm{p}(\bm{\sigma}),
\]
\begin{equation}
P=(\bm{p}(\bm{\sigma})-\bm{c}+\bm{\sigma}) \cdot \bm{q}(\bm{p}(\bm{\sigma})) %
\label{eq:welfare_defs}
\end{equation}

\subsection{The statistical framework and robust interventions}
 A \emph{market state}, denoted by $\bm{\theta}$, is a pair $(\bm{D},\bm{q}^0)$; the set of possible market states is $\bm{\Theta}$. The authority receives a signal, denoted by $\widehat{\bm{\theta}} \in \widehat{\bm{\Theta}}$, about the market state. In our setting, $\widehat{\bm{\theta}}$ consists of signals about the Slutsky matrix and quantities:  
 $$ \widehat{\bm{D}} = \bm{D} + \bm{E} \qquad \text{and} \qquad \widehat{\bm{q}}^0 = \bm{q}^0 + \bm{\varepsilon}.  $$ {The joint distribution of $(\bm{E},\bm{\varepsilon})$ may depend on $\bm{\theta}$; we denote the resulting signal distribution by $\varphi_{\bm{\theta}} \in \Delta(\widehat{\bm{\Theta}})$.  }We will assume for simplicity throughout that, conditional on the market state $\bm{\theta}$, the quantity errors $\varepsilon_i$ have expectation zero, are i.i.d., and are independent of the matrix observation noise $\bm{E}$.

An \emph{intervention rule} $\bm{R}$ maps each signal $\widehat{\bm{\theta}}=(\widehat{\bm{D}},\widehat{\bm{q}}^0)$ to an intervention $\bm{\sigma}$.\footnote{We require this map to be measurable in a suitable sense, which is clear in our application.}
A \emph{market outcome} is a pair $(\bm{\theta},\bm{\sigma})$  consisting of a market state and an intervention; we refer to this pair as an outcome because it determines production, consumption, and  transfers.  A \emph{property} is a measurable subset  $\mathscr{P}$ of all possible market outcomes. For example, a desirable property is increasing total surplus: $ \mathscr{P} = \{ (\bm{\theta},\bm{\sigma}) : \dot{W}_{\bm{\sigma}} > 0$\}. We are interested in understanding which properties $\mathscr{P}$ can be achieved with high probability in all market states.

\begin{definition} \label{def:robust_intervention}
     An intervention rule $\bm{R}$ achieves a property $\mathscr{P}$ \emph{$\epsilon$--robustly} if, for every $\bm{\theta} \in \bm{\Theta}$, $$ \varphi_{\bm{\theta}} \left(  \widehat{\bm{\theta}} :  (\bm{\theta} , \bm{R}(\widehat{\bm{\theta}})) \in \mathscr{P}  \right) \geq 1-\epsilon. $$
\end{definition}

Note that we require this condition to hold in every state, so the only randomness in this definition is in the signal draw.\footnote{Equivalently, the condition is the requirement that the probability of achieving $\mathscr{P}$ is at least $1-\epsilon$ under every prior over $\bm{\Theta}$.} This is the frequentist approach following \citeauthor*{wald1950statistical}'s (\citeyear{wald1950statistical}) statistical decision theory. A recent line of research in economics works in this paradigm:   welfare or regret guarantees are given uniformly over structural economic parameters $\bm{\theta}$ (over which the analyst does not have a prior); signals conditional on such parameters are analyzed probabilistically, with their distribution being implied by the analyst's statistical procedure.\footnote{Our \Cref{ap:sampling} describes a sampling procedure giving rise to signals about the state in our framework.} For examples of recent work following this paradigm, see \citet{manski2004statistical}, \citet{kitagawa2018should} and \citet{athey2021policy}. It is worth remarking that \Cref{def:robust_intervention} defines success in terms of the probability of a good outcome. In our setting, we will be able to deduce from this statements about expected performance  conditional on $\bm{\theta}$ (which hold uniformly over $\bm{\theta}$), as is typical in the literature we have just mentioned.

In our asymptotic results, we say that an intervention rule achieves a property \emph{robustly} if, for any $\epsilon > 0$, the rule achieves the property $\epsilon$-robustly for all sufficiently large $n$.

\subsection{Asymptotic analysis in large markets}

It will be convenient to have notation for expressing that some quantity asymptotically vanishes relative to others.

An \emph{environment} is a sequence $$\mathcal{E}:=\left(\bm{\Theta}(n),\varphi(n)\right)_n,$$ where $\bm{\Theta}(n)$ is a set of market states with $n$ firms  and $\varphi(n)$ describes the distribution of noise.\footnote{ Given that the distribution of noise can depend on the market state, $\varphi(n)$ is a map $\bm{\theta} \mapsto \varphi_{\bm{\theta}}(n)$ assigning to each state $\bm{\theta}\in \bm{\Theta}(n)$ a signal distribution.} Given $\mathcal{E}$, for a sequence of real-valued random variables $X_n(\bm{\theta})$ and a constant $x$, we say that $X_n \xrightarrow{p} x$ \emph{uniformly over $\bm{\theta}\in\bm{\Theta}(n)$} if for every $\eta>0$,
\[
\sup_{\bm{\theta}\in\bm{\Theta}(n)}\mathbb{P}_{\varphi_{\bm{\theta}}(n)}\big(|X_n(\bm{\theta})-x|>\eta\big)\to_n 0.
\]
We write $X_n=o_{\mathrm{p},\bm{\Theta}}(1)$ if $X_n\xrightarrow{p}0$ uniformly over $\bm{\theta}\in\bm{\Theta}(n)$. More generally, we write $X_n=o_{\mathrm{p},\bm{\Theta}}(f(n))$ if $X_n/f(n)\xrightarrow{p}0$ uniformly over $\bm{\theta}\in\bm{\Theta}(n)$.  Finally, we write $X_n=O_{\mathrm{p},\bm{\Theta}}(1)$ if for every $\epsilon>0$ there exists $M<\infty$ such that $$\sup_{\bm{\theta}\in\bm{\Theta}(n)}\mathbb{P}_{\varphi_{\bm{\theta}}(n)}(|X_n(\bm{\theta})|>M)<\epsilon$$ for all sufficiently large $n$.
In this notation, the roman subscript $\mathrm{p}$ denotes convergence in probability and is not an additional parameter, while $\bm{\Theta}$ records that the convergence is uniform over the sequence of state spaces $(\bm{\Theta}(n))_n$ in the environment.

We will maintain the following technical assumption, which makes asymptotics well-behaved.
\begin{assumption}\label{Ass:Asymptotics} The environment $\mathcal{E}:=\left(\bm{\Theta}(n),\varphi(n)\right)_n$ is such that, for all $(\bm{D},\bm{q}^0) \in \bm{\Theta}(n)$, there exist universal constants $0<d_{\min}\le d_{\max}<\infty$ such that
\[
 -D_{ii} \in [d_{\min},d_{\max}]  \quad \text{for all } i.
\]   
\end{assumption}
\Cref{Ass:Asymptotics} rules out the possibility that own-price demand effects diverge relative to one another as we progress along the sequence defining the environment.

\subsection{Normalization of the Slutsky matrix} \label{sec:norm}
For a given market state $\bm{\theta}=(\bm{D},\bm{q}^{0})$, it will be convenient to describe the market outcomes with respect to a normalized Slutsky matrix, denoted by $\underline{\bm{D}}$. Specifically, the normalized Slutsky matrix is: 
\[
\underline{\bm{D}}=[\text{diag}(-\bm{D})]^{-1/2} \bm{D} [\text{diag}(-\bm{D})]^{-1/2}.
\]
where $\text{diag}(-\bm{D})$ is a diagonal matrix retaining only the diagonal elements of $-\bm{D}$ and setting all off-diagonal elements to zero. This means that 
\[
\underline{D}_{ii}=-1
\quad\text{and}\quad 
\underline{D}_{ij}=\frac{D_{ij}}{\sqrt{D_{ii}D_{jj}}} \quad\text{for all } i\neq j.
\]
This normalization entails a change of units for quantities from $\bm{q}^0$ to
\[
\underline{\bm{q}}^0=[\text{diag}(-\bm{D})]^{-1/2}\bm{q}^0.
\]
Keeping units of money fixed, the corresponding transformations for prices, costs, and interventions are:
\[
\underline{\bm{p}}^0=[\text{diag}(-\bm{D})]^{1/2}\bm{p}^0,
\quad
\underline{\bm{c}}^0=[\text{diag}(-\bm{D})]^{1/2}\bm{c}^0,
\quad
\underline{\bm{\sigma}}=[\text{diag}(-\bm{D})]^{1/2}\bm{\sigma}.
\] 

This normalization amounts to choosing  units for the goods so that  each good's own-price effect is $-1$, making the diagonal entries of $\underline{\bm{D}}$ homogeneous.\footnote{Every underlined variable (e.g., $\underline{\bm{q}}^0,\underline{\bm{p}}^0$) is defined relative to the Slutsky matrix $\bm{D}$, which depends on the state $\bm{\theta}$. A fully explicit notation would subscript the normalization operation (underline) by $\bm{\theta}$, but we omit this and take, in any single expression, the state to be a fixed $\bm{\theta}$.} The normalized coordinates are convenient because, in these, the equilibrium first-order conditions of equation (\ref{eq:FOC}) become:
\begin{equation} \label{eq:simple_demand_identity} \underline{q}_i(\bm{\underline{p}}(\bm{\underline{\sigma}}))=\underline{p}_i(\bm{\underline{\sigma}})-\underline{c}_i^0 + \underline{\sigma}_i, \end{equation} which we exploit throughout.

\section{The case of complete information: Surplus outcomes}\label{Sec:FullInfo}

As a benchmark, we characterize the effects of interventions undertaken by an authority with complete information on the market state. Formally,  $n$ is an arbitrary fixed number and $\bm{\Theta}(n)$ is a singleton throughout this section.

\begin{prop}  \label{prop:Welfare CI} \quad
\begin{itemize}
 \item[1.] For any intervention $\bm{\sigma}$, we have: 
 \begin{equation} \dot{C}_{\bm{\sigma}} +\frac{1}{2}\dot{P}_{\bm{\sigma}} =\dot{S}_{\bm{\sigma}}. \label{eq:Pareto_Identity} \end{equation} 
 Hence, if $\dot{C}_{\bm{\sigma}} \geq 0$, then $\dot{C}_{\bm{\sigma}} + \dot{P}_{\bm{\sigma}}  \leq 2 \dot{S}_{\bm{\sigma}}$.
\item[2.] Assume $\bm{D}$ has at least one nonzero off-diagonal entry. For generic $\bm{q}^0$, any surplus outcome $(\dot{C},\dot{P},\dot{S})$ that satisfies (\ref{eq:Pareto_Identity}) can be implemented by some intervention $\bm{\sigma}$.

 \end{itemize}
 
\end{prop}

It is useful to relate \Cref{prop:Welfare CI} to the standard single-product monopolist facing linear demand with constant marginal cost. In this case, a small per-unit subsidy $\sigma$ lowers the equilibrium price by $\sigma/2$ and so the consumer gains half the subsidy expenditure ($\dot{C} = \frac{1}{2}\dot{S}$). The producer, instead, gains $\dot{S}/2$ from the static incidence on the initial quantity, plus another $\dot{S}/2$ from the profit on the newly induced volume, totaling the full subsidy amount $\dot{S}$.

In the monopoly ($n=1$) setting, these changes satisfy (\ref{eq:Pareto_Identity}). \Cref{prop:Welfare CI} shows that this aggregate constraint carries over to an oligopoly with price competition. 

\begin{proof}[Proof of part (1).]
We fix the state $(\bm{D},\bm{q}^0)$ and present the surplus formula under the normalization introduced in \Cref{sec:norm}.  (Surpluses are invariant to the changes of units performed in the normalization, because the unit of money was not changed.) 
The effect of $\bm{\sigma}$ on consumer surplus satisfies the Marshallian formula 
$\dot{C}_{\bm{\sigma}}=-  \underline{\bm{q}}^0 \cdot \dot{\underline{\bm{p}}}_{\bm{\sigma}}$. 
The effect on producer surplus is 
$$\dot {P}_{\bm{\sigma}}=  \underline{\bm{q}}^0 \cdot (\dot{\underline{\bm{p}}}_{\bm{\sigma}}+\underline{\bm{\sigma}}) + (\underline{\bm{p}}^0-\underline{\bm{c}}) \cdot \dot{\underline{\bm{q}}}_{\bm{\sigma}}=2\underline{\bm{q}}^0\cdot \dot{\underline{\bm{q}}}_{\bm{\sigma}},$$ 
where the second equality follows because \Cref{eq:simple_demand_identity} at $\underline{\bm{\sigma}}=\bm{0}$ yields $\underline{\bm{q}}^0=\underline{\bm{p}}^0-\underline{\bm{c}}^0$. Applying $\left.\frac{d}{dt}\right|_{t=0}$ to both sides of \Cref{eq:simple_demand_identity} along the path $t\underline{\bm{\sigma}}$ gives
$\dot{\underline{\bm{q}}}_{\bm{\sigma}}=\dot{\underline{\bm{p}}}_{\bm{\sigma}}+\underline{\bm{\sigma}}$. Hence, any $(\dot{C}_{\bm{\sigma}},\dot{P}_{\bm{\sigma}}, \dot{S}_{\bm{\sigma}})$ satisfy
\[
\dot{C}_{\bm{\sigma}}+\frac{1}{2}\dot{P}_{\bm{\sigma}}=-\underline{\bm{q}}^0[\dot{\underline{\bm{p}}}_{\bm{\sigma}}-\dot{\underline{\bm{q}}}_{\bm{\sigma}}]={\underline{\bm{q}}}^0 \cdot \underline{\bm{\sigma}}=\dot{S}_{\bm{\sigma}},
\]
where the first equality follows by substituting the expressions for $\dot{C}_{\bm{\sigma}}$ and $\dot{P}_{\bm{\sigma}}$, the second equality follows by using $\dot{\underline{\bm{q}}}_{\bm{\sigma}}=\dot{\underline{\bm{p}}}_{\bm{\sigma}}+\underline{\bm{\sigma}}$, and the third equality by (\ref{eq:dotS}).  

The proof of part (2) uses the techniques we develop in \Cref{sec:spt}, and we relegate it to \Cref{app:complete}. \end{proof}

If products are independent or there is only one product (i.e., $\bm{D}$ is a diagonal matrix), then an intervention implements a specific point of (\ref{eq:Pareto_Identity}): $\dot{C}=\dot{S}/2$ and $\dot{P}=\dot{S}$. Part (2) of \Cref{prop:Welfare CI} shows that, as long as there is demand interdependency across products (i.e., $\bm{D}$ is not a diagonal matrix), the authority can leverage the resulting demand spillovers to generate any point that satisfies (\ref{eq:Pareto_Identity}). This, however, requires precise information about $\bm{D}$. To see this, note that the total surplus pass-through of an intervention $\bm{\sigma}$ is\footnote{To see this, note that $\dot{W}_{\bm{\sigma}} = \dot{C}_{\bm{\sigma}}+\dot{P}_{\bm{\sigma}} -\dot{S}_{\bm{\sigma}}$ and, using (\ref{eq:Pareto_Identity}) and the expressions for $\dot{P}$, we obtain $\dot{W}_{\bm{\sigma}}=   \dot{P}_{\bm{\sigma}}/2=\underline{\bm{q}}^0\cdot(\underline{\bm{D}} \dot{\underline{\bm{p}}}_{\bm{\sigma}})=-\underline{\bm{q}}^0\cdot\underline{\bm{D}}[\bm{I}-\underline{\bm{D}}]^{-1}\underline{\bm{\sigma}}$.}
\begin{equation}  \dot{W}_{\bm{\sigma}} = \underline{\bm{q}}^0 \cdot  \underline{\bm{D}} \dot{\underline{\bm{p}}}_{\bm{\sigma}} =- \underline{\bm{q}}^0 \cdot  \underline{\bm{D}} [\bm{I}-\underline{\bm{D}}]^{-1}\underline{\bm{\sigma}}.\label{eq:Wdot_formula}
\end{equation}
 
 Matrix inverses such as $[\bm{I}-\underline{\bm{D}}]^{-1}$ can be extremely sensitive to the entries of $\bm{D}$. Therefore, without precise knowledge of $\bm{D}$ and $\bm{q}^0$, the authority may not be able to implement a desired point on the Pareto frontier defined by (\ref{eq:Pareto_Identity}). In particular, the authority may not even be sure that a given intervention will increase total surplus ($\dot{W}_{\bm{\sigma}}>0$) rather than decrease it ($\dot{W}_{\bm{\sigma}}<0$). Indeed, it seems hard to justify the detailed study of comparative statics such as (\ref{eq:Wdot_formula}), at least when $n$ is large, without confronting the uncertainty about the ingredients of the formula an analyst or authority is likely to face.

These observations motivate the central question of this paper: \emph{Which interventions have surplus effects that can be predicted with confidence by an authority facing substantial uncertainty about market primitives?} The rest of the paper addresses this question.

The main result of this paper (\Cref{Th:Main} in \Cref{sec:mainresult}) provides a sufficient condition---called ``significant structure''---for the existence of intervention rules that improve surplus robustly, and characterizes these interventions. The next section illustrates these ideas using a simple example.

\section{Illustrative example: A block model of demand} \label{sec:SBM} 
To motivate our main definition (presented in \Cref{sec:rs}) we consider a simple example within a canonical oligopoly model and give a version of our main result in this setting (\Cref{prop:example}).

Suppose there are $K>1$ different product categories, each with a fixed number $m>1$ of products, for a total of $n=m K$ products. Throughout this example, the asymptotic sequence fixes $m$ and lets $K$ grow, so $n=mK$ grows by adding product categories. Products in each category are (imperfect) substitutes, while products across categories are complements. For instance, one product category may consist of different tennis rackets, while another consists of different tennis shoes.

Formally, we take $\mathcal{E}=(\mathbf{\Theta}(n),\varphi(n))_n$ to be an environment defined by a block model of demand. For a given state $\bm{\theta}=(\bm{D},\bm{q}^0)\in \bm{\Theta}(n)$ the demand relationships between products are specified by a $K$-by-$K$ symmetric matrix $\bm{B}(\bm{\theta})$:
$$ D_{ij} = \begin{cases}
-1 & \text{if } i = j, \\
 B_{k(i), k(i)}>0 & \text{if}  \ k(i)= k(j) \ \text{and} \ i\neq j \\ 
  B_{k(i), k(j)}<0 &  \text{if} \  k(i)\neq k(j),
\end{cases}$$
where $k(i)$ is the product category of $i$.\footnote{Note that, in this example, we have made the simplifying assumption that, for each market state $\bm{\theta}=(\bm{D},\bm{q}^0)$, $D_{ii}=-1$ for all $i$ (so that no normalization is required).} We assume without loss of generality (\Cref{rem:normalization} in the appendix) that the quantities satisfy $\Vert\bm{{q}}^0\Vert=1$. \Cref{Figure:Ex-block} depicts a two-product-category demand model.

\begin{figure}
\[
\begin{tikzpicture}[scale=0.72, baseline={(current bounding box.center)}]
    \definecolor{blockblue}{RGB}{80, 130, 200}     %
    \definecolor{blockred}{RGB}{180, 60, 60}       %

    \def\matsize{2.3}
    \def\halfsize{1.15}

    \node at (-2.2,\halfsize) {\large $\bm{D}=$};

    \begin{scope}[shift={(-0.4,0)}]
        \fill[white] (0,0) rectangle (\matsize,\matsize);
        \fill[blockred!75] (0,\halfsize) rectangle (\halfsize,\matsize);
        \fill[blockred!75] (\halfsize,0) rectangle (\matsize,\halfsize);
        \fill[blockblue!70] (0,0) rectangle (\halfsize,\halfsize);
        \fill[blockblue!70] (\halfsize,\halfsize) rectangle (\matsize,\matsize);
        \draw[black!70, line width=0.6pt] (0,0) rectangle (\matsize,\matsize);
        \draw[black!35, dashed, line width=0.5pt] (0,\halfsize) -- (\matsize,\halfsize);
        \draw[black!35, dashed, line width=0.5pt] (\halfsize,0) -- (\halfsize,\matsize);
        \node[above=0.08cm, font=\scriptsize] at (0.575,\matsize) {$A$};
        \node[above=0.08cm, font=\scriptsize] at (1.725,\matsize) {$B$};
        \node[left=0.08cm, font=\scriptsize] at (0,1.725) {$A$};
        \node[left=0.08cm, font=\scriptsize] at (0,0.575) {$B$};
    \end{scope}
\end{tikzpicture}
\]
\caption{Two product-category demand: products within a category are substitutes (red blocks) and across categories products are complements (blue blocks)}\label{Figure:Ex-block}
\end{figure}

\begin{condition}\label{Ass:B-Model} Assume that $\mathcal{E}=(\mathbf{\Theta}(n),\varphi(n))_n$ satisfies the following conditions:
 \begin{enumerate}[(a)]
\item There exists $\kappa<0$ and $\delta>0$ such that, in every market state $\bm{\theta}\in \bm{\Theta}(n)$:
\begin{enumerate}[(1)]
\item There is a minimum level of complementarity across product categories:  $D_{ij}<\kappa$ whenever $k(i)\neq k(j)$;
\item The quantities satisfy $\sum_i {q}_i^0 > \delta \sqrt{n}$.
\end{enumerate}
\item The errors $E_{ij}$ in observing $\bm{D}$ are independent, mean-zero,  and bounded uniformly over $\bm{\Theta}(n)$ and $n$. The quantity $\bm{q}^0$ is observed without noise.
\end{enumerate}
\end{condition}
\Cref{sec:hedonic} shows that \Cref{Ass:B-Model}(a) is consistent with a demand function derived by optimal consumption of a representative consumer with quadratic preferences \citep[see, e.g.,][]{VivesandSing}, and that the pricing game outcome is well behaved when $K$ grows large. \Cref{ap:sampling} shows that \Cref{Ass:B-Model}(b) is consistent with an authority estimating the market state by sampling households making two-good tradeoffs.

Part~(a)(2) of \Cref{Ass:B-Model} is worth a brief remark: it is satisfied by the norm-1 quantity vector given by ${q}_i^0 =\frac{1}{\sqrt{n}}$ for all $i$ (uniform quantities), but not satisfied by a vector that puts quantity $1$ in some fixed $i$ and $0$ elsewhere. It amounts to a condition that quantity is not concentrated too tightly on a small set of indices.

The following proposition is the main result of our paper instantiated in this block model of demand (it follows from \Cref{Th:Main} in \Cref{sec:mainresult}).

\begin{prop}\label{prop:example} 
 Consider the block model of demand and assume \Cref{Ass:B-Model} holds. For every $\varepsilon>0$ and for every target expenditure $s>0$, there is an intervention rule that, if the number of categories $K$ is large enough along this sequence, $\epsilon$-robustly achieves the property that $(\dot{P},\dot{C},\dot{S})$ is within $\varepsilon$ Euclidean distance of   $(2s,0,s)$.
\end{prop}

Recall that \Cref{prop:Welfare CI} established that gross surplus $\dot{C}+\dot{P}$ cannot be increased by more than twice the level of expenditure without reducing consumer surplus. \Cref{prop:example} provides a set of assumptions that are sufficient for the existence of robust intervention rules that achieve this bound. %

These assumptions allow recovery of some partial information about the market state that is sufficient to design intervention rules that increase total surplus robustly without decreasing consumer surplus. In particular, our approach to constructing robust interventions takes the following route:
\begin{enumerate}
    \item Show that, for targeting interventions to achieve good surplus outcomes, it is sufficient to have information on certain eigenvectors of $\underline{\bm{D}}$---those associated with large eigenvalues.
    \item State conditions under which those eigenvectors can be recovered from the noisy signal $\widehat{\bm{D}}=\bm{D}+\bm{E}$.
\end{enumerate} 

The argument for (1) leverages a new spectral decomposition of how interventions pass through to surpluses that we present in \Cref{sec:spt}. \Cref{sec:rs} then formalizes the conditions on market states and noise needed to recover eigenvectors with associated large eigenvalues.

\section{Spectral price theory and significant structure}\label{S:PT}

\subsection{A key tool: Spectral price theory} \label{sec:spt}

Denoting by $\bm{U}$ the matrix whose $\ell^{\text{th}}$ column is the $\ell^{\text{th}}$ eigenvector $\bm{u}^{\ell}$ of $\underline{\bm{D}}$, and by $\bm{\Lambda}$ the matrix whose off-diagonal elements are zero and whose $\ell^{\text{th}}$ diagonal element is the eigenvalue $\lambda_{\ell}$ associated with $\bm{u}^{\ell}$, we have:
$$\underline{\bm{D}}=\bm{U}\bm{\Lambda}\bm{U}^{\tr}.$$

 We can think of the eigenvectors of $\underline{\bm{D}}$ as the principal components of the normalized Slutsky matrix.  An intervention $\bm{\sigma}$ that subsidizes (or taxes) a single product affects the prices and quantities of that product and, through strategic interaction, also those of other products. However, if we think of the eigenvector $\bm{u}^\ell$ as a bundle of products, an intervention $\underline{\bm{\sigma}} \propto \bm{u}^\ell$ only passes through to the price $\bm{u}^\ell \cdot \underline{\bm{p}}$ and quantity $\bm{u}^\ell \cdot \underline{\bm{q}}$ of that bundle; the prices and quantities of all the other bundles $\bm{u}^{\ell'}$ are unchanged. Hence, we can decompose an intervention $\underline{\bm{\sigma}}=\sum_\ell (\bm{u}^{\ell}\cdot \underline{\bm{\sigma}}) \bm{u}^\ell $ into a combination of $n$ orthogonal interventions, each in the direction of an eigenvector. We can use this decomposition to obtain simple expressions for the pass-through of the intervention in terms of the eigenvalues of $\underline{\bm{D}}$ as summarized by the next Lemma.

\begin{lem} \label{prop:eig} The pass-throughs from any intervention $\bm{\sigma}$ to prices and quantities of each eigenvector are as follows:
$$
 \bm{u}^{\ell} \cdot \dot{\underline{\bm{p}}}_{\bm{\sigma}}  =-\frac{1}{1+|\lambda_\ell|} \bm{u}^{\ell} \cdot \underline{\bm{\sigma}} \quad \text{ and } \quad \bm{u}^{\ell} \cdot \dot{\underline{\bm{q}}}_{\bm{\sigma}}= \lambda_\ell  (\bm{u}^\ell \cdot \dot{\underline{\bm{p}}}_{\bm{\sigma}}) =\frac{|\lambda_\ell|}{1+|\lambda_\ell|} \bm{u}^{\ell} \cdot \underline{\bm{\sigma}}.$$
\end{lem}

\begin{proof}
Using the normalization, the equation (\ref{eq:PricePass}) becomes $[\bm{I}-\underline{\bm{D}}]\dot{\underline{\bm{p}}}_{\bm{\sigma}}=-\underline{\bm{\sigma}}$; substituting $\underline{\bm{D}}=\bm{U}\bm{\Lambda}\bm{U}^{\tr}$ we get
$\left(\bm{I}- \bm{U} \bm{\Lambda} \bm{U}^\tr \right) \dot{\underline{\bm{{p}}}}_{\bm{\sigma}}=-\underline{\bm{\sigma}}.$
Multiplying both sides by $\bm{U}^\tr$, we get the vector equation $\bm{U}^{\tr}\dot{\underline{\bm{{p}}}}_{\bm{\sigma}}=-(\bm{I}-\bm{\Lambda})^{-1} \bm{U}^{\tr} \underline{\bm{\sigma}}$ and, using equation (\ref{eq:PricePass}), we also get $\bm{U}^{\tr}\dot{\underline{\bm{q}}}_{\bm{\sigma}}=-\bm{\Lambda}(\bm{I}-\bm{\Lambda})^{-1} \bm{U}^{\tr} \underline{\bm{\sigma}}$. The scalar formulas in \Cref{prop:eig} follow from (i) taking the $\ell$-th coordinate of these vector equations, using that $\bm{U}$ is orthonormal and $(\bm{I}-\bm{\Lambda})^{-1}$ is diagonal, and (ii) recalling from Property NSD that $\underline{\bm{D}}$ is negative semidefinite, so $\lambda_{\ell}\leq 0$ and hence $1-\lambda_{\ell}=1+|\lambda_\ell|$.
\end{proof}

The magnitudes of the price and quantity pass-throughs in each direction $\bm{u}^\ell$ are ordered according to their corresponding eigenvalues: The larger $|\lambda_{\ell}|$ is, the less a subsidy $\bm{u}^{\ell}\cdot \underline{\bm{\sigma}}$ reduces prices, but the more it increases quantities. This asymmetry is the result of two opposing forces: On the one hand, the strategic interactions among firms imply that the equilibrium price $\bm{u}^{\ell}\cdot \underline{\bm{p}}$ is less sensitive to the subsidy $\bm{u}^{\ell}\cdot \underline{\bm{\sigma} }$ the larger is $|\lambda_{\ell}|$. On the other, the demand $\bm{u}^{\ell}\cdot \underline{\bm{q}}$ is more sensitive to the price $\bm{u}^{\ell}\cdot \underline{\bm{p}}$ the larger is $|\lambda_{\ell}|$; this is just a fact about the market's demand function, rather than equilibrium pricing.\footnote{Indeed, it follows from (\ref{eq:PricePass}) that $\bm{U}^{\tr}\underline{\dot{\bm{q}}}=\bm{\Lambda}\bm{U}^{\tr}\underline{\dot{\bm{p}}}$, so the derivative of the demand $\bm{u}^{\ell}\cdot \underline{\bm{q}}$  with respect to $\bm{u}^{\ell}\cdot \underline{\bm{p}}$ is equal to $\lambda_{\ell}$.} \Cref{prop:eig} shows that the second effect dominates the first in the sense that the larger is $|\lambda_{\ell}|$, the more sensitive is the equilibrium quantity $\bm{u}^{\ell}\cdot \underline{\bm{q}}$ to the subsidy $\bm{u}^{\ell}\cdot \underline{\bm{\sigma}}$. 

By combining the spectral decomposition of \Cref{prop:eig} with the surplus formulas derived in \Cref{sec:Interventions} we obtain that the effect of an intervention on consumer, producer, or total  surplus is a weighted sum of pass-throughs to each of the eigenvectors $\bm{u}^\ell$---with the weight being the corresponding bundle's quantity $\bm{u}^{\ell}\cdot \underline{\bm{q}}^0$.
\begin{lem}
The pass-throughs to consumer, producer, and total surpluses are:
$$
\dot{C}_{\bm{\sigma}} =
 -\sum_{\ell=1}^{n} (\bm{u}^{\ell}\cdot \underline{\bm{q}}^0)(\bm{u}^{\ell} \cdot \dot{\underline{\bm{p}}}_{\bm{\sigma}} ) \quad  
 \dot{P}_{\bm{\sigma}} =
 2 \sum_{\ell=1}^{n}(\bm{u}^{\ell}\cdot \underline{\bm{q}}^0)  (\bm{u}^{\ell} \cdot \dot{\underline{\bm{q}}}_{\bm{\sigma}} ) \quad \text{and} \quad \dot{W}_{\bm{\sigma}} =\dot{P}_{\bm{\sigma}}/2.
$$
\label{lem:welfare-pass}
\end{lem}
\begin{proof}
The effect of the intervention on consumer surplus is $\dot{C}_{\bm{\sigma}}=-  \underline{\bm{q}}^0 \cdot \dot{\underline{\bm{p}}}_{\bm{\sigma}}$. Using the fact that the eigenvectors form an orthonormal basis ($\bm{U}\bm{U}^{\tr}=\bm{I}$), we insert the identity matrix into the dot product to obtain  $\dot{C}_{\bm{\sigma}}=-\bm{U}^{\tr}\underline{\bm{q}}^0\cdot \bm{U}^{\tr} \dot{\underline{\bm{p}}}_{\bm{\sigma}}$. Similarly, the effect of the intervention on producer surplus is $\dot {P}_{\bm{\sigma}}=  \underline{\bm{q}}^0 \cdot (\dot{\underline{\bm{p}}_{\bm{\sigma}}}+\underline{\bm{\sigma}}) + (\underline{\bm{p}}^0-\underline{\bm{c}}) \cdot \dot{\underline{\bm{q}}}_{\bm{\sigma}}=2\underline{\bm{q}}^0\cdot \dot{\underline{\bm{q}}}_{\bm{\sigma}}$. Multiplying by $\bm{U}\bm{U}^{\tr}$ yields the expression of $\dot{P}_{\bm{\sigma}}$ in the Lemma. The expression for $\dot{W}_{\bm{\sigma}}$ is obtained by aggregating $\dot{P}_{\bm{\sigma}}$, $\dot{C}_{\bm{\sigma}}$ and the intervention expenditure.   
\end{proof}

We now explain why identifying large-eigenvalue eigenvectors of $\underline{\bm{D}}$ is helpful for the design of interventions that achieve a certain  point on the full-information surplus frontier, i.e., \Cref{eq:Pareto_Identity}. Indeed, suppose the authority knows $\bm{u}^1$ and evaluates the effect of intervention  $\underline{\bm{\sigma}}=\bm{u}^1$.  This intervention costs $\dot{S}_{\bm{\sigma}}=\bm{u}^1\cdot \underline{\bm{q}}^0$ and leads to an overall change in consumer and producer surplus equal to 
\begin{eqnarray*} \label{eq:welfare-pass-2} 
\dot{C}_{\bm{\sigma}}
=\frac{1}{1+|\lambda_1|}\dot{S}_{\bm{\sigma}} \quad \text{and} \quad \dot{P}_{\bm{\sigma}}=
 2  \frac{|\lambda_1|}{1+|\lambda_1|}\dot{S}_{\bm{\sigma}}.
\end{eqnarray*}
Hence, if:
\begin{itemize}
\item[a.] the normalized status quo quantity $\underline{\bm{q}}^0$ is not orthogonal to $\bm{u}^1$, which implies that $\dot{S}_{\bm{\sigma}}\neq 0$, and
\item[b.] the eigenvalue $\lambda_1$ associated with $\bm{u}^1$ is, in absolute value, sufficiently large,
\end{itemize}
then the intervention $\underline{\bm{\sigma}}=\bm{u}^1$ achieves $\dot{C}_{\bm{\sigma}}/\dot{S}_{\bm{\sigma}}\approx 0$ and $\dot{P}_{\bm{\sigma}}/\dot{S}_{\bm{\sigma}} \approx 2$. By \Cref{prop:Welfare CI}, this rate of surplus increase per dollar spent is the highest possible  subject to the constraint that the change in consumer surplus is non-negative. At this point, the sign of $\dot{S}$, $\dot{P}$, $\dot{W}$ is undetermined, because the sign of $\bm{u}^1 \cdot \underline{\bm{q}}^0$ is undetermined.\footnote{Recall that $\bm{u}^1$ and $-\bm{u}^1$ are equally good eigenvectors of $\underline{\bm{D}}$, so by bad luck we may have chosen a $\bm{u}^1$ that makes spending and surplus change negative.} By choosing the sign of the intervention appropriately, we can achieve an intervention that increases total surplus. It is not important that we design the intervention proportional to $\bm{u}^1$; if an intervention is approximately contained in a subspace spanned by eigenvectors that satisfy conditions (a) and (b), then the same reasoning applies to it. 

This discussion suggests that recovering information about large eigenvalues of $\underline{\bm{D}}$ and their eigenvectors is very helpful for designing an intervention. The next section develops a set of conditions on the underlying market state and observation noise, called \textit{significant structure}, that ensures the authority can do this. 

\subsection{Significant structure} \label{sec:rs}

 For matrices, $\|\cdot\|$ is the induced operator norm, which for symmetric matrices is the maximum absolute value of the eigenvalues.
  We also denote by $\mathcal{L}(\underline{\bm{D}},b) \subseteq \mathbb{R}^n$ the space spanned by eigenvectors of $\underline{\bm{D}}$ with eigenvalues at least $b$ in absolute value.

\begin{definition}\label{Def:SS} An environment $\mathcal{E}=\left(\bm{\Theta}(n),\varphi(n)\right)_n$ has significant structure if there is a sequence $b(n) \to \infty$ satisfying the following conditions.

\begin{enumerate}
  \item The quantity error vector $\bm{\varepsilon}$ is comparable to the quantity vector, in the sense that $\Vert \bm{\varepsilon} \Vert  = O_{\mathrm{p},\bm{\Theta}}(\Vert \bm{q}^0 \Vert).$ 
  \item The demand observation noise is negligible relative to $b(n)$ in the sense that $\Vert\bm{E}\Vert=o_{\mathrm{p},\bm{\Theta}}(b(n)) $. 
   \item Some eigenvalues of $\underline{\bm{D}}$ have absolute value at least $b(n)$ and normalized quantities have non-vanishing projection onto the eigenspaces spanned by corresponding eigenvectors: there exists $\delta>0$ such that, uniformly over market states and $n$, $$\big\Vert P_{\mathcal{L}(\underline{\bm{D}},b(n))}\,\underline{\bm{q}}^0\big\Vert\ge \delta \Vert \underline{\bm{q}}^0 \Vert.$$    
   
\end{enumerate}
\end{definition}

Note that in (1) and (2), the asymptotic notation requires that the statement hold uniformly over market states, under the error distribution $\varphi_{\bm{\theta}}(n)$ corresponding to that market state.

Condition (1) says that the quantity signal $\widehat{\bm{q}}^0=\bm{q}^0+\bm{\varepsilon}$ is not swamped by noise in overall magnitude. It permits the application of laws of large numbers to identify weighted moments of $\bm{q}^0$ from $\widehat{\bm{q}}^0$. Intuitively, if quantities are too small relative to the noise, the authority has no way to identify or target the sign of terms in the welfare decomposition, so it is not surprising that such a condition should be involved.

In the remaining conditions, a sequence $b(n)$ defines a diverging  threshold for eigenvalues being considered ``large.'' Condition (2) requires that the demand-signal noise $\bm{E}$ is small in operator norm relative to this threshold, while condition (3) requires that $\underline{\bm{D}}$ has some eigenvalues whose absolute value exceeds it. Together, conditions (2) and (3) ensure that the eigenvectors associated with large eigenvalues can be recovered from the noisy observation $\underline{\widehat{\bm{D}}}$; the key tool for this is the Davis--Kahan theorem, discussed in \Cref{S:DK}.

Condition (3) also requires that $\underline{\bm{q}}^0$ has a non-vanishing projection onto the eigenvectors of large eigenvalues. Together with (1), this ensures that the authority can align interventions  well enough so that, at prevailing quantities, an intervention based on major eigenvectors has non-trivial expenditure. If, instead, the status-quo quantities were asymptotically orthogonal to the recoverable eigenspace,  the identifiable structure of $\underline{\bm{D}}$ would be unrelated to the quantities, leading to negligible and/or very noisy welfare effects from interventions based on this structure.

\subsubsection{Significant structure: Example}\label{sec:illustration3}
We now illustrate the concept of significant structure by presenting an instance of the block model of demand in \Cref{sec:SBM} in which significant structure is present. For simplicity, this instance has the feature that the substitutability within each of the $K$ categories is homogeneous, as is the complementarity among them. That is, for market state $\bm{\theta}\in \bm{\Theta}(n)$,  the Slutsky matrix is 
\[
D_{ij} = 
\begin{cases}
-1 & \text{if } i = j, \\
\alpha_{\bm{\theta}} < 0 & \text{if } k(i) \neq k(j), \\
\omega_{\bm{\theta}} > 0 & \text{otherwise.}
\end{cases}
\]
The associated eigenvalues of $\bm{D}$ are 
$$
\begin{array}{ll}

\lambda_3(\bm{\theta})=-1 - \omega_{\bm{\theta}}   & \text{with multiplicity $K(m-1)$}\\
\lambda_2(\bm{\theta})=\lambda_3(\bm{\theta}) +m (\omega_{\bm{\theta}}-\alpha_{\bm{\theta}})   & \text{with multiplicity $(K-1)$}\\
\lambda_1(\bm{\theta})=\lambda_2(\bm{\theta})+mK\alpha_{\bm{\theta}}   & \text{with multiplicity $1$.}
\end{array}
$$
Property NSD is satisfied if and only if all these eigenvalues are negative, which is the case as long as $\lambda_2(\bm{\theta})<0$ for all $\bm{\theta}$. This imposes an upper bound on the number of products within each category, namely, $$m<m(n):=\inf_{\bm{\theta}\in \bm{\Theta}(n)}\lfloor (1+\omega_{\bm{\theta}})/(\omega_{\bm{\theta}}-\alpha_{\bm{\theta}}) \rfloor.$$

Consider an environment $\mathcal{E} := \{ (\bm{\Theta}(n), \varphi(n)) \}_{n}$ constructed by successively adding product categories of fixed size $m$, with $n=mK$ and fixed $m<m(n)$ along the sequence. Assuming that the environment satisfies \Cref{Ass:B-Model}\footnote{That is, cross-category complementarities are uniformly bounded away from zero, aggregate quantity satisfies $\sum_i {q}_i^0 > \delta \sqrt{n}$, and the observation errors $E_{ij}$ are independent, mean-zero, with variances bounded uniformly over $\bm{\Theta}(n)$ and $n$.}, we demonstrate that it exhibits significant structure.

First, quantities are observed without noise in this example, which takes care of the conditions involving quantity errors.\footnote{Formally, $\bm{\varepsilon}=\bm{0}$, so $\Vert \bm{\varepsilon}\Vert/\Vert \bm{q}^0\Vert=0$, as required in part (1) of \Cref{Def:SS}.} Next, when the number $K$ of categories is large enough, the dominant eigenvalue (in absolute value) is $\lambda_1(\bm{\theta})$, so we take
$$b(n)=\inf_{\bm{\theta}\in\bm{\Theta}(n)}|\lambda_1(\bm{\theta})|=\inf_{\bm{\theta}\in\bm{\Theta}(n)} \bigl(1-\omega_{\bm{\theta}}(m-1)+m(K-1)|\alpha_{\bm{\theta}}|\bigr).$$
This threshold grows linearly in $n$ because cross-category complementarities are bounded away from zero.\footnote{Under \Cref{Ass:B-Model}(a)(1), we have $|\alpha_{\bm{\theta}}|\ge |\kappa|>0$ uniformly, so $b(n)=\Theta(K)=\Theta(n)$ when $n=mK$ grows by increasing $K$ with $m$ fixed.} Moreover, under the observation-error assumptions in \Cref{Ass:B-Model}(b), standard random matrix bounds give $\Vert \bm{E}\Vert = O_{\mathrm{p},\bm{\Theta}}(\sqrt{n})$, and hence $\Vert \bm{E}\Vert/b(n)=o_{\mathrm{p},\bm{\Theta}}(1)$. (See, for example, \citet{vershynin2018high,bandeiravanhandel2016norm}.) Finally, for large $K$, the eigenvector associated with $\lambda_1$ is proportional to $\bm{1}$, and therefore the projection of $\bm{q}^0$ onto $\mathcal{L}(\underline{\bm{D}},b(n))$ is proportional to $\sum_i q_i^0$ and bounded away from $0$.\footnote{When $\bm{u}^1\propto \bm{1}$, we have $|\bm{u}^1\cdot \bm{q}^0|=(\sum_i q_i^0)/\sqrt{n}$; \Cref{Ass:B-Model}(a)(2) imposes $\sum_i q_i^0 > \delta\sqrt{n}$, which yields part (3) of \Cref{Def:SS}.}

Hence, this environment has significant structure. As we will see, significant structure allows us to estimate accurately, when $n$ is sufficiently large, the eigenvector $\bm{u}^1$ associated with the true market state $\bm{\theta}\in \bm{\Theta}(n)$. Using the insights from \Cref{sec:spt}, we can then design an intervention $\bm{\sigma}\propto \bm{u}^1$ and, since the associated $|\lambda_1|$ is very large (and converges to $\infty$ as $n$ grows), this intervention leads to the outcome described in \Cref{prop:example}.

We have shown significant structure in the block model where demand inter-dependencies are homogeneous within a category and across categories. A generalization of these arguments applies to the environment of \Cref{prop:example}. The important properties are simply that the dominant eigenvalue grows, in absolute value, linearly in $n$, whereas the operator norm of the observation noise grows slower, only at rate $\sqrt{n}$. This minimum level of complementarities implies that as $K$ increases, at least one low-dimensional and significant part of the complementarity pattern can be accurately estimated.

Whereas this section has illustrated plausible conditions on interactions for significant structure to be present, practically speaking it is important to explain how it can be verified in applications. We take up such tests in \Cref{concluding remarks_sampling}.

\section{Main result}\label{sec:mainresult}

Our main result is that in an economy with significant structure, the authority can design interventions that substantially increase total surplus. To state the result we introduce one technical assumption, which simplifies the proof considerably.

\begin{assumption}[Recoverable diagonal]
    \label{as:error_assumption}
 In the environment $\mathcal{E}:=\left(\bm{\Theta}(n),\varphi(n)\right)_n$, for all $(\bm{D},\bm{q}^0) \in \bm{\Theta}(n)$ and for all $n$, there exist mappings $\Delta_n : \mathbb{R}^{n\times n} \to \mathbb{R}^n$ (defined for all $n$) such that 
$\max_{i} \big|\Delta_n({\bm{D}}+\bm{E})_i - D_{ii}\big| = o_{\mathrm{p},\bm{\Theta}}(1)$. 
\end{assumption}

This assumes that, for each state $\bm{\theta}$, the diagonal of $\bm{D}$ (the own-price effects) can be estimated consistently from the noisy observation $\widehat{\bm{D}}$. It allows the authority to perform calculations with the normalized version of the noisy observation $\underline{\widehat{\bm{D}}}$. For instance, the assumption holds when $\max_i |E_{ii}| = o_{\mathrm{p},\bm{\Theta}}(1)$. In this case we can take $\Delta_n(\bm{D}+\bm{E})=\diag(\bm{D}+\bm{E})$. For a natural setting in which this assumption holds, see Appendix \ref{ap:sampling}.\footnote{We believe the assumption can be relaxed without changing the possibility of robust interventions, but it does simplify arguments considerably.}

\begin{thm}

\label{Th:Main}
Consider an environment with significant structure satisfying Assumptions 1--4. For every $\epsilon>0$ and for every target expenditure $s > 0$, the following properties can be simultaneously achieved $\epsilon$--robustly for sufficiently large $n$: 
\begin{enumerate}[(i)]
\item The sum of marginal consumer and producer surplus gains $\dot{C}+\dot{P}$ is at least twice the marginal expenditure, up to a small multiplicative error: $\dot{C}+\dot{P} \geq (2 - \epsilon)\dot{S}$.
\item The marginal effect on consumer surplus is vanishing: $|\dot{C}| < \epsilon$. Moreover, no individual consumer's surplus changes significantly---i.e., $|\dot{C}^h| < \epsilon $ for each household $h$.
\item The marginal expenditure is arbitrarily close to the target expenditure $s$, i.e.,  $|\dot{S} - s| < \epsilon $.
\end{enumerate}
\end{thm}
  Part (i) states that the authority can robustly achieve at least two dollars of surplus per dollar spent. Part (ii) states that it is possible to achieve this while leaving consumer surplus essentially unchanged. In particular, under the robust intervention rule we construct,  producer surplus increases by approximately two dollars  per dollar spent, while consumer surplus changes only by a vanishing amount.\footnote{Recalling the definition $\dot{W} = \dot{P} + \dot{C} - \dot{S}$, this implies that every dollar spent  yields approximately one unit increase in net total surplus $\dot{W}$.} Finally, part (iii) says the authority can precisely target the  realized expenditure (and thus the total surplus impact) of the intervention.
  Comparing the theorem with \Cref{prop:Welfare CI}, there is a robust intervention that, subject to a given target expenditure $\dot{S}=s$, achieves approximately the maximum total surplus possible subject to $\dot{C} \geq 0$.  In \Cref{S:Tight} we will see that at least in some cases, this is the only welfare outcome that can be implemented robustly (see part (2) of \Cref{Prop:LackOfAS}).  

We now provide intuition for the properties of the intervention. First recall that, by \Cref{prop:eig}, the pass-through to prices is proportional to $1/(1+|\lambda_\ell|)$, while the pass-through to quantities is proportional to $|\lambda_\ell|/(1+|\lambda_\ell|)$.  Combining this with \Cref{lem:welfare-pass}, each expenditure component $(\bm{u}^{\ell}\cdot\underline{\bm{q}}^0)(\bm{u}^{\ell}\cdot\underline{\bm{\sigma}})$ contributes $1/(1+|\lambda_\ell|)$ to consumer surplus, $2|\lambda_\ell|/(1+|\lambda_\ell|)$ to producer surplus, and $|\lambda_\ell|/(1+|\lambda_\ell|)$ to net total surplus.  When $|\lambda_\ell|\geq M$ for a large $M$, each dollar of intervention yields consumer-surplus incidence $O(1/M)$, while the producer-surplus incidence is $2-O(1/M)$. It is critical to the first part that prices move only slightly, but this small movement creates a quantity effect sufficient to generate the gains in the second part.

The same pass-through reasoning applies to any individual household after replacing $\underline{\bm{q}}^0$ by that household's normalized status-quo quantity vector. Thus, we can upgrade the statement about consumer surplus to a uniform one across households (provided consumers are not compelled to make other transfers); this is especially valuable when targeting consumers  in the product market for personalized transfers is not straightforward. Of course, the authority \emph{may} separately collect producer surplus through fixed transfers and redistribute it to consumers.\footnote{For example, if the authority is a platform, possible mechanisms include commissions or entry fees charged to firms, along with rebates to consumers.} Absent such transfers, the result yields a robust total-surplus intervention with negligible direct consumer-surplus incidence in the strong ``for all households'' sense we have described.

\begin{remark}
Our notion of $\epsilon$--robustness guarantees these surplus properties can also be stated as bounds on \emph{expected} performance conditional on each possible $\bm{\theta}$. The reason is that interventions can be readily adapted so that relevant welfare functionals become uniformly bounded; therefore, as we explain in \Cref{app:expectations}, the convergence-in-probability conclusions in \Cref{Th:Main} upgrade to  expectation statements. 
\end{remark}

\subsection{Sketch of the proof}\label{sec:sketch}
To prove \Cref{Th:Main}, we apply a statistical method that accurately identifies a subspace spanned by top eigenvectors of the Slutsky matrix from noisy observations. We then use the spectral price theory developed in \Cref{sec:spt} to show that interventions projecting exclusively onto this subspace possess the welfare properties stated in \Cref{Th:Main}.

\subsubsection{Recovering the subspace of top eigenvectors: The Davis--Kahan theorem}\label{S:DK}
\begin{figure}
\begin{tikzpicture}[scale=.8]
\tikzset{
  dot/.style={circle, fill, inner sep=1.5pt},
  eigenvalue/.style={dot, label={below:#1}},
  eigenvector/.style={->, thick, blue},
  projection/.style={dotted, thick, red},
}

\draw[{Latex[length=3mm]}-] (-5,0) -- (4.1,0);
\draw[{Latex[length=3mm]}-] (-5,2) -- (4.1,2);

\node[left, above=0.2cm] at (-5,0) {$\underline{\bm{D}}$};
\node[left, above=0.2cm] at (-5,2) {$\widehat{\underline{\bm{D}}}$};

\foreach \x/\l in {-4/$\lambda_1$, -2/$\lambda_2$, 1/$\lambda_3$, 3/$\lambda_4$}
  \node[eigenvalue={\l}] at (\x,0) {};

\node[below] at (4.1,0) {0};
\draw (4.1,-0.1) -- (4.1,0.1); %

\foreach \x/\l in {-1.8/$\widehat{\lambda}_2$, 1.2/$\widehat{\lambda}_3$, 3.2/$\widehat{\lambda}_4$}
  \node[eigenvalue={\l}] at (\x,2) {};

\node[below] at (4.1,2) {0};
\draw (4.1,1.9) -- (4.1,2.1); %

\node[eigenvalue={$\widehat{\lambda}_1$}, red] at (-3.8,2) {};

\draw[eigenvector] (-3.8,2) -- ++(40:1) node[above right] {eigenvector};

\draw[eigenvector] (-4,0) -- ++(20:0.5);
\draw[eigenvector] (-2,0) -- ++(60:0.5);

\draw[projection] (-3.8,2) -- (-4,0);
\draw[projection] (-3.8,2) -- (-2,0);

\draw (-4.8,-1.0) -- (-0.8,-1.0);
\draw (-4.8,-1.1) -- (-4.8,-0.9);
\draw (-0.8,-1.1) -- (-0.8,-0.9);
\node[below=5pt] at (-2.8,-1.0) {$\Vert \bm{E} \Vert$};

\node[anchor=west] at (7,3) {\emph{Legend}:};
\node[dot] at (7.5,2) {};
\node[anchor=west] at (8.1,2) {Eigenvalue};
\draw[eigenvector] (7,1.2) -- (8,1.2);
\node[anchor=west] at (8.1,1.2) {Eigenvector};
\draw[projection] (7,.3) -- (8,.3);
\node[anchor=west] at (8.1,.3) {Projection};
\end{tikzpicture}
\caption{\footnotesize An illustration of the Davis--Kahan theorem: How eigenvectors of $\widehat{\underline{\bm{D}}}$ (perturbed matrix) project onto eigenvectors of $\underline{\bm{D}}$ (true matrix) with similar eigenvalues (off by at most $\Vert \bm{E} \Vert$). This relationship ensures that the subspace generated by the ``top'' (i.e., large-eigenvalue) eigenvectors of $\widehat{\underline{\bm{D}}}$ is a good approximation of the subspace generated by the top eigenvectors of $\underline{\bm{D}}$. For our economic problem this implies that interventions based on top eigenvectors of $\widehat{\underline{\bm{D}}}$ yield high surplus pass-through, despite noise $\bm{E}$.}
\label{fig:DK}
\end{figure}

 The key tool in our statistical exercise that leverages the property of significant structure is the Davis--Kahan theorem.  Under the hypothesis that some eigenvalues of $\underline{\bm{D}}$ are large, standard eigenvalue perturbation bounds guarantee that, despite the noise in $\bm{E}$, the large eigenvalues of the normalized observed matrix are good approximations of the true large eigenvalues of $\underline{\bm{D}}$.  In other words, the noise in $\bm{E}$ cannot cause the large eigenvalues of $\underline{\bm{D}}$ to become ``mixed up'' with the eigenvalues far away in the spectrum; see \Cref{fig:DK} for an illustration.  The Davis--Kahan theorem then permits the recovery of eigenvectors.  More precisely, these perturbation results have the following two central implications in our setting: if $\underline{\bm{D}}$ has large eigenvalues then:  
\begin{enumerate}
    \item[(i)] the normalized matrix $\underline{\widehat{\bm{D}}}$ of $\widehat{\bm{D}}=\bm{D}+\bm{E}$ has some eigenvalues that are themselves large; 
    \item[(ii)] %
     the eigenvectors of $\underline{\widehat{\bm{D}}}$ associated with such eigenvalues can be expressed (up to a small error) as linear combinations of eigenvectors of $\underline{\bm{D}}$ with  large eigenvalues.  \end{enumerate}

\subsubsection{Intervention rules based on recovered eigenvectors} \label{sec:illustration_recovery}

  To facilitate the illustration, we assume that the largest eigenvalue of $\underline{\bm{D}}$ is sufficiently well-separated from all other eigenvalues by a ``gap'' much larger than $\Vert \bm{E} \Vert$.  Under this condition, the Davis--Kahan theorem yields an even stronger implication:  we can use $\underline{\widehat{\bm{D}}}$ to recover a normalized eigenvector  $\widehat{\bm{u}}^1$ that correlates almost perfectly with the corresponding eigenvector $\bm{u}^1$ of $\underline{\bm{D}}$.  This property is, for example, satisfied by the block model of demand in \Cref{sec:illustration3}.  
 
\Cref{prop:eig} and \Cref{lem:welfare-pass} together imply that the effect of intervention $\bm{\sigma}$ is
\begin{equation}
    \dot{W}_{\bm{\sigma}} = \sum_{\ell=1}^n (\bm{u}^\ell \cdot \underline{\bm{q}}^0)(\bm{u}^\ell \cdot \underline{\bm{\sigma}})\frac{|\lambda_\ell|}{1+|\lambda_\ell|} \quad \text{and} \quad \dot{S}_{\bm{\sigma}} = \sum_{\ell=1}^n (\bm{u}^\ell \cdot \underline{\bm{q}}^0)(\bm{u}^\ell \cdot \underline{\bm{\sigma}}). \label{eq:intuition1} \end{equation}

 Let us design an intervention $\underline{\bm{\sigma}} \propto \widehat{\bm{u}}^1$ based on the recovered information and choose the sign of $\underline{\bm{\sigma}}$ to ensure $(\widehat{\bm{u}}^1 \cdot \underline{\bm{q}}^0)(\widehat{\bm{u}}^1 \cdot \underline{\bm{\sigma}})$ is positive.  Since, by the Davis--Kahan theorem, $\widehat{\bm{u}}^1$ is effectively equal to its true counterpart $\bm{u}^1$ with a very small error, from now on we will treat $\bm{u}^1$ as known.  We can see from the equations above that $\dot{W}_{\bm{\sigma}}$ closely approximates $\dot{S}_{\bm{\sigma}}$, since a sufficiently large $|\lambda_1|$ implies $\frac{|\lambda_1|}{1+|\lambda_1|} \approx 1$.  Moreover, if we know $\bm{u}^1 \cdot \underline{\bm{q}}^0$ and this differs from zero (which is a requirement of significant structure in the illustrative case where $\mathcal{L}(\underline{\bm{D}},b(n))$ is one-dimensional), we can scale the intervention to be of the size that we desire, and achieve $\dot{S}_{\bm{\sigma}}=s$.  
 
This argument contains some wishful thinking, however. When we arranged the sign of $\bm{\sigma}$ so that $\dot{S}_{\bm{\sigma}}>0$, we did not consider that we only have a noisy observation $\widehat{\bm{q}}^0$ of $\bm{q}^0$. So part of the challenge of the proof is to manage the observation error that makes $\widehat{\bm{q}}^0$ different from ${\bm{q}}^0$, and to show that we can obtain a correct estimate of the sign with probability arbitrarily close to $1$ as $n\to\infty$. If we fail to do this correctly, our intervention will decrease welfare with positive probability. This explains the necessity of controlling the magnitude of the noise in the quantity signal so that its projections on low-dimensional subspaces are $o_{\mathrm{p},\bm{\Theta}}(1)$. This is ensured by part (1) of \Cref{Def:SS}.

However, bounded noise alone is insufficient. If $\bm{u}^1 \cdot \underline{\bm{q}}^0$ is very small, there may be no hope for consistently recovering the true magnitude or sign of  $\bm{u}^1 \cdot \underline{\bm{q}}^0$ from the signal $\widehat{\bm{u}}^1 \cdot \underline{\widehat{\bm{q}}}^0$. That is, even very small noise could still overwhelm a similarly small underlying mean projection $\bm{u}^1 \cdot \underline{\bm{q}}^0$. This unrecoverability would make it impossible to orient and scale our intervention appropriately, because the orientation procedure relies on such projections. Part (3) of the definition of significant structure prevents this problem by requiring that the projection of $\underline{\bm{q}}^0$ onto eigenvectors with large eigenvalues is bounded away from zero.

This discussion illustrates the following special case of our main result: If the largest eigenvalue of $\underline{\bm{D}}$ is well-separated from others and if $\underline{\bm{q}}^0 \cdot \bm{u}^1$ is not vanishingly small, then a subsidy profile proportional to $\bm{u}^1$ can, if it is suitably scaled, achieve all the properties of \Cref{Th:Main}.

\subsubsection{Intervention rules based on recovered eigenvectors: General case}
The proof of the main result improves on this sketch in two ways. First, it does not rely only on the eigenspace spanned by $\bm{u}^1$. Instead, it uses a potentially much larger eigenspace of $\underline{\widehat{\bm{D}}}$. The general intervention projects $\widehat{\underline{\bm{q}}}^0$ onto $\mathcal{L}(\underline{\widehat{\bm{D}}},b(n))$, the eigenspace of all eigenvectors of the observed $\underline{\widehat{\bm{D}}}$ with eigenvalues whose absolute value is at least $b(n)$. This makes it  easier for the analog of $\widehat{\bm{u}}^1 \cdot \underline{\bm{q}}^0$ not to be too small, since the projection of $\underline{\bm{q}}^0$ onto a larger eigenspace will have a larger norm. Second, the general proof dispenses with assuming that any eigenvalues are well-separated. Instead, it handles any possible spectrum of $\underline{\bm{D}}$ subject to our maintained assumptions. This introduces considerable complexity, as it is no longer generally possible to recover any true eigenvector $\bm{u}^1$ with any accuracy. We instead work directly with a recovered eigenspace that generalizes the span of $\widehat{\bm{u}}^1$. We show that despite limited knowledge of  individual true eigenvectors of $\underline{\bm{D}}$ underlying this space, we can use the fact that all of them have large eigenvalues to generalize our argument for showing that surplus expressions in (\ref{eq:intuition1}) can be made very close and nonzero with a feasible intervention. This is where the arguments go beyond standard applications of the Davis--Kahan theorem.

\section{Illustrating recoverability in a two-block example}\label{Sec:Illustration}

We develop a Monte Carlo experiment designed to demonstrate the central mechanism behind  \Cref{Th:Main}: significant structure enables recovery of welfare-relevant information from noisy observations of the demand system.
Technical details supporting the simulation are in \Cref{app:hadamard_sim}.

\subsection{Demand system and noise in the experiment}
This example uses a block model in which the classification of goods into categories is unknown. The authority knows that there are two categories, that each category includes half of the products, and that within-category products are complements, while across-category products are substitutes. However, the authority does not know which products belong to the same category.

A concrete interpretation is a pair of competing product ecosystems: within an ecosystem, products such as a console, controllers, games, and subscription services are consumed together, while lower prices in one ecosystem draw demand away from the rival ecosystem. The two-block specification captures this reduced-form pattern by treating goods within an ecosystem as complements and goods across ecosystems as substitutes.

The sign pattern of substitutes and complements is different from the example presented in \Cref{sec:SBM}, where we had substitution within small categories and complementarity across categories. The configuration we use in this section makes it possible for the significant structure to be encoded in a rank-one matrix, and this makes the analysis more transparent.\footnote{When we have within-category substitutes, this is not possible because the sign pattern of a matrix $-\lambda \bm{u} \bm{u}^\tr$ necessarily has negative entries (complements) on the diagonal blocks and positive entries (substitutes) on the off-diagonal blocks.} We introduce a tunable parameter $b$ governing the intensity with which this component is represented (and this parameter corresponds to the largest eigenvalue of the Slutsky matrix). \Cref{Th:Main} implies that when this eigenvalue is large enough, the authority can robustly increase total market surplus at the best rate possible. Our simulations will demonstrate this. They will also show that when this eigenvalue is small, the interventions presented here fail to achieve this goal, which provides intuition for the negative results.

\medskip

\paragraph{Constructing a demand system.}  For each market state $\bm{\theta}$ we construct a normalized Slutsky matrix $\underline{\bm{D}}$ which has a block structure corresponding to this partition:

\[
\begin{tikzpicture}[scale=0.72, baseline={(current bounding box.center)}]
    \definecolor{blockblue}{RGB}{80, 130, 200}     %
    \definecolor{blockred}{RGB}{180, 60, 60}       %

    \def\matsize{2.3}
    \def\halfsize{1.15}

    \begin{scope}[shift={(-5.0,0)}]
        \fill[blockblue!55] (0,\halfsize) rectangle (\halfsize,\matsize);
        \fill[blockblue!55] (\halfsize,0) rectangle (\matsize,\halfsize);
        \fill[blockred!50] (0,0) rectangle (\halfsize,\halfsize);
        \fill[blockred!50] (\halfsize,\halfsize) rectangle (\matsize,\matsize);
        \draw[black!70, line width=0.6pt] (0,0) rectangle (\matsize,\matsize);
        \node[below=0.25cm] at (\halfsize,0) {$\underline{\bm{D}}$};
    \end{scope}

    \node at (-2.2,\halfsize) {\large $=$};

    \begin{scope}[shift={(-0.4,0)}]
        \fill[white] (0,0) rectangle (\matsize,\matsize);
        \fill[blockblue!75] (0,\halfsize) rectangle (\halfsize,\matsize);
        \fill[blockblue!75] (\halfsize,0) rectangle (\matsize,\halfsize);
        \fill[blockred!70] (0,0) rectangle (\halfsize,\halfsize);
        \fill[blockred!70] (\halfsize,\halfsize) rectangle (\matsize,\matsize);
        \draw[black!70, line width=0.6pt] (0,0) rectangle (\matsize,\matsize);
        \draw[black!35, dashed, line width=0.5pt] (0,\halfsize) -- (\matsize,\halfsize);
        \draw[black!35, dashed, line width=0.5pt] (\halfsize,0) -- (\halfsize,\matsize);
        \node[above=0.08cm, font=\scriptsize] at (0.575,\matsize) {$A$};
        \node[above=0.08cm, font=\scriptsize] at (1.725,\matsize) {$B$};
        \node[left=0.08cm, font=\scriptsize] at (0,1.725) {$A$};
        \node[left=0.08cm, font=\scriptsize] at (0,0.575) {$B$};
       \draw[decorate, decoration={brace, amplitude=4pt, mirror}] (-.5,-0.15) -- (2.5,-0.15);
        \node[below=0.25cm, font=\small] at (\halfsize,0) {$\lambda_1 \bm{u}^1 (\bm{u}^1)^\top$};
        \node[below=0.7cm, font=\small] at (\halfsize,0) {dominant direction};
    \end{scope}

    \node at (2.45,\halfsize) {\large $+$};

    \begin{scope}[shift={(5.3,0)}]
        \node at (0,\halfsize) {$\displaystyle\sum_{\ell=2}^{n-1} \lambda_\ell \bm{u}^\ell (\bm{u}^\ell)^\top$};
        \draw[decorate, decoration={brace, amplitude=4pt, mirror}] (-2.0,-0.15) -- (2.0,-0.15);
        \node[below, font=\footnotesize, align=center] at (0,-0.55) {contributions from};
        \node[below, font=\footnotesize, align=center] at (0,-1.05) {bulk directions};
    \end{scope}
    \node at (7.7,\halfsize) {\large $+$};
      \begin{scope}[shift={(10,0)}]
        \node at (0,\halfsize) {$\displaystyle \lambda_n \bm{u}^n (\bm{u}^n)^\top$};
        \draw[decorate, decoration={brace, amplitude=4pt, mirror}] (-2.0,-0.15) -- (2.0,-0.15);
        \node[below, font=\footnotesize, align=center] at (0,-0.55) {contributions from};
        \node[below, font=\footnotesize, align=center] at (0,-1.05) { aggregate direction};
    \end{scope}
\end{tikzpicture}
\]
where:\medskip

\medskip

\paragraph{Dominant direction (large eigenvalue).} A market state $\bm{\theta}$ determines a partition of the products into two sets, $A$ and $B$, of the same size. Given $\bm{\theta}$, the $A$-entries of $\bm{u}^1$ are $1/\sqrt{n}$ and the $B$-entries are $-1/\sqrt{n}$. We associate with $\bm{u}^1$ an eigenvalue of $\lambda_1 = -(1 + b)$, where $b \geq 0$ is a parameter we vary. The ``main'' part of $\underline{\bm{D}}$ is a rank-one component $\lambda_1 \bm{u}^1 (\bm{u}^1)^\top$ that creates a block structure in the Slutsky matrix, with a product's block membership determined by the sign of the corresponding entry of $\bm{u}^1$. For goods $i$ and $j$ in the \emph{same} group (same sign in $\bm{u}^1$), we have $u^1_i u^1_j > 0$, so $\lambda_1 u^1_i u^1_j < 0$ (since $\lambda_1 < 0$). This means same-group goods are {complements}. For goods in {different} groups (opposite signs), $u^1_i u^1_j < 0$, so $\lambda_1 u^1_i u^1_j > 0$---different-group goods are {substitutes}.\medskip

\medskip

\paragraph{Aggregate direction (zero eigenvalue).} The eigenvector $\bm{u}^n = \frac{1}{\sqrt{n}}\bm{1}$ represents uniform changes across all goods. We assign it eigenvalue $\lambda_n = 0$.\medskip

\paragraph{Bulk directions.} We now turn to most of the eigenvectors---the so-called bulk. These $n-2$ eigenvectors receive eigenvalues chosen so that they are close to $-1$ and $\sum_\ell \lambda_\ell = -n$. Specifically, $\lambda_\ell = -1 + \frac{b - 1}{n-2}$ for $\ell\notin \{1,n\}$. This ensures that the diagonal $\underline{D}_{ii}=-1$ is constant, as required of the normalized Slutsky matrix.

\medskip
\paragraph{Status quo quantities.} 
The normalized status-quo quantity vector is $\underline{\bm{q}}^0 = \bm{u}^n + 0.05\,\bm{u}^1$. Hence, $\underline{\bm{q}}^0$ projects primarily onto the aggregate direction $\bm{u}^n$ with a smaller projection onto the dominant direction $\bm{u}^1$. %

\medskip
\paragraph{Observation noise and the recovery threshold.} In this simulation, we assume the authority observes $\underline{\bm{q}}^0$ without error, and observes $\widehat{\underline{\bm{D}}} = \underline{\bm{D}} + \bm{E}$, where $\bm{E}$ is a symmetric random matrix with $E_{ij} \sim \mathcal{N}(0,1)$ i.i.d.\ for $i < j$, $E_{ij} = E_{ji}$, and $E_{ii} = 0$, in line with \Cref{as:error_assumption}. This is a Wigner matrix (with zeros on the diagonal), and classical results give $\Vert \bm{E} \Vert \approx 2\sqrt{n}$ with high probability. Thus, the uncertainty about $\bm{\theta}$ is which dominant direction $\bm{u}^1$ (equivalently, which A/B partition) obtains. Each Monte Carlo repetition corresponds to a fresh draw of this signal holding $\bm{\theta}$ fixed.

 The significant structure condition (\Cref{Def:SS}) requires $|\lambda_1| \gg \Vert \bm{E} \Vert$, which is equivalent to $b \gg \sqrt{n}$.  In this simulation, the recovery pattern changes sharply when $b/\sqrt{n}$ is near $1$: below that range, the dominant eigenvalue is comparable to the noise and the leading eigenvector of $\widehat{\underline{\bm{D}}}$ is largely noise-driven, while for $b/\sqrt{n} \gg 1$ the leading eigenvector $\widehat{\bm{u}}$ aligns closely with the true dominant direction $\bm{u}^1$.  The Davis--Kahan theorem is the general version of this last part and is used in the main positive result of the paper: recovery improves once the spike is sufficiently large relative to $\Vert \bm{E} \Vert$.

\subsection{The intervention rule} \label{sec:rule}
Let $\widehat{\bm{u}}$ be the eigenvector of $\widehat{\underline{\bm{D}}}$ with the largest eigenvalue in absolute value. The authority implements the intervention
\[
\underline{\bm{\sigma}} = \frac{s \cdot \widehat{\bm{u}}}{\widehat{\bm{u}} \cdot \underline{\widehat{\bm{q}}}^0},
\]
noting that this guarantees spending $\dot{S} = \underline{\bm{q}}^0 \cdot \underline{\bm{\sigma}}$ is positive; assume the authority scales the intervention to achieve target spending $s = 1$.

\subsection{Results}

We fix $n=256$ and we vary the size of the dominant eigenvalue $\lambda_1$ by changing $b$. At each value of $b/\sqrt{n}$ we run 500 Monte Carlo draws of the noise matrix $\bm{E}$ and implement the intervention $\underline{\bm{\sigma}}$. For each draw we compute total surplus outcomes and the alignment of the dominant eigenvector $\widehat{\bm{u}}$ with the dominant eigenvector $\bm{u}^1$ measured by $|\bm{u}^1\cdot\widehat{\bm{u}}|^2$. We report median outcomes and confidence intervals. \Cref{fig:pareto} summarizes the results.

\begin{figure}[htbp]
    \centering
    \begin{subfigure}[a]{\textwidth}
        \centering
     \includegraphics[width=.92\textwidth]{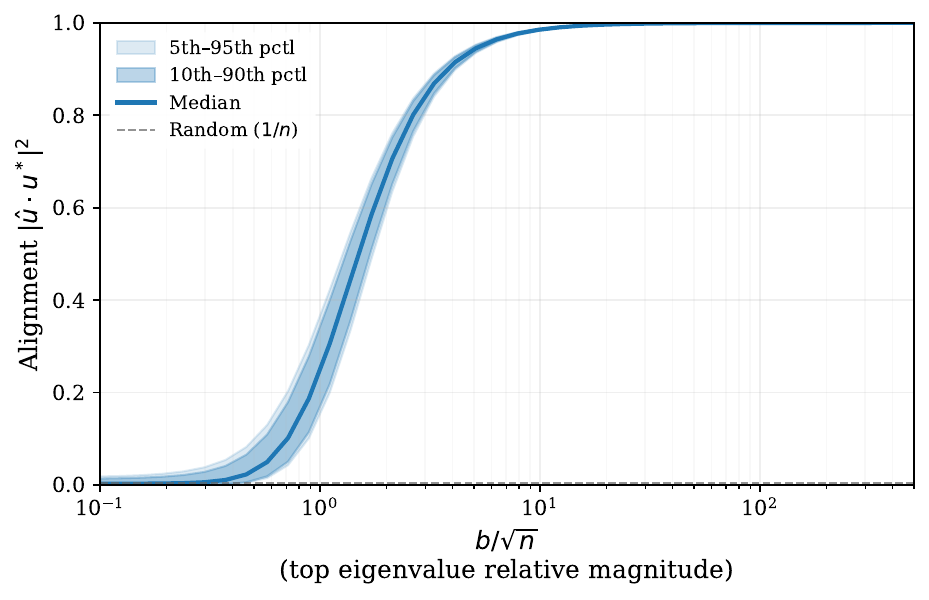} \quad \quad \quad
        \caption{Eigenvector recovery}

        \label{fig:pareto_recovery}
    \end{subfigure}

    \vspace{1em}

    \begin{subfigure}[b]{\textwidth}
        \centering
         \includegraphics[width=1.05\textwidth]{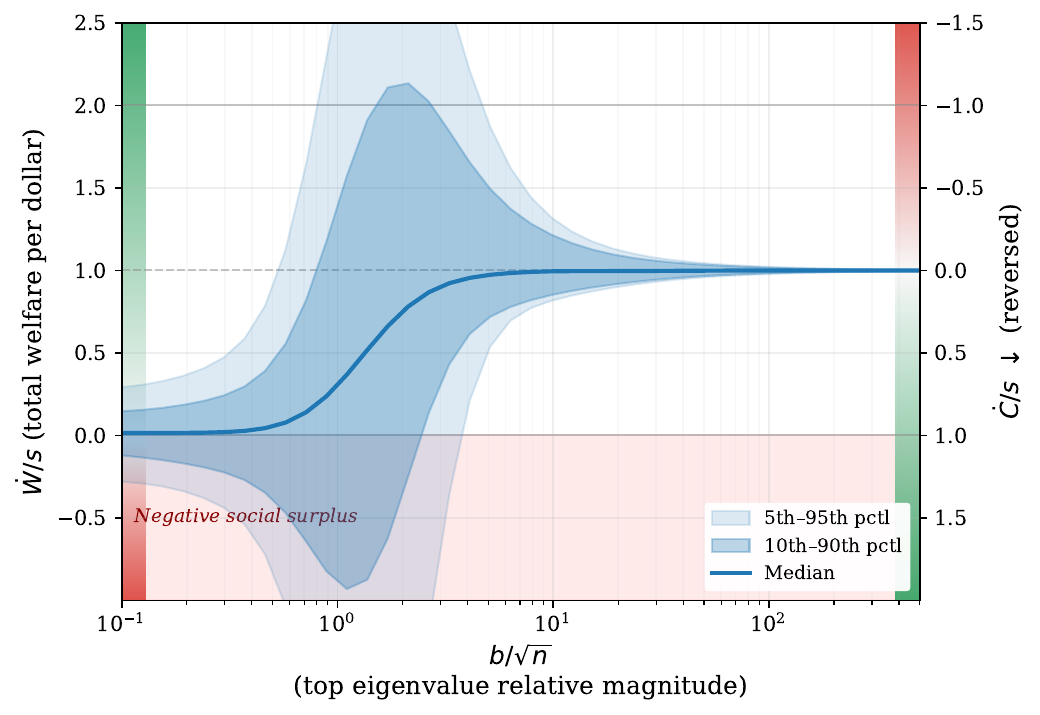}
        \caption{Total Surplus outcomes}
        \label{fig:pareto_welfare}
    \end{subfigure}
    \caption{\textbf{Recovery of $\underline{\bm{D}}$ and welfare outcomes} ($n = 256$ and $500$ draws per value of $b$).  \textbf{(a)}~Alignment $|\widehat{\bm{u}} \cdot \bm{u}^1|^2$ between recovered and true spike eigenvectors. Shaded bands: 5th--95th and 10th--90th percentiles; line: median.
    \textbf{(b)}~Total surplus per dollar $\dot{W}/\dot{S}$ (left axis) and consumer surplus per dollar $\dot{C}/\dot{S}$ (reversed right axis). The shaded region marks negative social surplus.
    }
    \label{fig:pareto}
\end{figure}

\medskip
 \paragraph{Panel (A): Eigenvector recovery.} This panel illustrates the sharp recovery threshold in the present  simulation.  Below the critical signal-to-noise ratio threshold of $1$ ($b/\sqrt{n} < 1$), the median alignment of $\widehat{\bm{u}}$ and $\bm{u}^1$ is of order $1/n$, which is the value expected for two random unit vectors.  Above the threshold, alignment rises rapidly toward the case of perfect alignment of $1$, with narrowing bands indicating increasingly reliable recovery.

\medskip
\paragraph{Panel (B): Welfare outcomes.} Consider the spectral decomposition of total surplus:
\[
\dot{W} =(\bm{u}^1 \cdot \underline{\bm{q}}^0)(\bm{u}^1 \cdot \underline{\bm{\sigma}}) \frac{|\lambda_1|}{1+|\lambda_1|}+\sum_{\ell=2}^{n-1} (\bm{u}^\ell \cdot \underline{\bm{q}}^0)(\bm{u}^\ell \cdot \underline{\bm{\sigma}}) \frac{|\lambda_\ell|}{1+|\lambda_\ell|}+(\bm{u}^n \cdot \underline{\bm{q}}^0)(\bm{u}^n \cdot \underline{\bm{\sigma}}) \frac{|\lambda_n|}{1+|\lambda_n|}.
\]
Below the critical threshold, median total surplus per dollar $\dot{W}/\dot{S}$ is close to zero. The wide bands indicate high variability: some noise realizations happen to align with $\bm{u}^1$ by chance, yielding positive welfare, while others yield negative total surplus---outcomes in the shaded region represent net welfare destruction.

The reason is that lack of recoverability of $\bm{u}^1$ implies that  $\underline{\bm{\sigma}}$ is based on a random direction and, therefore, the projection $\bm{u}^\ell \cdot \underline{\bm{\sigma}}$ is roughly equal across all $\ell$. But the quantity vector $\underline{\bm{q}}^0$ loads almost entirely on $\bm{u}^n$, and $\bm{u}^n$ has pass-through weight $|\lambda_n|/(1+|\lambda_n|)=0$ in our construction. Thus, failed recovery systematically traps spending in a low-pass-through direction and, so, expected total surplus approaches zero.

Above the threshold, $\dot{W}/\dot{S}$  rises quickly toward the theoretical maximum and the bands narrow dramatically. The intervention reliably delivers close to the optimal total surplus per dollar. As $b$ increases, $\dot{W}/\dot{S}$ approaches~$1$: nearly one dollar of net total surplus per dollar of spending. This is equivalent to a $2$-for-$1$ gain in gross total surplus $\dot{C}+\dot{P}$. This is why the significant structure condition matters: it ensures the authority can identify and target the high-pass-through directions that yield substantial welfare gains.

Finally, panel~(B) also shows consumer surplus per dollar $\dot{C}/\dot{S}$ (right axis, reversed). Above threshold, $\dot{C}/\dot{S}$ approaches zero while $\dot{W}/\dot{S}$ approaches~$1$. This means nearly all the welfare gain accrues to producers ($\dot{P}/\dot{S} \approx 2$), confirming the prediction of \Cref{Th:Main}: interventions targeting high-pass-through directions achieve maximum total surplus but allocate almost all the surplus increase to the supply side.

\subsection{What does successful recovery mean economically?}
\Cref{fig:ground_truth} shows the ground-truth block structure $\lambda_1 \bm{u}^1 (\bm{u}^1)^\top$. \Cref{fig:blocks} shows the authority's recovered rank-one component $\widehat{\lambda}\widehat{\bm{u}}(\widehat{\bm{u}})^\top$ of $\widehat{\underline{\bm{D}}}$ for three signal-to-noise ratios. We use the same noise realization $\bm{E}$ across panels, varying only the value of the dominant eigenvalue $\lambda_1$.

\begin{figure}[htbp]
    \centering
    \includegraphics[width=0.5\textwidth]{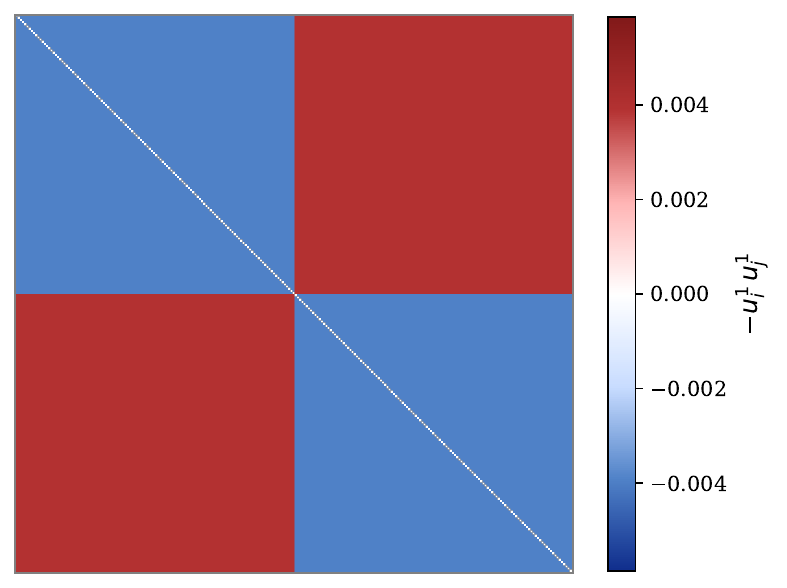}
    \caption{\textbf{Ground-truth block structure} ($n = 256$). The rank-one component $\lambda_1 \bm{u}^1 (\bm{u}^1)^\top$ corresponding to the dominant eigenvalue of the true normalized Slutsky matrix $\underline{\bm{D}}$, with goods ordered by sign in $\bm{u}^1$. Blue: complements (same group); red: substitutes (different groups).
    }
    \label{fig:ground_truth}
\end{figure}

\begin{figure}[htbp]
    \centering
    \includegraphics[width=\textwidth]{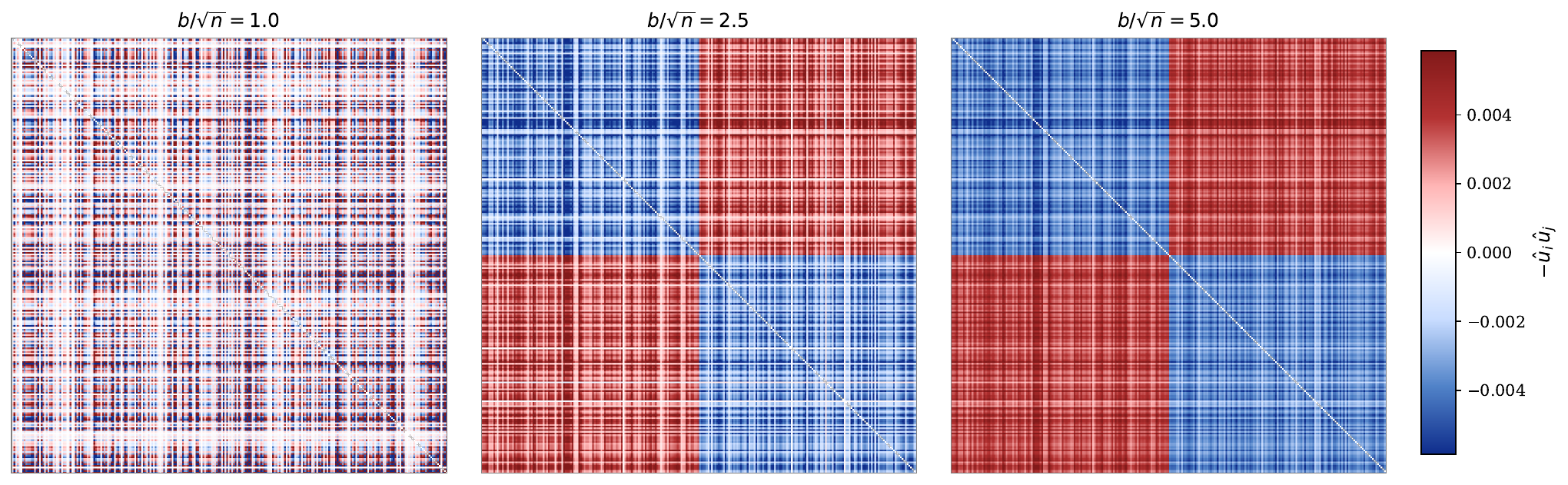}
    \caption{\textbf{Recovery of substitute/complement structure} ($n = 256$, same noise realization). Leading rank-one component $\widehat{\lambda} \widehat{\bm{u}} \widehat{\bm{u}}^\top$ of $\widehat{\underline{\bm{D}}}$, with goods ordered by true $\bm{u}^1$ signs. \textbf{Left} ($b/\sqrt{n} = 1$): Pure noise, no discernible structure. \textbf{Center} ($b/\sqrt{n} = 2.5$): Block structure emerges but with significant noise. \textbf{Right} ($b/\sqrt{n} = 5$): Near-perfect recovery of the two-block pattern.}
    \label{fig:blocks}
\end{figure}

At threshold ($b/\sqrt{n} = 1$), the estimated structure is weakly related to the true structure---the authority cannot tell which goods are substitutes versus complements. Above threshold, the two-block pattern emerges with increasing clarity. At $b/\sqrt{n} = 5$, the recovered structure closely matches the ground truth.

This visualization makes concrete what ``significant structure'' means for intervention design. The authority does not need to estimate every entry of $\underline{\bm{D}}$ accurately. Instead, it needs to recover the \emph{global} pattern---which families of goods interact as substitutes versus complements---that determines where welfare-improving interventions are possible. When this pattern is sufficiently pronounced relative to noise (significant structure), eigenvector-based methods extract it reliably.

\section{Discussion and concluding remarks}\label{concluding remarks}

We conclude with some observations about the scope of our analysis. \Cref{S:Tight} shows the limits of robust interventions. \Cref{concluding remarks_sampling} discusses the theoretical foundation of significant structure and develops an empirical test for it. \Cref{sec:hedonic} formally maps our leading example, the block model of demand in \Cref{sec:illustration3}, to a representative consumer with quadratic utilities. \Cref{concluding remarks_NL} discusses local linear demand. \Cref{concluding remarks_NG} relates our work to the literature on network games.

\subsection{The limits of robust interventions}\label{S:Tight}

We have established that if demand has significant structure, the authority can robustly achieve the maximum possible total surplus per dollar spent subject to the constraint that consumers are not harmed. The associated intervention rule has a precise distributional property as it boosts production with minimal price changes, resulting in firms capturing all efficiency gains.

The key condition of significant structure ensures sufficiently strong signal-to-noise ratio in the signal about the Slutsky matrix. This raises two questions. First, when this property is violated, so that the signal is weak relative to observation noise, can we find interventions that robustly increase total surplus? Second, under the assumption of significant structure, are the distributional properties of the robust interventions in our main result necessary for any robust intervention?

\begin{prop}[Limits of robust interventions]\label{Prop:LackOfAS} Consider interventions with a positive expenditure by the authority, $\dot{S}>0$.
  
$\;$ \\
    \begin{enumerate}
\item There exist environments satisfying Assumptions 1--4, but without significant structure, such that for any $\varepsilon \in (0,1/2)$ there is no intervention rule that $\varepsilon$-robustly achieves $\dot{C} + \dot{P}\geq (1+\varepsilon)\dot{S}$.
\item There exist environments satisfying Assumptions 1--4, with significant structure, such that, for any $\kappa>0$, no intervention rule robustly achieves $\dot{C} \geq \kappa \dot{S}$.\end{enumerate}
\end{prop}

Whereas with significant structure we could construct an intervention returning almost two dollars in market surplus for each dollar invested, part (1) says that without significant structure, we may not be able to achieve any non-negligible return on the authority's spending. And part (2) says that guaranteeing any fixed positive share of spending accrues to consumer surplus  is impossible even with significant structure.

Both parts are proved in \Cref{app:lackofas_proof} using Hadamard-basis constructions closely related to the two-block example in \Cref{Sec:Illustration}. In part~(1), the market state is defined by a two-block classification of goods: a partition of the $n$ products into two equal-sized subsets determining which pairs behave as complements versus substitutes. We use a special type of basis (a Hadamard basis) to index a large collection of such balanced classifications. Each basis vector has entries $\pm 1/\sqrt{n}$, with the signs corresponding to the assignment of goods to blocks.\footnote{The special feature of the Hadamard basis is that there are $n$ such vectors that are all orthogonal to one another.}

For part~(1), we construct $\bm{D}$ so that the vector corresponding to the block partition that is realized has a large eigenvalue; however, all realizations of $\bm{D}$ are designed to generate the same signal. Any deterministic rule must therefore select the same intervention across these states. We prove that for any such intervention, there exists a state in which it is nearly orthogonal to the welfare-relevant spike direction, so $\dot{W}$ is necessarily small. This is akin to the phenomenon illustrated in \Cref{Sec:Illustration} below the recovery threshold, where the leading eigendirection of $\widehat{\underline{\bm{D}}}$ is dominated by noise and eigenvector-based interventions behave like random directions (see \Cref{fig:pareto_recovery} and \Cref{fig:pareto_welfare}  for $b/\sqrt{n}<1$). The proof here shows that, not only does the specific type of intervention studied there fail, but any possible intervention fails. In our terminology, this is an environment without significant structure: although the true demand system features a welfare-relevant two-block pattern, the signal is too weak relative to observation noise to recover which pattern obtains.\footnote{Formally, the demand noise is of the same order as the eigenvalue scale that distinguishes the blocks, so the associated major eigendirections are not stably identified uniformly over states.}

For part~(2), we instead ensure significant structure by making the ``major'' eigenspace fixed and strongly detectable through noise, so it is recoverable and high-pass-through welfare gains remain available. The key observation is that \emph{consumer} surplus pass-through behaves very differently from producer-surplus pass-through (\Cref{prop:eig,lem:welfare-pass}). Along the large-eigenvalue directions that drive $\dot{W}$ (\Cref{eq:intuition1}), prices barely move, so $\dot{C}$ is negligible. To obtain a fixed positive consumer-surplus incidence per dollar of expenditure uniformly, a rule must load on directions with smaller eigenvalues, where price pass-through is order one; we hide one such direction in the quantity vector while making the relevant quantity signals nearly indistinguishable, and with a proof similar to that of part (1) show that the desired lower bound on $\dot{C}/\dot{S}$ must fail in some state.

We now note a negative result that focuses on expected performance, which is often a focus of the frequentist statistical decision theory literature \citep{manski2004statistical}.

\begin{remark} The construction used for \Cref{Prop:LackOfAS}(1) also yields the following implication for expected performance.

 After normalizing any positive-expenditure rule so that $\dot{S}_{\bm R(\widehat{\bm{\theta}})}=1$, 
    \[
    \inf_{\bm{\theta} \in \bm{\Theta}(n)}\mathbb{E}_{\varphi_{\bm{\theta}}(n)}
    \left[
    \dot{C}
    +\dot{P}
    -\dot{S}
    \right]
    \leq \varepsilon.
    \]
A proof is given in \Cref{app:lackofas_expected_performance}.
\end{remark}

\subsection{Theoretical foundations and empirical tests for significant structure}\label{concluding remarks_sampling} 

Significant structure is a joint property concerning the norm of estimation errors in market quantities and $\bm{D}$, and the operator norm of $\bm{D}$. We now discuss foundations for and tests of the key parts of this assumption.

\subsubsection{Theoretical foundation}
In \Cref{sec:SBM} we provide a block demand structure under which the norm of $\bm{D}$ grows linearly in $n$. In this class of demand environments, whenever the norm of the estimation errors grows in $n$ at a slower rate, we obtain recoverability of eigenvectors associated with large eigenvalues. The results raise the question of whether estimation procedures that an authority can carry out are consistent with this requirement. 

In \Cref{ap:sampling} we formalize a natural estimation procedure under which $\Vert \bm{E} \Vert \in O_{\mathrm{p}}(\sqrt{n})$. The authority samples a distinct household for each product pair $(i,j)$, and performs a demand experiment to obtain \emph{independent, unbiased} estimates of how the demands for both products change when $p_i$ is perturbed.   
For each product $i$, we obtain $n$ distinct estimates of the own-price effect, and so we can recover the own-price demand derivative ($D_{ii}$ for each $i$) by the law of large numbers; thus we can satisfy \Cref{as:error_assumption}. Turning to off-diagonal entries, standard random matrix theory implies that, under suitable assumptions, the operator norm of the estimation error matrix is $O_{\mathrm{p}}(\sqrt{n})$.

Estimation errors on demand derivatives may exhibit spatial or structural correlations: for example, errors in estimating cross-price elasticities between nearby product categories may be positively correlated.  Recent advances in random matrix theory show that the key consequence we require---namely, that the spectral norm $\Vert \bm{E} \Vert$ is $o(n)$---holds under far more general dependence structures.\footnote{When correlations between entries decay polynomially with distance under a suitable metric, or when entries exhibit finite-range dependence (i.e., $E_{ij}$ and $E_{kl}$ are independent whenever indices are sufficiently far apart), the spectral norm of a symmetric $n \times n$ error matrix $\bm{E}$ satisfies $\Vert \bm{E} \Vert = O_{\mathrm{p}}(\sqrt{n})$; see \citet{erdos2019slow} for slowly decaying correlations, \citet{anderson2008finite} for finite-range dependence, and \citet{reker2022operator} for polynomially decaying metric correlations. These results provide the operator-norm control needed for condition (2) of \Cref{Def:SS}.}

\subsubsection{A test for significant structure}

In applications, significant structure is a substantive joint condition on (i) the scale of recoverable eigenvalues of the normalized Slutsky matrix and (ii) the magnitude of observation error. Since this condition is not directly observable, it is useful to have an empirical diagnostic. The proof of our main result implies a statistical test that can be used to assess whether significant structure is present.

The test works as follows: Take one estimate of the market state and construct the intervention recommended in our positive result. Then  take an independent estimate of the market state, and \emph{treating this estimate as if it were exactly the true market state}, assess the net total surplus effect from using the intervention constructed based on the first estimate. If significant structure is present, the predicted net total surplus increase per dollar of predicted spending should be close to $1$. 

The reason for this is that under significant structure, the information on the underlying Slutsky matrix $\bm{D}$ will accurately reveal eigenspaces corresponding to large eigenvalues. Along eigenspaces with large eigenvalues, the total surplus pass-through behavior of interventions is the same for the true matrix and for any perturbed version $\bm{D} + \bm{E}$ where $\bm{E}$ is drawn like our error matrices; that is what makes the recovery procedure so robust. Thus, an intervention that is chosen to have a good effect for one sample should perform well in another.

More formally: Suppose the analyst has \(T\ge 4\) independent estimates of the market state (e.g., from disjoint subsamples); call them \((\widehat{\bm{D}}^{(t)},\widehat{\bm{q}}^{0,(t)})_t\). Let \(\underline{\widehat{\bm{\sigma}}}^{(t)}\) denote the  intervention produced by our rule (see \Cref{sec:rule}) from the  estimate $t$, with target spending level \(s=1\):
\[
\underline{\widehat{\bm{\sigma}}}^{(t)} = \frac{P_{\mathcal{L}(\underline{\widehat{\bm{D}}}^{(t)},\, b(n))}\,\underline{\widehat{\bm{q}}}^{0,(t)}}{\bigl\Vert P_{\mathcal{L}(\underline{\widehat{\bm{D}}}^{(t)},\, b(n))}\,\underline{\widehat{\bm{q}}}^{0,(t)} \bigr\Vert^2}
\]
For two splits \(t\neq t'\), the predicted welfare effect of the intervention designed on split \(t\), evaluated at the state estimate from split \(t'\), is:  
\begin{equation}\label{eq:test}
 \widehat{\tau}^{(t\to t')}:= \sum_{\ell=1}^n \bigl(\widehat{\bm{u}}^{\ell,(t')} \cdot \underline{\widehat{\bm{q}}}^{0,(t')}\bigr)\bigl(\widehat{\bm{u}}^{\ell,(t')} \cdot \underline{\widehat{\bm{\sigma}}}^{(t)}\bigr)\frac{|\widehat{\lambda}_\ell^{(t')}|}{1+|\widehat{\lambda}_\ell^{(t')}|},  
\end{equation}
 where \((\widehat{\lambda}_\ell^{(t')},\widehat{\bm{u}}^{\ell,(t')})\) is the eigendecomposition of \(\underline{\widehat{\bm{D}}}^{(t')}\) and the normalization is as in \Cref{sec:norm}. We pair the \(T\) estimates arbitrarily. For simplicity, assume \(T\) is even, writing \(T=2Z\), and for each pair \(z=1,\dots,Z\), define the pair statistic  
\[
\widehat{\tau}^{(z)}:=\frac{1}{2}\Bigl(\widehat{\tau}^{(2z-1\to 2z)}+\widehat{\tau}^{(2z\to 2z-1)}\Bigr),
\qquad
\overline{\tau}:=\frac{1}{Z}\sum_{z=1}^Z \widehat{\tau}^{(z)}.
\]
Here the arrow notation \(t\to t'\) means: construct the intervention using split \(t\), but evaluate its predicted welfare effect using the independently estimated state from split \(t'\). The statistic \(\widehat{\tau}^{(z)}\) symmetrizes this cross-evaluation within each paired split.

We now formalize the prediction that if the true state $\bm{\theta}$ belongs to an environment with significant structure then \(\widehat{\tau}^{(z)}\) is near one.

Let \(\mu_n(\bm{\theta}):=\mathbb{E}[\widehat{\tau}^{(1)}\mid \bm{\theta}]\) denote the (conditional) population mean of the paired statistic when there are \(n\) products. We form a one-sided \((1-\alpha)\) lower confidence bound for \(\mu_n(\bm{\theta})\) via the nonparametric bootstrap at the level of pairs: resample \(\{\widehat{\tau}^{(z)}\}_{z=1}^Z\) with replacement many times (indexing these bootstrap pairings $r=1,2,\ldots,R$) compute the resampled means \(\overline{\tau}^{(r)}\), and let \({L}_{1-\alpha}\) be the empirical $\alpha$-quantile of \(\overline{\tau}^{(r)} \). We say significant structure is not certified at tolerance \(\epsilon\) whenever \(L_{1-\alpha}<1-\epsilon\). Intuitively, a bound well below one indicates that predicted total surplus falls materially short of predicted expenditure in the resampling exercise. \Cref{app:split_sample_testing} formalizes the coverage guarantee for this diagnostic, proving that certification occurs with high probability when significant structure is indeed present.\medskip 

\paragraph{Illustration based on \Cref{Sec:Illustration}.} We consider the environment constructed in \Cref{Sec:Illustration} and consider a variety of values of $b/\sqrt{n}$. We generate $T=500$ independent noisy estimates of demand and quantities, and  for each bootstrap draw $r$, we pair them up and then compute the $\overline{\tau}^{r}$ outcome for that pairing. We repeat this many times with different pairings, take the $5^{\text{th}}$ percentile of the resulting empirical distribution (this is the curve in \Cref{fig:stat_test}), and reject when that number is below $1-\epsilon$.

Thus, the test is quite demanding in terms of how the intervention has to perform when an out-of-sample market state is taken as ground truth. And yet our result in \Cref{app:split_sample_testing} shows that when significant structure is present and ``strong'' in the sense of a large $b/\sqrt{n}$ ratio, it is very likely that the intervention does indeed perform very well out of sample, so that the test statistic falls above the line.

\begin{figure}[htbp]
\centering
\includegraphics[width=0.7\textwidth]{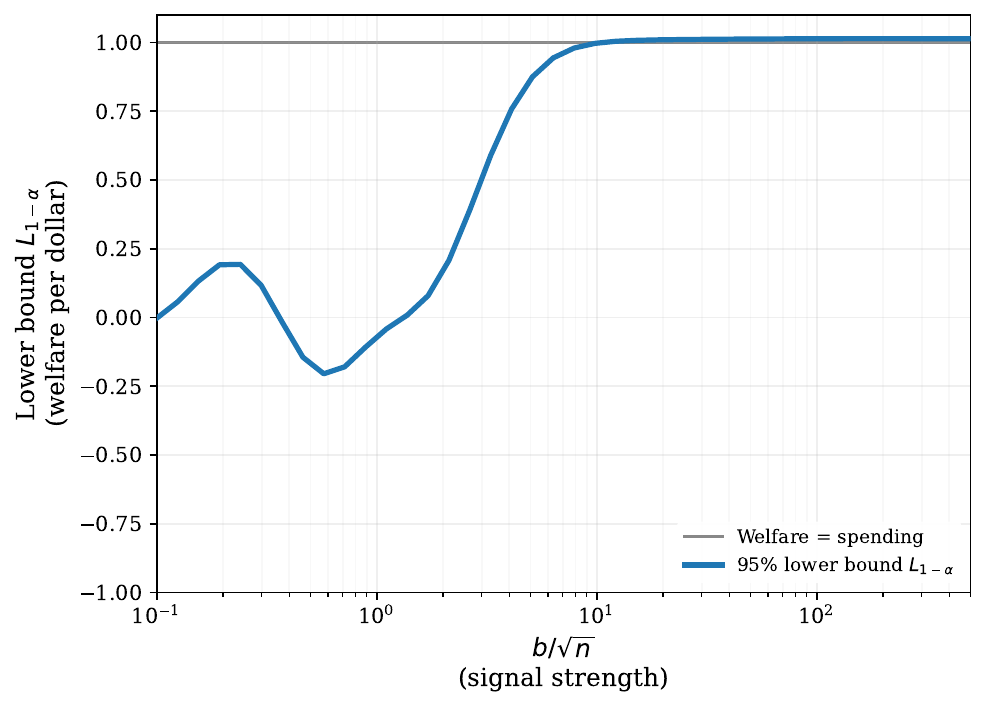}
\caption{Bootstrap lower confidence bound for cross-evaluated total surplus. For each value of $b/\sqrt{n}$, the line gives the $5^{\text{th}}$ percentile of the empirical distribution of test statistics across the bootstrap experiments. At low signal strength this is well below one. At high signal strength this value approaches one, consistent with significant structure. At tolerance \(\epsilon\), values below \(1-\epsilon\) indicate a failure to certify significant structure according to this diagnostic.}
\label{fig:stat_test}
\end{figure}

\subsection{Linear--quadratic utility and relationship with hedonic utility models}\label{sec:hedonic}

\subsubsection{Linear--quadratic utility}
Our leading example, the block model of demand in \Cref{sec:illustration3}, is consistent with optimal behavior of a representative consumer with a standard quasilinear quadratic utility function (as in \cite{VivesandSing}). Suppose there are $K$ product categories with $m$ products in each category. The representative consumer's utility is:
$$
U(\bm{q})=
\bm{a}^{\tr} \bm{q}-\frac{1}{2}\bm{q}^\tr\bm{H}\bm{q},
$$
where $\bm{a}\in \mathbb{R}^n_{+}$, $\bm{H}$ is an $n\times n$ symmetric matrix with $H_{ii}=z>0$, $H_{ij}=z_{\text{In}}>0$ if $i\neq j$ belong to the same category and, otherwise, $H_{ij}=z_{\text{Ex}}<0$. Hence, the marginal utility of product $i$'s consumption is decreasing in the consumption of products within the same category and increasing in the consumption of products in different categories. Given prices $\bm{p}$, the first-order condition of the maximization of \( U-\bm{p}\cdot \bm{q} \)  with respect to \( \bm{q} \) provides the demand function:
\[
\bm{q}(\bm{p}) = \bm{H}^{-1}(\bm{a} - \bm{p}),
\]
and, differentiating it with respect to prices yields the Jacobian:
\begin{equation} \label{Eq:discussion}
\bm{D}=\frac{\partial \bm{q}(\bm{p})}{\partial \bm{p}} = -\bm{H}^{-1}.
\end{equation}
For each market state $\bm{\theta}$, we wish to find $(z^*(\bm{\theta}),z^*_{\text{In}}(\bm{\theta}),z^*_{\text{Ex}}(\bm{\theta}))$ so that the $\bm{D}$ in (\ref{Eq:discussion})  corresponds to the Slutsky matrix in \Cref{sec:illustration3}, which is equivalent to solving a system of three equations fully specified in \Cref{app:linear}.

 As we increase the number of product categories $K$ we want the largest (in absolute value) eigenvalue of $\bm{D}$ grows linearly in $K$ and that $\bm{D}$ is negative definite.  This requires that, as we increase $K$, the following holds:  
\begin{equation}\label{Eq:B}
 \lim_{K\rightarrow \infty} z^*_{\text{Ex}}(\bm{\theta})m(K-1)+(z^*(\bm{\theta})+z^*_{\text{In}}(\bm{\theta})(m-1))= 0^{+}.  
 \end{equation}
The left-hand side is the corresponding eigenvalue of $\bm{H}$ in the all-ones direction (equivalently, the curvature of $U$ in that direction is its negative). Hence, (\ref{Eq:B}) requires that, as we expand product categories, this eigenvalue remains positive (preserving concavity) but approaches zero. Intuitively, this captures that increases of consumption in this direction remain valuable, corresponding to the idea of large-scale complementarities.

 One might have thought that a near-singular Hessian causes some kind of instability or bad behavior of the equilibrium around such a point.  Our results show that this is not the case, due to the way the Hessian maps into $\bm{D}$ and then into the system of equations governing the equilibrium.  When all is said and done, cost perturbations in any direction have bounded pass-through to prices and quantities.  
 
Finally, we note that equilibrium quantities and prices are well-behaved in an economy defined by (\ref{Eq:B}). To see this, consider the case where $a_i=a$ and $c_i=c$ for all $i$; this ensures that the equilibrium is symmetric and takes the following form:
\[
p=\frac{c}{1+|\lambda_1(\bm{\theta})|}+\frac{|\lambda_1(\bm{\theta})|a}{1+|\lambda_1(\bm{\theta})|} \text{ and } q= \frac{|\lambda_1(\bm{\theta})|}{1+|\lambda_1(\bm{\theta})|}(a-c),
\]
where $\lambda_1(\bm{\theta})$ is the dominant eigenvalue of the Slutsky matrix defined in \Cref{sec:illustration3}. Note that equilibrium price and quantity depend on $K$ only via $|\lambda_1(\bm{\theta})|$. As $K$ goes to infinity $|\lambda_1(\bm{\theta})|\rightarrow \infty$ and  $p\rightarrow a$ and $q\rightarrow (a-c)$. This is consistent with \Cref{Ass:B-Model}(a)(2) imposed on quantity  in the illustrative example of \Cref{sec:SBM}, since any uniform vector satisfies it after normalization.

\subsubsection{Relationship to hedonic models}
Recent work by \citet{pelligrino2021} and \citet{edererpelligrino2021} uses a canonical oligopoly model to empirically quantify the evolution of market power. In particular, in \citet{pelligrino2021}, the model of demand is hedonic, in the spirit of  \citet{lancaster1966new}: the household's utility is additively separable in the contributions of various \emph{characteristics}, and a product provides a bundle of these characteristics. 

As in the linear--quadratic demand above, the Slutsky matrix $\bm{D}$ derived from the hedonic model can be expressed as (\ref{Eq:discussion}). 
In this case, $\bm{H}$ is a convex combination of the identity matrix and the cosine similarity matrix of products' characteristics, which \citet{hoberg2016text} estimated for a large set of consumer goods  using text data.\footnote{\citet{pelligrino2021} and \citet{edererpelligrino2021} consider quantity competition, but the Slutsky matrix does not depend on this choice.} An implication of this is that all entries $H_{ij}$ are positive. This means that the marginal utility of consuming product $i$ is decreasing in product $j$'s consumption, and this holds for all products $(i,j)$. By assumption, the hedonic model does not allow for products that the consumer likes to consume together, e.g., a tennis racket and a tennis ball.

This turns out to impose a bound on $\Vert\underline{\bm{D}}\Vert$, and so significant structure cannot be realized, as we show in \Cref{app:hedonic}.\footnote{We checked that using the data and the parameter values used by \citet{pelligrino2021}, the leading eigenvalues of the Slutsky matrix are small.}  

The intuition can be found in \Cref{S:PT}: blocks of substitutes do not create large eigenvalues.  However, as that section also illustrated, significant structure will naturally emerge once we broaden our view of the market to allow for both complementary and substitute products. There is a straightforward economic reason for why such complementarities more readily produce significant structure: when  one good's price decreases, all its complements can experience comparable nonvanishing increases in demand. This creates the clusters of nonvanishing entries in $\bm{D}$ that are the hallmark of significant structure.

In summary, direct complementarities seem practically important and can naturally yield the significant structure central to our results. We hope these observations will motivate further empirical research on the structure of large-scale oligopoly models with complementary goods.

\subsection{Nonlinear demand}\label{concluding remarks_NL}
We have assumed that each good’s demand is exactly linear in its own price in a neighborhood of the status quo equilibrium. This assumption has helped significantly with tractability; for example, it afforded a very clean characterization of the complete-information benchmark in \Cref{prop:Welfare CI}.

This section clarifies the aspects of our analysis that extend with non-linear demand, and outlines aspects that need further development.

\subsubsection{Two conceptually distinct steps}
Our analysis rests on two results that play different roles:
\begin{enumerate}
\item A \emph{spectral decomposition of total surplus pass-through}: given any price perturbation $\dot{\bm{p}}$, the first-order total surplus effect depends only on the Slutsky matrix $\bm{D}$ and equilibrium markups at the status quo, via an interpretable spectral decomposition. This result is independent of the assumption that each good’s demand is locally linear in its own price.
\item A \emph{mapping from subsidies to prices}: in the linear environment, equation (\ref{eq:PricePass}) gives $\dot{\underline{\bm{p}}}=-(\bm{I}-\underline{\bm{D}})^{-1}\underline{\bm{\sigma}}$, so that the eigenbasis of $\underline{\bm{D}}$ diagonalizes both total surplus and the cost-to-price mapping. This result exploits the assumption that each good’s demand is locally linear in its own price in a fundamental way.
\end{enumerate}
The first step identifies which price movements are desirable; the second step shows how to implement them via per-unit subsidies. Local linearity is essential only for the second step.

\subsubsection{First-order welfare theory extends fully}
Because demand $\bm{q}(\bm{p})$ is differentiable, the first-order quantity response to any small price perturbation is $\dot{\bm{q}}=\bm{D}\dot{\bm{p}}$, without any linearity assumptions whatsoever. Indeed, the first-order welfare effect of a price perturbation is
\begin{equation}\label{eq:welfare_nonlinear}
\dot{W} = \sum_\ell \lambda_\ell \big[(\underline{\bm{p}}^0-\underline{\bm{c}}^0)\cdot \bm{u}^\ell\big]\big(\dot{\underline{\bm{p}}}\cdot \bm{u}^\ell\big),
\end{equation}

This is of interest in its own right. The expression shows that welfare gains depend on the alignment among three objects: the status-quo markup vector $(\underline{\bm{p}}^0-\underline{\bm{c}}^0)$, the demand responsiveness encoded in $\underline{\bm{D}}$, and the direction of the price perturbation $\dot{\underline{\bm{p}}}$. The spectral form decomposes this into independent contributions from each eigendirection: a price movement in direction $\bm{u}^\ell$ contributes to total surplus in proportion to the eigenvalue $\lambda_\ell$ and the markup's loading on that direction.

 Evaluating the Bertrand first-order condition at the status quo and normalizing gives  
\begin{equation}\label{eq:markup_quantity}
\underline{\bm{p}}^0-\underline{\bm{c}}^0=\underline{\bm{q}}^0,
\end{equation}
 so the total surplus formula becomes $$\dot{W} = \underline{\bm{q}}^0 \cdot \underline{\bm{D}} \;  \dot{\underline{\bm{p}}} = \sum_\ell \lambda_\ell \big(\underline{\bm{q}}^0 \cdot \bm{u}^\ell\big)\big(\dot{\underline{\bm{p}}} \cdot \bm{u}^\ell\big).$$ This is operationally advantageous: quantities are typically easier to measure than marginal costs, so the authority can evaluate potential price movements using observable data.  (Note, however, that the normalized quantities and prices require good estimates of $\partial_i q_i(\bm{p})/\partial p_i$ to perform the normalization---arguably still more accessible than precise measurements of markups.)  

The spectral perspective---that large eigenvalue directions with substantial quantity loadings matter most for total surplus---is thus fully general and does not depend on local linearity.

\subsubsection{Price pass-through of cost changes does not extend directly}
The assumption that each good’s demand is locally linear in its own price is used in our theory only in the mapping from subsidies to prices. Under this assumption, subsidies in the eigendirection $\bm{u}^\ell$ pass through exclusively to the price and quantity of that bundle, with coefficients that depend only on $\lambda_\ell$ (\Cref{prop:eig}). This ``diagonalization'' is what allows us to design robust interventions by targeting cost interventions to the estimated top eigenspace.

 When the assumption that each good’s demand is locally linear in its own price fails, the price pass-through of a per-unit subsidy depends additionally on the relevant slice of the demand-curvature tensor \citep[see, e.g.,][]{miklos2021pass} $$\left( \partial^2 q_i(\bm{p})/\partial p_i \partial p_k \right)_{i,k}.$$ The eigenbasis of $\bm{D}$ no longer diagonalizes the cost-to-price mapping, and interventions designed to move prices in a particular eigendirection may have significant effects in other directions.  
 
\subsubsection{Paths forward}
Extending \Cref{Th:Main} to the case beyond the linear setting of \Cref{A1} requires new techniques for one of two tasks: either (i) characterizing which cost perturbations achieve desired price movements when the price pass-through operator is more complicated, or (ii) designing alternative instruments that directly control prices.

For the first approach, the challenge is learning the mapping from cost perturbations to price movements. For example, under the most direct extension of the Bertrand equilibrium theory, predicting pass-through requires  recovering both the Slutsky matrix and certain mixed partials in the Hessian of the demand function. Performing this recovery from noisy observations is a harder problem than the one our framework addresses.\footnote{See, e.g., \cite*{miravete2024targeted} and \cite*{birchall2024estimating} for recent work on some of the challenges introduced by uncertainty about the shape/curvature of demand in differentiated-product markets.} There are many more numbers to recover, and a priori it is unclear how robust interventions should be generalized, since the linear algebra becomes more complicated. There is also the possibility that pass-through from costs to prices does not obey the equations of a Bertrand model.

A general approach here is to assume that prices move differentiably in costs but otherwise take the local mapping $\bm{L}$ in $ \dot{\bm{p}} = \bm{L} \dot{\bm{c}}$ as a black box. The goal then is simply to recover nonzero vectors in the preimage of the subspace of price vectors that we know have predictable welfare effects. Here \emph{compressed sensing} techniques are relevant---methods for recovering structured signals from limited measurements by exploiting sparsity or low-rank structure \citep[see, e.g.,][]{foucart2013compressive}.  We leave to future work the question of what assumptions on pass-through (whether with or without structural foundations) enable this type of robust intervention design.
 
For the second approach, one can sidestep the pass-through problem entirely by designing mechanisms that directly enforce target price perturbations, or incentivize them in a direct way (rather than by tailoring cost perturbations). Our spectral analysis of welfare effects remains useful for a planner who has power to control prices in this way: the Davis--Kahan analysis identifies which eigenspaces of $\widehat{\underline{\bm{D}}}$ can be trusted. This allows the authority to know that targeting price movements in these directions yields reliable total surplus improvements.  This is the most immediate way to use the ideas behind \Cref{Th:Main} in environments that do not satisfy \Cref{A1}.

\subsection{Games on networks}\label{concluding remarks_NG}
 We have studied robust interventions in an oligopoly pricing model, which can be seen as a network game with linear best replies \citep[e.g.,][]{Ballesteretal2006,bramoulle2016oxford} where the network summarizes the complement and substitute relationships among products. The methods we have developed can be extended to other settings where robust interventions have not been studied. For example, in a public goods setting, interventions would aim to realign private marginal returns with social marginal returns. The literature has developed tools to understand how to do this when the authority has precise information on the spillovers causing the underprovision of public goods.\footnote{See, for instance, \cite{bramoulle2014strategic} and \cite{ElliottGolub2019public}.
} However, we know little about designing interventions under noisy information about such externalities. Similarly, in contracting for teams under moral hazard, network methods have recently been developed for locally perturbing contracts to achieve better outcomes for a principal \citep{IncentiveSpillovers}. But it is a considerable challenge to extend these results to the realistic case where the strategic interactions among members of an organization are only imperfectly known. General games will lack some of the structure we have leveraged, including the symmetry and negative semidefiniteness of the strategic interactions matrix. So there are challenges to overcome in extending our results. We hope this paper stimulates research in these directions.

\newpage

{\footnotesize
\bibliographystyle{ecta}
\bibliography{svd,tax,refs}}
\newpage 
\appendix
\crefalias{section}{appendix}
\crefalias{subsection}{appendix}

\section{Omitted proofs and details for main results}

\subsection{Proof of Part (2) of \Cref{prop:Welfare CI}} \label{app:complete}
Since $\bm{D}$ is a non-diagonal matrix, the normalized $\underline{\bm{D}}\neq -\bm{I}$. Hence, there exists $\ell\neq \ell'$ with $\lambda_\ell\neq \lambda_\ell'$. For generic $\underline{\bm{q}}^0$, we also have that $(\bm{u}^\ell\cdot \underline{\bm{q}}^0)(\bm{u}^{\ell'}\cdot \underline{\bm{q}}^0)\neq 0$. Without loss, take $\ell=1$ and $\ell'=2$. We now construct a set of interventions under which we can obtain any outcome $(\dot{C}, \dot{P}, \dot{S})$ that satisfies (\ref{eq:Pareto_Identity}).

For any real number $\beta\neq 0$ and any $\dot{S}$, consider the following class of interventions
$$\underline{\bm{\sigma}}^0(\beta)=\beta\left( (\bm{u}^{1}\cdot \underline{\bm{q}}^0)\bm{u}^{1} +\alpha (\bm{u}^{2}\cdot \underline{\bm{q}}^0)\bm{u}^{2}\right),$$ where $\alpha$ is chosen so that $ \underline{\bm{\sigma}}^0 \cdot \underline{\bm{q}}^0=\dot{S}$; that is:
$$  \alpha = \frac{\dot{S}/\beta -(\bm{u}^{1}\cdot \underline{\bm{q}}^0)^2}{(\bm{u}^{2}\cdot \underline{\bm{q}}^0)^2}.$$
Hence, as we vary $\beta$, we keep $\underline{\bm{\sigma}}^0(\beta)\cdot \underline{\bm{q}}^0=\dot{S}$. Next, note that 
\begin{eqnarray*}
\dot{C} &=&\sum_{\ell}\frac{1}{1+|\lambda_\ell|}(\bm{u}^\ell \cdot \underline{\bm{q}}^0)(\bm{u}^\ell \cdot \underline{\bm{\sigma}}^0(\beta))=\frac{\beta}{1+|\lambda_1|}(\bm{u}^{1}\cdot \underline{\bm{q}}^0)^2  
 + \frac{\alpha \beta}{1+|\lambda_2|}(\bm{u}^{2}\cdot \underline{\bm{q}}^0)^2\\
 &=& 
\frac{\beta}{1+|\lambda_1|}(\bm{u}^{1}\cdot \underline{\bm{q}}^0)^2 +\frac{\dot{S} -\beta(\bm{u}^{1}\cdot \underline{\bm{q}}^0)^2}{(\bm{u}^{2}\cdot \underline{\bm{q}}^0)^2} \frac{1}{1+|\lambda_2|}(\bm{u}^{2}\cdot \underline{\bm{q}}^0)^2\\
&=&\beta(\bm{u}^{1}\cdot \underline{\bm{q}}^0)^2\left(\frac{1}{1+|\lambda_1|}- \frac{1}{1+|\lambda_2|}\right)+\frac{1}{1+|\lambda_2|} \dot{S}.
\end{eqnarray*}
Hence, $\dot{C}$ evaluated at $\bm{\sigma}^0(\beta)$ depends linearly on $\beta$, which implies that, for a given $\dot{S}$, we can achieve any given $\dot{C}$ by choosing $\beta$ appropriately.  \qedhere

\subsection{Proof of Proposition 2: Significant structure in the block model demand } \label{app:largeK}

We show that the environment covered by \Cref{prop:example} is a special case of the environment covered by \Cref{Th:Main}. Since in the block model $D_{ii}=-1$, \Cref{Ass:Asymptotics,as:error_assumption} are satisfied. Thus, it remains to show that significant structure holds. The following Lemma is key for this.

\begin{lem}[Spectral radius and positive block--constant eigenvector]
\label{lem:spectral-block}
Fix the state and let $\bm{D}$ satisfy the block structure with $K$ product categories each of size $m$. Denote its spectral radius by
\(
\rho\bigl(\bm{D}\bigr) = \max_j |\lambda_j(\bm{D})|.
\)
Then:

\begin{enumerate}
  \item[(a)] There exist constants $0 < c_1 \le c_2 < \infty$, independent of $K$, such that
  \[
    c_1 K \;\le\; \rho\bigl(\bm{D}\bigr) \;\le\; c_2 K
    \qquad \text{for all large enough} K.
  \]
  \item[(b)] For all large enough $K$ there exists an eigenpair $(\lambda, \bm{u})$ of $\bm{D}$ such that
  \begin{itemize}
    \item $\lambda < 0$ and $|\lambda| \ge c_1 K$;
    \item $\bm{u}_i > 0$ for all products $i$;
    \item $\bm{u}_i = \bm{u}_j$ whenever $k(i) = k(j)$ (the eigenvector is constant within categories).
  \end{itemize}
\end{enumerate}
\end{lem}

A key assumption we use for \Cref{lem:spectral-block} is part~(a)(1) of \Cref{Ass:B-Model} which rules out the possibility that cross-category demand complementarity vanishes. Part~(a) of \Cref{lem:spectral-block} shows that, for each state, the spectral radius is of order K. Part~(b) of \Cref{lem:spectral-block} shows that there is a large (negative) eigenvalue of order K whose eigenvector is strictly positive and constant within product categories. 

\Cref{lem:spectral-block} shows that $b(n)$ goes to infinity at a rate of $n$ when we increase $n$ by increasing $K$, holding $m$ constant. Part~(b) of \Cref{Ass:B-Model} implies that the norm of the noise diverges at a rate of $\sqrt{n}$. Hence, part~(2) of \Cref{Def:SS} holds. Part~(a)(2) of \Cref{Ass:B-Model} and part~(b) of \Cref{lem:spectral-block} imply part~(3) of \Cref{Def:SS} holds.

\begin{proof}[Proof of \Cref{lem:spectral-block}]

\emph{Step 1: Reduction to a $K \times K$ matrix.}
Let
\[
  \mathcal S
  \left\{
    \bm{u} \in \mathbb R^n :
    \bm{u}_i = \bm{u}_j \ \text{whenever } k(i) = k(j)
  \right\}
\]
be the subspace of vectors that are constant within each category.
For $\bm{u} \in \mathcal S$, let $u_c$ denote its common value on category $c$.
If $i$ is a product in category $c$, then
\[
  (\bm{D} \bm{u})_i
  =
  \bigl(-1 + (m-1) \bm{B}_{cc}\bigr) u_c
  +
  \sum_{\ell \neq c} m \bm{B}_{c\ell} u_\ell,
\]
which depends only on $c$.
Hence $\bm{D} \mathcal S \subseteq \mathcal S$.
On $\mathcal S$, $\bm{D}$ acts as the $K \times K$ symmetric matrix $\bm{A} = (\bm{A}_{c\ell})$ with
\begin{equation}
  \bm{A}_{cc} = -1 + (m - 1) \bm{B}_{cc},
  \qquad
  \bm{A}_{c\ell} = m \bm{B}_{c\ell} \quad (c \neq \ell).
  \label{eq:def-AK}
\end{equation}
In particular, if $\bm{A} \bm{y}=\lambda \bm{y}$, then the lifted vector $\bm{u}\in\mathbb{R}^n$ defined by $\bm{u}_i:=\bm{y}_{k(i)}$ satisfies $\bm{D}\bm{u}=\lambda \bm{u}$.

\smallskip

\emph{Step 2: Spectral-radius bounds.}
Let $\bm{v}:=\bm{1}/\sqrt{K}$.
Since $\bm{A}$ is symmetric, Rayleigh--Ritz implies
\[
  \lambda_{\min}\bigl(\bm{A}\bigr)
  \le
  \bm{v}^\top \bm{A} \bm{v}
  =
  \frac{1}{K}\sum_{c,\ell=1}^K \bm{A}_{c\ell}.
\]
By Property~NSD and the normalization $D_{ii}=-1$, we have $|D_{ij}|\le 1$ for all $i\neq j$, hence $\bm{B}_{cc}\le 1$ and $\bm{A}_{cc}\le m-2$.
Moreover, by part~(a)(1) of \Cref{Ass:B-Model}, for $c\neq \ell$ we have $\bm{B}_{c\ell} < \kappa<0$, so $\bm{A}_{c\ell} < m\kappa$.
Therefore,
\[
  \lambda_{\min}\bigl(\bm{A}\bigr)
  \le
  (m-2) + m\kappa(K-1)
  \le
  -c_1 K
\]
for all sufficiently large $K$, for some $c_1>0$ independent of $K$.
Since $\lambda_{\min}(\bm{A})$ is also an eigenvalue of $\bm{D}$, it follows that
\(\rho(\bm{D}) \ge |\lambda_{\min}(\bm{A})| \ge c_1 K\).

For the upper bound, again $|D_{ij}|\le 1$ implies
\[
  \rho\bigl(\bm{D}\bigr)
  \le
  \|\bm{D}\|_\infty
  := \max_i \sum_{j=1}^n |D_{ij}|
  \le
  n
  =
  mK,
\]
which proves part~(a) with $c_2:=m$.

\smallskip

\emph{Step 3: A strictly positive block-constant eigenvector.}
Define $\bm{M}:= m\bm{I}-\bm{A}$.
Then $\bm{M}_{cc}=m-\bm{A}_{cc}\ge 2$ and, for $c\neq \ell$, $\bm{M}_{c\ell}=-\bm{A}_{c\ell}=-m\bm{B}_{c\ell}>0$, so $\bm{M}$ is entrywise positive.
By the Perron--Frobenius theorem, $\bm{M}$ has a largest eigenvalue $\mu=\rho(\bm{M})$ with an associated eigenvector $\bm{y}\in\mathbb{R}^K$ satisfying $\bm{y}_c>0$ for all $c$.
Since $\bm{M}=m\bm{I}-\bm{A}$, we have $\bm{A}\bm{y}=(m-\mu)\bm{y}$.
Because the eigenvalues of $\bm{M}$ are $m-\lambda_j(\bm{A})$, we have $m-\mu=\lambda_{\min}(\bm{A})$.
Thus $\bm{y}$ is an eigenvector of $\bm{A}$ associated with $\lambda_{\min}(\bm{A})$, which satisfies $\lambda_{\min}(\bm{A})\le -c_1K$ by Step~2.
Lifting $\bm{y}$ to $\bm{u}\in\mathbb{R}^n$ by $\bm{u}_i:=\bm{y}_{k(i)}$ yields an eigenpair $(\lambda,\bm{u})$ of $\bm{D}$ with $\lambda=\lambda_{\min}(\bm{A})<0$, $|\lambda|\ge c_1K$, $\bm{u}_i>0$ for all $i$, and $\bm{u}$ constant within categories.
This proves part~(b).
\end{proof}

\subsection{Proof of \Cref{Th:Main}} \label{ap:main_proof}
First, we introduce a normalization that will be helpful in some of our asymptotic arguments, to avoid carrying around the norm of the quantity vector.

\begin{assumption} \label{as:normalized_q} $\Vert \bm{q}^0 \Vert =1$. 
\end{assumption}

\begin{remark} \textbf{Normalization of $\bm{q}^0$.}  \label{rem:normalization}
Our analysis is invariant to a common rescaling of quantities together with a compensating change in the unit of the numeraire.  Fix any sequence of market states $\bm{\theta}(n)=(\bm{D}(n),\bm{q}^0(n))\in\bm{\Theta}(n)$ in the environment, with associated baseline price vectors $\bm{p}^0(n)$.  For any sequence $\kappa(n)>0$, define transformed quantities by $\tilde{\bm{q}}^0(n):=\kappa(n)\bm{q}^0(n)$  and measure the numeraire in units that are $\kappa(n)$ times larger.  Then the numerical price vectors $\bm{p}^0(n)$, budget sets, Hicksian demands, and normalized Slutsky matrices are unchanged, and welfare and expenditure derivatives (which depend on $(\bm{q}^0(n),\bm{\sigma}(n))$ only through inner products) are identical, in terms of the new numeraire, to those computed from $(\tilde{\bm{q}}^0(n),\tilde{\bm{\sigma}}(n))$.  Taking $\kappa(n) = 1/\Vert \bm{q}^0(n) \Vert$ yields an observationally equivalent specification with $\Vert \tilde{\bm{q}}^0(n) \Vert = 1$ for all $n$, so the normalization entails no loss of generality.
Under this rescaling, the content of the significant structure condition does not change because the rescaling constants cancel out of all the relevant comparisons. \end{remark}

We also assume the normalization $D_{ii}=-1$ for all $i$, which is also without loss of generality by \Cref{sec:norm}. This saves on writing underlines throughout the proof.

Next, we introduce some notation. For any matrix $\bm{M}$, we define $\mathbb{\Lambda}(\bm{M},\underline{\lambda})$ as the set of eigenvalues of $\bm{M}$ with absolute value at least $\underline{\lambda}$, and $\mathcal{L}(\bm{M},\underline{\lambda})$ as the space spanned by corresponding eigenvectors. 

For the rest of the proof we will fix an environment $\mathcal{E}=(\bm{\Theta}(n),\varphi(n))$ with significant structure and let $b(n)$ denote the associated sequence from \Cref{Def:SS}. Fix an arbitrary market state $\bm{\theta}=(\bm{D},\bm{q}^0)\in\bm{\Theta}(n)$; by \Cref{Def:SS}, all bounds below are required to hold uniformly over such $\bm{\theta}$.

Under \Cref{Def:SS}, the signal about ${\bm{D}}$ can be written
$\widehat{\bm{D}} = \bm{D} + \bm{E}$, where condition~(2) of \Cref{Def:SS} guarantees that, for every $\bm{\theta}\in\bm{\Theta}(n)$ and every $\eta>0$,
\begin{equation} \mathbb{P}_{\varphi_{\bm{\theta}}(n)}\!\left( \Vert \bm{E} \Vert > \eta\, b(n) \right) \to_n 0. \label{eq:E_bound} \end{equation}

We now define three threshold sequences
\[
   \underline{M}(n) := \tfrac{1}{2}b(n),\quad
   \widehat{M}(n) := \tfrac{3}{4}b(n),\quad
   M(n) := b(n).
\]
For each $n$:
\[
 \underline{M}(n) < \widehat{M}(n) < M(n),
\]
and when $n$ is sufficiently large the differences
\[
M(n)-\underline{M}(n), \quad \widehat{M}(n)-\underline{M}(n)\quad \text{ and } \quad M(n)-\widehat{M}(n)
\]
are all of order $b(n)$. Since $\Vert\bm{E}\Vert = o_{\mathrm{p},\bm{\Theta}}(b(n))$ by \eqref{eq:E_bound}, these differences dominate $\Vert\bm{E}\Vert$ in probability.

Now fix an arbitrary sequence $(\bm{D}(n),\bm{q}^0(n))\in\bm{\Theta}(n)$. We use the following notation: 
    \begin{align*}
        \underline{\Lambda}(n):=\mathbb{\Lambda}(\bm{D}; \underline{M}(n)), \quad \widehat{\Lambda}(n)&:=\mathbb{\Lambda}(\widehat{\bm{D}}; \widehat{M}(n)),  \quad {\Lambda}(n):=\mathbb{\Lambda}({\bm{D}}; M(n))  \\ 
        \underline{L}(n):=\mathcal{L}(\bm{D},\underline{M}(n)), \quad \widehat{L}(n)&:=\mathcal{L}({\widehat{\bm{D}}},\widehat{M}(n)), \quad L(n):=\mathcal{L}({\bm{D}},M(n)).
    \end{align*}
    
    Let $P_V$ be the projection operator onto subspace $V$ and $P^\perp_V$ its orthogonal complement. Let $(\lambda_1, \bm{u}^1), \ldots, (\lambda_n, \bm{u}^n)$ be eigenpairs of $\bm{D}$, with $|\lambda_1| \geq |\lambda_2| \geq \cdots \geq |\lambda_n|$.

    \medskip
    \paragraph{Proof strategy: sandwiching eigenspaces.}
    The three thresholds define three nested eigenspaces. The \emph{core space} $L(n)$ (threshold $M$) contains the eigenvectors of $\bm{D}$ with eigenvalues whose absolute value is at least $b(n)$. This is the space that we know is nonempty; indeed, this is exactly what is guaranteed, uniformly across possible market states, by condition (3) of \Cref{Def:SS}. The \emph{estimated space} $\widehat{L}(n)$ (threshold $\widehat{M}$, applied to the noisy matrix $\widehat{\bm{D}}$) is what we can actually compute from data; because $\|\bm{E}\|=o_{\mathrm{p},\bm{\Theta}}(b(n))$, the Davis--Kahan theorem guarantees that this space captures all directions in $L(n)$. The gap $M(n) - \widehat{M}(n) = \Theta(b(n))$ provides sufficient slack: any true eigenvalue above $M(n)$ produces an estimated eigenvalue above $\widehat{M}(n)$, and moreover the associated eigenvectors are almost fully in the span of eigenvectors of $\widehat{\bm{D}}$ with eigenvalues at least $\widehat{M}$.
    
    The \emph{buffer space} $\underline{L}(n)$ (threshold $\underline{M}$) plays a subtler role. We will design an intervention $\bm{\sigma}$ that lives in the estimated space $\widehat{L}(n)$---this is what we can compute from data. But to verify that this intervention actually works, we must analyze its welfare effects in the \emph{true} eigenbasis of $\bm{D}$, which we do not observe. The buffer space provided by $\underline{L}(n)$ helps us achieve this: it is defined based on the true matrix, yet Davis--Kahan guarantees that $\widehat{L}(n)$ is almost entirely contained in $\underline{L}(n)$. And $\underline{L}(n)$, since it is defined using the true $\bm{D}$, has nice properties (in particular, being spanned by eigenvectors with large eigenvalues) that guarantee good behavior of the designed intervention. Here again we use Davis--Kahan to show that the differences between $\widehat{L}(n)$ and $\underline{L}(n)$ cause only mild errors that do not disrupt a good intervention. We now dive into the details.
        
    \medskip

    We now define our intervention:
    \begin{equation}
    \bm{\sigma} = \frac{P_{\widehat{L}(n)}\widehat{\bm{q}}^0}{\Vert P_{\widehat{L}(n)}\widehat{\bm{q}}^0 \Vert^2} 
    \end{equation}

  The expenditure of this intervention is   
     \begin{equation} \label{eq:define_expenditure} 
    \dot{S}=\bm{\sigma} \cdot \bm{q}^0 =  \frac{P_{\widehat{L}(n)}\widehat{\bm{q}}^0}{\Vert P_{\widehat{L}(n)}\widehat{\bm{q}}^0 \Vert^2}     \cdot \bm{q}^0
     \end{equation} Our first main lemma, which we will prove shortly, will assert that $\dot{S}$ converges in probability to $1$.   The challenge in proving such a result is that the actual expenditure depends on true quantities, whereas the intervention is built based on \emph{estimated} quantities projected onto an ``estimated'' eigenspace of $\widehat{\bm{D}}$.  We begin with a technical result that will be key to controlling the differences between actual and estimated objects.   It relies on the Davis--Kahan theorem.

\begin{lem}  
    The norm $\Vert P_{\underline{L}(n)}^\perp P_{\widehat{L}(n)} \Vert$ converges to $0$ in probability as $n \to \infty$.
\label{lem:bound_projection_product}
\end{lem}

 \begin{proof} Standard eigenvalue perturbation bounds and the Davis--Kahan theorem guarantee that, for any $\eta > 0$, there exists $N_1(\eta)$ such that for all $n > N_1(\eta)$, with probability greater than or equal to $1-\eta$, two  properties hold.  First, every eigenvalue in the set $\Lambda(n)$, which is nonempty by construction (it always contains at least the largest eigenvalue of $\bm{D}$), has a corresponding eigenvalue within distance $O(b(n))$, and therefore in $\widehat{\Lambda}(n)$.  Second,  
    \begin{equation}
        \Vert P_{\underline{L}(n)}^\perp P_{\widehat{L}(n)} \Vert \leq \frac{2\Vert \bm{E} \Vert}{g}, \label{eq:use_gap}
    \end{equation}
     where $g$, the ``gap,'' is the minimum distance between some eigenvalue of $\widehat{\bm{D}}$ in $\widehat{\Lambda}(n)$ and some eigenvalue of $\bm{D}$ not contained in $\underline{\Lambda}(n)$.  By our choice of thresholds, this gap is at least $\widehat{M}(n)-\underline{M}(n)$, which is a fixed positive fraction of $b(n)$ and therefore bounded below by a positive constant multiple of $b(n)$ for all large $n$. Combining this with $\Vert \bm{E} \Vert = o_{\mathrm{p},\bm{\Theta}}(b(n))$, we obtain that the right-hand side of \eqref{eq:use_gap} converges to zero in probability. Thus, increasing $N_1(\eta)$ if necessary, we conclude the statement of the lemma.  
    \end{proof}

    \begin{lem} \label{lem:control_E_dot}
        As $n\to \infty$, the expenditure derivative $\dot{S}$ converges in probability to $1$.
    \end{lem}

    \begin{proof}
   
  Write $\bm{q}^0=\widehat{\bm{q}}^0-\bm{\varepsilon}$ and calculate
        \begin{align*}
    \dot{S} = \bm{\sigma} \cdot \bm{q}^0 &=  \frac{P_{\widehat{L}(n)}\widehat{\bm{q}}^0}{\Vert P_{\widehat{L}(n)}\widehat{\bm{q}}^0 \Vert^2}  \cdot (\widehat{\bm{q}}^0 - \bm{\varepsilon}) \\
    &= 1-\frac{P_{\widehat{L}(n)}\widehat{\bm{q}}^0}{\Vert P_{\widehat{L}(n)}\widehat{\bm{q}}^0 \Vert^2} \cdot P_{\widehat{L}(n)} \bm{\varepsilon}
    \end{align*}

     By the Cauchy--Schwarz inequality,  
     \[  
     \bigl|\dot{S}-1\bigr|  
    \le
     \frac{\Vert P_{\widehat{L}(n)}{\bm{\varepsilon}} \Vert}{\Vert P_{\widehat{L}(n)}\widehat{\bm{q}}^0 \Vert}.  
      \]  
     Condition (3) of \Cref{Def:SS} and \Cref{as:normalized_q} imply $\Vert P_{L(n)}\bm{q}^0 \Vert \geq \delta$.  By the same Davis--Kahan perturbation argument used above (now comparing $L(n)$ and $\widehat{L}(n)$ with gap $M(n)-\widehat{M}(n)=\Theta(b(n))$), we have $\Vert P_{\widehat{L}(n)}^\perp P_{L(n)} \Vert \xrightarrow{p} 0$.  Moreover,  
     \[  
     P_{L(n)}\bm{q}^0  
     =  
     P_{L(n)}P_{\widehat{L}(n)}\bm{q}^0  
     +  
     P_{L(n)}P_{\widehat{L}(n)}^\perp \bm{q}^0.  
      \]  
     Therefore,  
     \[  
    \delta
    \le
     \Vert P_{L(n)}\bm{q}^0 \Vert  
    \le
     \Vert P_{\widehat{L}(n)}\bm{q}^0 \Vert  
     +  
     \Vert P_{L(n)}P_{\widehat{L}(n)}^\perp\Vert\,\Vert \bm{q}^0 \Vert  
     =  
     \Vert P_{\widehat{L}(n)}\bm{q}^0 \Vert + o_{\mathrm{p}}(1),  
     \]  
     where the last equality uses \Cref{as:normalized_q}.  Hence  
     \[  
     \Vert P_{\widehat{L}(n)}\bm{q}^0 \Vert \ge \delta - o_{\mathrm{p}}(1).  
      \]  
     Since $\widehat{\bm{q}}^0=\bm{q}^0+\bm{\varepsilon}$,  
     \[  
     \Vert P_{\widehat{L}(n)}\widehat{\bm{q}}^0 \Vert  
     \ge \Vert P_{\widehat{L}(n)}\bm{q}^0 \Vert - \Vert P_{\widehat{L}(n)}\bm{\varepsilon} \Vert. 
      \]  
     So it remains to show that $\Vert P_{\widehat{L}(n)}{\bm{\varepsilon}} \Vert$ tends in probability to $0$.  
     
     For this, we state a lemma.  
     
     \begin{lem} \label{lem:generalized}  
     Fix an environment that has significant structure, subject to the normalization $\Vert \bm{q}^0 \Vert =1$.  Let $V(n)$ be a sequence of linear subspaces of $\mathbb{R}^n$, possibly random but measurable with respect to $\widehat{\bm{D}}$, such that $\dim V(n)/n \xrightarrow{p} 0$.  Then $\Vert P_{V(n)} \bm{\varepsilon} \Vert$ tends to $0$ in probability.  
     \end{lem}
    
    \begin{proof}
     If $\bm{\varepsilon}=\bm{0}$ almost surely, there is nothing to prove.  Otherwise let $\sigma_n^2:=\operatorname{Var}(\varepsilon_i)\in(0,\infty)$ and define $Z_i:=\varepsilon_i/\sigma_n$, so $Z:=(Z_1,\dots,Z_n)^\tr$ has i.i.d.\ entries with mean $0$ and variance $1$.  Since scaling cancels,  
     \[  
     \frac{\Vert P_{V(n)}\bm{\varepsilon}\Vert^2}{\Vert\bm{\varepsilon}\Vert^2}  
     =  
     \frac{\Vert P_{V(n)}Z\Vert^2}{\Vert Z\Vert^2}  
     =  
     \frac{\frac{1}{n}\Vert P_{V(n)}Z\Vert^2}{\frac{1}{n}\Vert Z\Vert^2}.  
      \]  
     Let $d(n):=\dim V(n)$.  Because $V(n)$ is measurable with respect to $\widehat{\bm{D}}$ and $\bm{\varepsilon}$ is independent of $\bm{E}$ by our signal assumption, the vector $Z$ is independent of $V(n)$.  Therefore,  
     \[  
     \mathbb{E}\!\left[\frac{1}{n}\Vert P_{V(n)}Z\Vert^2 \,\middle|\, V(n)\right]  
     =  
     \frac{1}{n}\operatorname{tr}\!\left(P_{V(n)}\right)  
     =  
     \frac{d(n)}{n}.  
      \]  
     Since $0\le d(n)/n\le 1$ and $d(n)/n\xrightarrow{p}0$, we also have $\mathbb{E}[d(n)/n]\to 0$.  Markov's inequality then gives, for every $\eta>0$,  
     \[  
     \mathbb{P}\!\left(\frac{1}{n}\Vert P_{V(n)}Z\Vert^2>\eta\right)  
    \le
     \frac{1}{\eta}\,\mathbb{E}\!\left[\frac{1}{n}\Vert P_{V(n)}Z\Vert^2\right]  
     =  
     \frac{1}{\eta}\,\mathbb{E}\!\left[\frac{d(n)}{n}\right]  
     \to 0.  
      \]  
     By the weak law of large numbers, $\frac{1}{n}\Vert Z\Vert^2=\frac{1}{n}\sum_{i=1}^n Z_i^2\xrightarrow{p}1$.  Slutsky's theorem implies $\Vert P_{V(n)}\bm{\varepsilon}\Vert/\Vert\bm{\varepsilon}\Vert\xrightarrow{p}0$.  Under significant structure and the normalization $\Vert \bm{q}^0 \Vert =1$, condition (1) of \Cref{Def:SS} implies $\Vert \bm{\varepsilon} \Vert=O_{\mathrm{p},\bm{\Theta}}(1)$, and therefore $\Vert P_{V(n)}\bm{\varepsilon}\Vert\xrightarrow{p}0$.  
     \end{proof}
    
     Applying \Cref{lem:generalized} with $V(n)=\widehat{L}(n)$, and using that\footnote{ Weyl's inequality implies that with probability $1-o(1)$, every eigenvalue counted by $\widehat{L}(n)$ comes from a true eigenvalue of $\bm{D}$ whose absolute value exceeds $b(n)/2$. Since $\bm{D}$ is negative semidefinite with trace $-n$, there are at most $2n/b(n)=o(n)$ such eigenvalues.}  $\dim \widehat{L}(n)/n \xrightarrow{p}0$, we obtain $\Vert P_{\widehat{L}(n)}{\bm{\varepsilon}} \Vert \xrightarrow{p} 0$.  Combining this with  what we stated just before the lemma,   \[  
     \Vert P_{\widehat{L}(n)}\widehat{\bm{q}}^0 \Vert  
     \ge \Vert P_{\widehat{L}(n)}\bm{q}^0 \Vert - \Vert P_{\widehat{L}(n)}\bm{\varepsilon} \Vert, 
      \]   shows that $\Vert P_{\widehat{L}(n)}\widehat{\bm{q}}^0 \Vert \geq \delta/2$ with probability $1-o(1)$.  Returning to the bound on $|\dot{S}-1|$, we conclude that $\dot{S}\xrightarrow{p}1$.  
     \end{proof}

Now, to prove the theorem, we use  \Cref{lem:welfare-pass} to write:
    \begin{equation}
    \dot{W} = \underbrace{\sum_{\lambda_\ell \in \underline{\Lambda}(n)} (\bm{u}^\ell \cdot \bm{q}^0)(\bm{u}^\ell \cdot \bm{\sigma})\frac{|\lambda_\ell|}{1+|\lambda_\ell|}}_{\dot{W}_M} + \underbrace{\sum_{\lambda_\ell \notin \underline{\Lambda}(n)} (\bm{u}^\ell \cdot \bm{q}^0)(\bm{u}^\ell \cdot \bm{\sigma})\frac{|\lambda_\ell|}{1+|\lambda_\ell|}}_{\dot{W}_R}. \label{eq:welfare_split}
    \end{equation}
where we have divided the expression into a main (M) part and the rest (R). A similar (but simpler) decomposition applies to expenditure
\begin{equation}
    \dot{S} =\underbrace{\sum_{\lambda_\ell \in \underline{\Lambda}(n)} (\bm{u}^\ell \cdot \bm{q}^0)(\bm{u}^\ell \cdot \bm{\sigma})}_{\dot{S}_M} + \underbrace{\sum_{\lambda_\ell \notin \underline{\Lambda}(n)} (\bm{u}^\ell \cdot \bm{q}^0)(\bm{u}^\ell \cdot \bm{\sigma})}_{\dot{S}_R}. \label{eq:expenditure_split}
    \end{equation}

The proof can be completed with two further lemmas.

   \begin{lem} As $n\to \infty$, both 
    $\dot{W}_{R} $ and $  \dot{S}_{R} $ converge in probability to $0$.  \label{lem:small_remainders}
    \end{lem}
    \begin{proof} This follows from \Cref{lem:bound_projection_product} and the fact that by construction, $\bm{\sigma} \in \widehat{L}(n)$. \end{proof}    
    
     \begin{lem}  As $n\to \infty$,   
    $$\dot{W}_M - \dot{S}_M \xrightarrow{p} 0.$$ \label{lem:WM_lower_bound} 
    \end{lem}
    \begin{proof} Using the expressions above, write the difference 
     $$ | \dot{W}_M - \dot{S}_M | \leq {\sum_{\lambda_\ell \in \underline{\Lambda}(n)} \frac{1}{1+|\lambda_\ell|} \left|(\bm{u}^\ell \cdot \bm{q}^0)(\bm{u}^\ell \cdot \bm{\sigma})\right|}.$$  
     Note that for all $\ell \in \underline{\Lambda}(n)$, $|\lambda_\ell| \geq \underline{M}(n)$, so $\frac{1}{1+|\lambda_\ell|} \leq \frac{1}{\underline{M}(n)}$.  Therefore,    
    $$
       | \dot{W}_M - \dot{S}_M | \leq \frac{1}{\underline{M}(n)} \sum_{\lambda_\ell \in \underline{\Lambda}(n)} |(\bm{u}^\ell \cdot \bm{q}^0)(\bm{u}^\ell \cdot \bm{\sigma})|.  $$
      By Cauchy--Schwarz for each $\ell$ and then summing, we have
    $$
       | \dot{W}_M - \dot{S}_M | \leq \frac{1}{\underline{M}(n)} \sum_{\lambda_\ell \in \underline{\Lambda}(n)} \Vert P_{\bm{u}^\ell} \bm{q}^0 \Vert \Vert P_{\bm{u}^\ell} \bm{\sigma} \Vert.$$
      $$ | \dot{W}_M - \dot{S}_M | \leq \frac{1}{\underline{M}(n)} \Vert P_{\underline{L}(n)} \bm{q}^0 \Vert \Vert P_{\underline{L}(n)} \bm{\sigma} \Vert.$$
      Since $\Vert P_{\underline{L}(n)} \bm{q}^0 \Vert \leq \Vert \bm{q}^0 \Vert = O(1)$ by \Cref{as:normalized_q}, and $\Vert P_{\underline{L}(n)} \bm{\sigma} \Vert \leq \Vert \bm{\sigma} \Vert = O_{\mathrm{p}}(1)$ by the denominator lower bound in the proof of \Cref{lem:control_E_dot}, the result follows.  \end{proof}

We can put everything together.  \Cref{lem:control_E_dot} gives that $\dot{S} \xrightarrow{p} 1$. Combining this with (\ref{eq:expenditure_split}) and \Cref{lem:small_remainders} (which says that $\dot{S}_{R} \xrightarrow{p} 0$) we see that $\dot{S}_M \xrightarrow{p} 1$. Then using \Cref{lem:WM_lower_bound}, we find that $\dot{W}_M \xrightarrow{p} 1$. Another application of \Cref{lem:small_remainders} gives that $\dot{W}_{R} \xrightarrow{p} 0$, so that $\dot{W} \xrightarrow{p} 1$. The claim about the effect on $\dot{C}$ follows immediately from \Cref{prop:Welfare CI}.

Finally, we consider the effect on individual consumer surpluses. For any consumer $h$,
    \begin{equation}
     \dot{C}^h = -\bm{q}^h \cdot \dot{\bm{p}} = -\sum_{\ell=1}^n (\bm{u}^\ell \cdot \bm{q}^h)(\bm{u}^\ell \cdot \dot{\bm{p}})  
    \end{equation}
    
     Using \Cref{prop:eig}:  
    \begin{equation}
     \bm{u}^\ell \cdot \dot{\bm{p}} = -\frac{1}{1+|\lambda_\ell|}(\bm{u}^\ell \cdot \bm{\sigma})  
    \end{equation}
    Then by an argument very similar to the proof of \Cref{lem:small_remainders}, we conclude that $\dot{C}^h \xrightarrow{p} 0$.

\subsection{Expectation guarantees}
\label{app:expectations} Here we upgrade our in-probability guarantees on the performance of the intervention to expectation guarantees.

\medskip

\paragraph{A bounded modification of the intervention rule}

The intervention in the proof of \Cref{Th:Main} is (up to scaling to hit spending $s$)
of the form
$$
\bm{\sigma}
=
P_{\widehat{L}(n)}\widehat{\bm{q}}^0/\|P_{\widehat{L}(n)}\widehat{\bm{q}}^0\|^2,
$$
which can be large on realizations where $\|P_{\widehat{L}(n)}\widehat{\bm{q}}^0\|$ is very small.
Under significant structure, these realizations are asymptotically negligible uniformly over states.
For expectation statements, it is convenient to enforce boundedness deterministically.

Fix the target spending level $s>0$, and let $\delta>0$ be the projection lower bound
in condition (3) of \Cref{Def:SS} after the normalization $\|\bm{q}^0\|=1$
from \Cref{as:normalized_q}. Define the truncated rule
\[
\bm{R}^{\mathrm{tr}}(\widehat{\bm{\theta}})
=
s\cdot
\frac{P_{\widehat{L}(n)}\widehat{\bm{q}}^0}
{\max\left\{\|P_{\widehat{L}(n)}\widehat{\bm{q}}^0\|^2,\;(\delta/4)^2\right\}}.
\]
Then for all signal realizations,
\begin{equation}
\label{eq:sigma_uniform_bound_short}
\|\bm{R}^{\mathrm{tr}}(\widehat{\bm{\theta}})\|\le \frac{4s}{\delta}.
\end{equation}
Moreover, by the same argument used in the proof of \Cref{Th:Main}, we have
$\|P_{\widehat{L}(n)}\widehat{\bm{q}}^0\|\ge \delta/2$ with probability $1-o(1)$
uniformly over $\bm{\theta}$, so the truncation event
$\{\|P_{\widehat{L}(n)}\widehat{\bm{q}}^0\|<\delta/4\}$ has probability $o(1)$
uniformly over states. Therefore, $\bm{R}^{\mathrm{tr}}$ coincides with the original rule
except on a vanishing-probability event, so replacing the rule by $\bm{R}^{\mathrm{tr}}$
does not change any $\epsilon$--robustness conclusions.

\medskip

\paragraph{Uniform boundedness of the welfare objects}

Let $\bm{\sigma}=\bm{R}^{\mathrm{tr}}(\widehat{\bm{\theta}})$. Since $\|\bm{q}^0\|=1$,
\[
|\dot S_{\bm{\sigma}}| = |\bm{q}^0\cdot \bm{\sigma}|\le \|\bm{\sigma}\|.
\]
By \Cref{prop:eig}, the pass-through multipliers that map $\bm{\sigma}$ into first-order
price changes have absolute values in $[0,1]$ in every eigenvector; equivalently, there is a uniform
operator-norm bound $\|\dot{\bm{p}}_{\bm{\sigma}}\|\le \|\bm{\sigma}\|$. Hence
$
|\dot C_{\bm{\sigma}}|
=|\bm{q}^0\cdot \dot{\bm{p}}_{\bm{\sigma}}|
\le \|\dot{\bm{p}}_{\bm{\sigma}}\|
\le \|\bm{\sigma}\|
$.
Using $\dot W_{\bm{\sigma}}=\dot P_{\bm{\sigma}}+\dot C_{\bm{\sigma}}-\dot S_{\bm{\sigma}}$
and the identity $\dot P_{\bm{\sigma}}=2(\dot S_{\bm{\sigma}}-\dot C_{\bm{\sigma}})$ from (\ref{eq:Pareto_Identity}) 
there is a constant $K$  such that
\[
|\dot S_{\bm{\sigma}}|,\;|\dot C_{\bm{\sigma}}|,\;|\dot P_{\bm{\sigma}}|,\;|\dot W_{\bm{\sigma}}|
\le
K\|\bm{\sigma}\|
\le
K\cdot \frac{4s}{\delta}
\qquad\text{almost surely, uniformly over $\bm{\theta}$.}
\]
Thus each welfare error term appearing in \Cref{Th:Main} is uniformly bounded under $\bm{R}^{\mathrm{tr}}$.

\medskip

\paragraph{Conclusion: High-probability guarantees imply expectation guarantees}

\begin{prop}[Expectation upgrade]
\label{prop:expected_performance_short}
Maintain the assumptions of \Cref{Th:Main} and fix spending $s>0$. Replace the intervention rule
in the proof of \Cref{Th:Main} by the bounded modification $\bm{R}^{\mathrm{tr}}$ above, and let
$\bm{\sigma}=\bm{R}^{\mathrm{tr}}(\widehat{\bm{\theta}})$.
Then 
uniformly over $\bm{\theta}\in\bm{\Theta}(n)$,
\[
\mathbb{E}\big[|\dot S_{\bm{\sigma}}-s|\big]\to 0,
\qquad
\mathbb{E}\big[|\dot C_{\bm{\sigma}}|\big]\to 0,
\qquad
\mathbb{E}\big[|\dot P_{\bm{\sigma}}-2s|\big]\to 0,
\qquad
\mathbb{E}\big[|\dot W_{\bm{\sigma}}-s|\big]\to 0.
\]
\end{prop}

\begin{proof}
As noted above, $\bm{R}^{\mathrm{tr}}$ differs from the original rule only on a vanishing-probability event
uniformly over $\bm{\theta}$, so the same high-probability statements remain valid.
For each displayed term, let $X_n(\bm{\theta})$ denote the corresponding error random variable (e.g.\ $\dot S_{\bm{\sigma}}-s$).
The proof of \Cref{Th:Main} gives $X_n=o_{\mathrm{p},\bm{\Theta}}(1)$, and the previous subsection gives a uniform almost-sure bound
$|X_n(\bm{\theta})|\le M$ for some finite $M$ independent of $n$ and $\bm{\theta}$.
For such random variables, the $L^1$ convergence of $X_n$ follows from the corresponding in-probability convergence statement proved in \Cref{Th:Main}.
\end{proof}

\medskip

Because all conclusions are uniform over $\bm{\theta}$, they also hold under any prior over states by iterated expectations.

\section{Proof of \Cref{Prop:LackOfAS}}\label{app:lackofas_proof}

All arguments in this proof are in normalized coordinates. The constructions below have $D_{ii}=-1$, so underlined and unnormalized variables coincide. We use Hadamard constructions to hide, respectively, the welfare-relevant demand direction and the consumer-surplus-relevant quantity direction.

We use a common Hadamard setup. Work along the subsequence $n=2^m$, and let $\{\bm{u}^1,\ldots,\bm{u}^n\}$ be an orthonormal Hadamard basis with $\bm{u}^n=\bm{1}/\sqrt{n}$. Thus $(u_i^\ell)^2=1/n$ for every $i$ and $\ell$. Whenever eigenvalues $\lambda_1,\ldots,\lambda_n$ satisfy $\sum_{\ell=1}^n\lambda_\ell=-n$, the matrix
\[
\bm{D}=\sum_{\ell=1}^n \lambda_\ell\bm{u}^\ell(\bm{u}^\ell)^\tr
\]
has diagonal entries equal to $-1$. If all eigenvalues are nonpositive, then $\bm{D}$ is negative semidefinite. The bounded-own-price and recoverable-diagonal assumptions then hold immediately.

\subsection{Part 1: No significant structure}

The quantity vector in this construction is common across states and is observed without error. Because the inequalities in \Cref{Prop:LackOfAS} are homogeneous in the intervention, any signal-measurable rule with positive expenditure can be rescaled, using the observed $\bm{q}^0$, to satisfy
\[
\dot{S}_{\bm{\sigma}}=\bm{q}^0\cdot \bm{\sigma}=1.
\]
Under this normalization, part~(1) requires $\dot{W}_{\bm{\sigma}}\ge \varepsilon$, because $\dot{C}+\dot{P}=\dot{S}+\dot{W}$.

Let $z_n=\lfloor n^{1/4}\rfloor$ and define the common quantity vector
\[
\bm{q}^0=2z_n\bm{u}^n+\sum_{k=1}^{z_n}\bm{u}^k .
\]
This vector is strictly positive for all large $n$, and its only nonzero Hadamard coordinates are $\bm{u}^n\cdot\bm{q}^0=2z_n$ and $\bm{u}^k\cdot\bm{q}^0=1$ for $k\le z_n$.
The hidden state will identify one of these $z_n$ quantity-loaded Hadamard directions as the high-pass-through direction. Because the signal does not reveal this index, any signal-measurable rule must choose the same intervention without knowing which of the $z_n$ directions is valuable.

For each hidden state $j\in\{1,\ldots,z_n\}$, set
\[
\tau_n:=
\frac{1}{z_n}
\left(
\frac{n^{1/2}}{1+n^{1/2}}
+(z_n-1)\frac{n^{-2}}{1+n^{-2}}
\right),
\]
so $\tau_n$ is the average welfare pass-through weight across the hidden directions in a fixed state: one direction has pass-through close to $1$, and the other $z_n-1$ directions have pass-through close to $0$. Thus $\tau_n\le 1/z_n+n^{-2}\to0$.
Let the eigenvalues of $\bm{D}^{(j)}$ be
\[
\lambda_j^{(j)}=-n^{1/2},\qquad
\lambda_k^{(j)}=-n^{-2}\quad\text{for } k\le z_n,\ k\ne j,
\]
\[
\lambda_n^{(j)}=-\frac{\tau_n}{1-\tau_n},
\]
and, for all $\ell\notin\{1,\ldots,z_n,n\}$,
\[
\lambda_\ell^{(j)}
=
\frac{-n+n^{1/2}+(z_n-1)n^{-2}+\tau_n/(1-\tau_n)}
{n-z_n-1}.
\]
The eigenvalue in the common direction $\bm{u}^n$ is chosen so that its pass-through weight is also $\tau_n$, since $|\lambda_n^{(j)}|/(1+|\lambda_n^{(j)}|)=\tau_n$.

The last value makes the trace equal to $-n$ and converges to $-1$. Hence all eigenvalues are negative for large $n$.

The signal is deterministic and common across states:
\[
\widehat{\bm{D}}=-\bm{I},\qquad \widehat{\bm{q}}^0=\bm{q}^0.
\]
Thus a rule chooses the same normalized intervention $\bm{\sigma}$ in every state. Since $\bm{q}^0$ is orthogonal to all remaining Hadamard directions, the normalization gives
\[
(\bm{u}^n\cdot\bm{q}^0)(\bm{u}^n\cdot\bm{\sigma})
+\sum_{k=1}^{z_n}(\bm{u}^k\cdot\bm{q}^0)(\bm{u}^k\cdot\bm{\sigma})
=1.
\]
By \Cref{eq:intuition1}, in state $j$,
\[
\dot{W}_{\bm{\sigma}}^{(j)}
=
\frac{n^{1/2}}{1+n^{1/2}}
(\bm{u}^j\cdot\bm{q}^0)(\bm{u}^j\cdot\bm{\sigma})
+\frac{n^{-2}}{1+n^{-2}}
\sum_{\substack{k=1\\ k\ne j}}^{z_n}
(\bm{u}^k\cdot\bm{q}^0)(\bm{u}^k\cdot\bm{\sigma})
+\tau_n(\bm{u}^n\cdot\bm{q}^0)(\bm{u}^n\cdot\bm{\sigma}) .
\]
Averaging over $j=1,\ldots,z_n$ and using the definition of $\tau_n$ yields
\[
\frac{1}{z_n}\sum_{j=1}^{z_n}\dot{W}_{\bm{\sigma}}^{(j)}
=
\tau_n
\left[
\sum_{k=1}^{z_n}(\bm{u}^k\cdot\bm{q}^0)(\bm{u}^k\cdot\bm{\sigma})
+(\bm{u}^n\cdot\bm{q}^0)(\bm{u}^n\cdot\bm{\sigma})
\right]
=
\tau_n .
\]
Therefore some state $j$ satisfies $\dot{W}_{\bm{\sigma}}^{(j)}\le \tau_n$. Since $\tau_n\to0$, for every fixed $\varepsilon\in(0,1/2)$ and all sufficiently large $n$, the normalized target $\dot{W}_{\bm{\sigma}}\ge\varepsilon$ fails conditional on that state (with probability $1$, since the signal is deterministic).

It remains only to verify that this environment lacks significant structure. In state $j$, the observation error $\bm{E}^{(j)}=-\bm{I}-\bm{D}^{(j)}$ has operator norm of order $n^{1/2}$. If a candidate threshold $b(n)\to\infty$ is not asymptotically larger than $n^{1/2}$, condition~(2) of \Cref{Def:SS} fails. If $b(n)/n^{1/2}\to\infty$, then the large-eigenvalue space is eventually empty, so condition~(3) fails. Hence significant structure does not hold.
\subsection{Part 2: Significant structure present}

We now construct an environment with significant structure in which the large-eigenvalue direction is perfectly recoverable, but the quantity signal does not reveal a small component of quantities that is relevant for consumer surplus.

Let $\alpha_n=n^{-1}$. Work with the same Hadamard basis as above. Set
\[
\eta_n:=n^{-2},\qquad
\lambda_\ell=-\eta_n\quad(\ell=1,\ldots,n-1),
\qquad
\lambda_n=-n+(n-1)\eta_n,
\]
and define
\[
\bm{D}
:=
\sum_{\ell=1}^n \lambda_\ell\bm{u}^\ell(\bm{u}^\ell)^\tr .
\]
The eigenvalues sum to $-n$, so the diagonal entries of $\bm{D}$ are $-1$. All eigenvalues are negative. Thus $\bm{D}$ is negative semidefinite, the bounded-own-price assumption holds, and the diagonal is exactly recoverable.

The state space consists of the quantity vectors
\[
\bm{q}^{j,s}:=\bm{u}^n+s\alpha_n\bm{u}^j,
\qquad
j\in\{1,\ldots,n-1\},\quad s\in\{-1,1\}.
\]
Because $\bm{u}^n=\bm{1}/\sqrt{n}$ and every coordinate of $\bm{u}^j$ is $\pm1/\sqrt{n}$, all quantities are strictly positive:
\[
q_i^{j,s}\in\{(1-\alpha_n)/\sqrt{n},(1+\alpha_n)/\sqrt{n}\}.
\]
The signal about $\bm{D}$ is exact, $\widehat{\bm{D}}=\bm{D}$. The quantity signal is
\[
\widehat{\bm{q}}^0=\bm{q}^{j,s}+\bm{\varepsilon},
\]
where, conditional on the state $(j,s)$, the coordinates $\varepsilon_i$ are i.i.d.\ $N(0,1/n)$ and are independent of the matrix observation noise. Moreover,
\[
\|\bm{q}^{j,s}\|=\sqrt{1+\alpha_n^2},
\qquad
\|\bm{\varepsilon}\|^2=(1/n)\chi_n^2.
\]
Hence $\|\bm{\varepsilon}\|=O_{\mathrm{p},\bm{\Theta}}(\|\bm{q}^0\|)$.

The environment has significant structure. Take $b(n)=n/2$. Since $\lambda_n=-n+o(1)$ and all other eigenvalues have absolute value $\eta_n=o(1)$, the large-eigenvalue space is
\[
\mathcal{L}(\bm{D},b(n))=\operatorname{span}\{\bm{u}^n\}
\]
for all large $n$. The matrix observation noise is zero, so condition~(2) of \Cref{Def:SS} holds. Finally,
\[
\big\|P_{\mathcal{L}(\bm{D},b(n))}\bm{q}^{j,s}\big\|
=
|\bm{u}^n\cdot\bm{q}^{j,s}|
=1
\ge
\frac{1}{\sqrt{1+\alpha_n^2}}\|\bm{q}^{j,s}\|,
\]
so condition~(3) holds with any fixed $\delta<1$ for all large $n$.

It remains to show that a fixed positive consumer-surplus incidence per dollar of expenditure cannot be robustly guaranteed. The intuition is that consumer surplus can be created only through a hidden direction $\bm{u}^j$, whose sign the quantity signal cannot resolve. We make this precise in two steps: at any realized signal, whatever intervention the rule selects, the hidden direction can take the sign that opposes it, producing a state in which the rule fails---generating either no positive expenditure or a consumer-surplus incidence per dollar below the target level. The two such candidate states differ only in this sign and are statistically almost indistinguishable, so the rule fails with probability bounded away from zero.

Fix any $\kappa\in(0,1/2)$.\footnote{This proves the result for every $\kappa>0$: if $\kappa\ge 1/2$, failure to robustly achieve the inequality for any fixed $\kappa'\in(0,1/2)$ implies failure to robustly achieve the stronger inequality for $\kappa$.} Say a rule \emph{fails} in a state if its expenditure there is nonpositive, or is positive with $\dot{C}<\kappa\dot{S}$. Consider an arbitrary rule and a realized signal, and let $\bm{\sigma}$ be the selected intervention, with components
\[
a:=\bm{u}^n\cdot\bm{\sigma}
\qquad\text{and}\qquad
b_j:=\bm{u}^j\cdot\bm{\sigma}
\]
along the recoverable direction $\bm{u}^n$ and a hidden direction $\bm{u}^j$. Since $\bm{q}^{j,s}$ is orthogonal to every other eigendirection, \Cref{prop:eig,lem:welfare-pass} give
\[
\dot{S}_{\bm{\sigma}}^{j,s}
=
a+s\alpha_n b_j,
\qquad
\dot{C}_{\bm{\sigma}}^{j,s}
=
\frac{a}{1+|\lambda_n|}
+
\frac{s\alpha_n b_j}{1+\eta_n}.
\]
Fix $j\le n-1$ and pick $s$ with $s b_j\le0$, so the hidden direction opposes the intervention: $s\alpha_n b_j=-\alpha_n|b_j|$. Then $\dot{S}_{\bm{\sigma}}^{j,s}=a-\alpha_n|b_j|$, and since $\alpha_n|b_j|\ge0$ and $(1+\eta_n)^{-1}\ge(1+|\lambda_n|)^{-1}$,
\[
\dot{C}_{\bm{\sigma}}^{j,s}
=
\frac{a}{1+|\lambda_n|}
-
\frac{\alpha_n|b_j|}{1+\eta_n}
\le
\frac{a-\alpha_n|b_j|}{1+|\lambda_n|}
=
\frac{\dot{S}_{\bm{\sigma}}^{j,s}}{1+|\lambda_n|}.
\]
If $\dot{S}_{\bm{\sigma}}^{j,s}\le0$, the rule delivers no positive expenditure; if $\dot{S}_{\bm{\sigma}}^{j,s}>0$, then $\dot{C}_{\bm{\sigma}}^{j,s}/\dot{S}_{\bm{\sigma}}^{j,s}\le1/(1+|\lambda_n|)<\kappa$ for all large $n$, since $|\lambda_n|\to\infty$. Either way the rule fails in state $(j,s)$. Writing $F_{j,s}$ for the set of signals at which the rule fails in state $(j,s)$, we have shown that
\[
F_{j,1}\cup F_{j,-1}=\widehat{\bm{\Theta}}
\qquad\text{for all large } n.
\]

We now translate this pointwise ambiguity into the i.i.d.\ signal experiment. In state $(j,s)$ the quantity signal is distributed as $P_{j,s}=N(\bm{u}^n+s\alpha_n\bm{u}^j,\,n^{-1}\bm{I})$, so
\[
D_{\mathrm{KL}}(P_{j,1}\,\|\,P_{j,-1})
=
\frac{1}{2}(2\alpha_n\bm{u}^j)^\tr(n\bm{I})(2\alpha_n\bm{u}^j)
=
2n\alpha_n^2
=
\frac{2}{n},
\]
and Pinsker's inequality gives $\|P_{j,1}-P_{j,-1}\|_{\mathrm{TV}}\le n^{-1/2}$. Because $F_{j,1}\cup F_{j,-1}=\widehat{\bm{\Theta}}$,
\[
P_{j,1}(F_{j,1})+P_{j,-1}(F_{j,-1})
\ge
1-\|P_{j,1}-P_{j,-1}\|_{\mathrm{TV}}
\ge
1-n^{-1/2}.
\]
Hence some state in the pair $\{(j,1),(j,-1)\}$ has failure probability at least $(1-n^{-1/2})/2$, which is bounded away from zero. Therefore, no intervention rule can robustly guarantee any fixed positive consumer-surplus incidence per dollar of expenditure.

\subsection{Expected-performance implication}\label{app:lackofas_expected_performance}

We record the expected-performance statement used in the remark following \Cref{Prop:LackOfAS}. Consider the environment from part~(1). The signal is deterministic and common across the hidden states. Thus any rule chooses the same intervention $\bm{\sigma}$ in every hidden state. If the rule selects positive expenditure, the homogeneity of the criterion lets us normalize $\dot{S}_{\bm{\sigma}}=\bm{q}^0\cdot\bm{\sigma}=1$. The averaging calculation above gives
\[
\frac{1}{z_n}\sum_{j=1}^{z_n}\dot{W}_{\bm{\sigma}}^{(j)}
=\tau_n,
\qquad \tau_n\to0.
\]
Therefore some state $j$ satisfies $\dot{W}_{\bm{\sigma}}^{(j)}\leq\tau_n$. Since the signal is deterministic in this state,
\[
\mathbb{E}_{\varphi_{\bm{\theta}^{(j)}}(n)}
\left[
\dot{C}_{\bm R(\widehat{\bm{\theta}})}
+\dot{P}_{\bm R(\widehat{\bm{\theta}})}
-\dot{S}_{\bm R(\widehat{\bm{\theta}})}
\right]
=
\dot{W}_{\bm{\sigma}}^{(j)}
\leq \tau_n.
\]
Since $\dot{W}=\dot{C}+\dot{P}-\dot{S}$, it follows that
\[
\inf_{\bm{\theta}\in\bm{\Theta}(n)}
\mathbb{E}_{\varphi_{\bm{\theta}}(n)}
\left[
\dot{C}
+\dot{P}
-\dot{S}
\right]
\leq \tau_n.
\]
For every fixed $\varepsilon>0$, $\tau_n\leq\varepsilon$ for all sufficiently large $n$, which proves the displayed claim in the remark.
\newpage

\section{Hadamard basis and simulation construction}\label{app:hadamard_sim}

This appendix formalizes the Hadamard-basis construction used in the Monte Carlo
illustration in \Cref{Sec:Illustration}. The goal is to define an orthonormal
basis with entries $\pm 1/\sqrt{n}$, build a Slutsky matrix with a single
``spike'' eigenvalue, and obtain a transparent block structure once goods are
ordered by the spike eigenvector. This makes it easy to separate signal
strength from noise and to interpret the welfare incidence plots.

\paragraph{Hadamard basis.}
We work along the subsequence $n=2^m$ so that an $n\times n$ Hadamard matrix
exists. Define the Sylvester recursion
\[
\bm{H}_1=\begin{pmatrix}1\end{pmatrix},
\qquad
\bm{H}_{2n}=
\begin{pmatrix}
\bm{H}_n & \bm{H}_n\\
\bm{H}_n & -\bm{H}_n
\end{pmatrix}.
\]
Then $\bm{H}_n \bm{H}_n^\top = n\,\bm{I}$ and the normalized columns
$\bm{u}^\ell=\bm{H}_n[:,\ell]/\sqrt{n}$ form an orthonormal basis
$\{\bm{u}^1,\ldots,\bm{u}^n\}$ with entries $\pm 1/\sqrt{n}$.
We index so that the uniform vector is $\bm{u}^n=\bm{1}/\sqrt{n}$.

\paragraph{Eigenvalue spectrum and demand matrix.}
To match the construction in the main text, we take the spike direction to be
$\bm{u}^1$ and choose parameters
$\lambda>0$ and $\varepsilon>0$. Define eigenvalues
\[
\mu_n=-\varepsilon,\qquad
\mu_1=-(1+\lambda),\qquad
\mu_{\text{bulk}}=\frac{-n+\varepsilon+(1+\lambda)}{n-2},
\]
and set $\mu_\ell=\mu_{\text{bulk}}$ for all $\ell\notin\{1,n\}$.
This guarantees $\sum_{\ell=1}^n \mu_\ell=-n$, so the diagonal entries of
\[
\underline{\bm{D}}=\sum_{\ell=1}^n \mu_\ell\,\bm{u}^\ell(\bm{u}^\ell)^\top
\]
equal $-1$ exactly (since $(u_i^\ell)^2=1/n$ for all $i,\ell$), and
$\underline{\bm{D}}$ is negative semidefinite. The spike eigenvalue
$\mu_1=-(1+\lambda)$ controls signal strength; the aggregate eigenvalue
$\mu_n=-\varepsilon$ is near zero; and the bulk compensates to satisfy the
trace condition.

\paragraph{Status quo quantities.}
Let $v^1=\sqrt{n}\,\bm{u}^1\in\{\pm1\}^n$ and define
\[
\underline{\bm{q}}^0=\bm{1}+\eta\,v^1
=\sqrt{n}\,(\bm{u}^n+\eta\,\bm{u}^1),
\]
so $q_i^0\in[1-\eta,1+\eta]$ and $\underline{\bm{q}}^0$ has a small but
nonzero projection on the spike direction. This choice isolates the
spike's contribution to welfare pass-through while keeping quantities
strictly positive.

\paragraph{Block structure interpretation.}
Ordering goods by the sign of $u_i^1$ reveals the two-block structure
 of the outer product $\bm{u}^1(\bm{u}^1)^\top$:  
within-block entries are positive and across-block entries are negative.
  Multiplying by $\mu_1<0$ reverses those signs in the rank-one term  
 $\mu_1\bm{u}^1(\bm{u}^1)^\top$, so within-block entries  
 are negative and across-block entries are positive.  
  In other words, the Hadamard-basis construction is a two-block model up to a  
 permutation of goods; sorting by the sign of $u^1$ makes the outer  
 product exactly a $2\times2$ block matrix with constant entries $\pm 1/n$.  
  When the recovered eigenvector $\widehat{\bm{u}}$ aligns with $\bm{u}^1$,  
the estimated outer product $\widehat{\bm{u}}\widehat{\bm{u}}^\top$ displays
the same block pattern; below the recovery threshold it appears scrambled.

\paragraph{Signal and intervention.}
The authority observes $\widehat{\underline{\bm{D}}}=\underline{\bm{D}}+\underline{\bm{E}}$
with $\underline{\bm{E}}=\bm{E}$ a Wigner matrix (zero diagonal, i.i.d.\ $\mathcal{N}(0,1)$
off-diagonal). Let $\widehat{\bm{u}}$ be the leading eigenvector of
$-\widehat{\underline{\bm{D}}}$ and orient it so that
 $\widehat{\bm{u}}\cdot\underline{\bm{q}}^0>0$. The intervention  
\[
\underline{\bm{\sigma}}=s\cdot\frac{\widehat{\bm{u}}}{\widehat{\bm{u}}\cdot\underline{\bm{q}}^0}
\]
spends $s$ and yields welfare and incidence outcomes computed via the
 eigen-decomposition formula in \Cref{eq:intuition1}.  In this Wigner-noise  
 simulation, the key threshold is the ratio $\lambda/\sqrt{n}$, which governs  
 eigenvector recovery and therefore the transition in the welfare-incidence  
 panels.  
 
\section{Sampling errors: Theoretical foundations \\ and empirical tests}
\label{ap:sampling}

\subsection{An explicit sampling procedure}
The definition of significant structure (\Cref{Def:SS}), as well as the recoverable diagonal assumption (\Cref{as:error_assumption}), impose conditions on the distribution of the error matrix $\underline{\bm{E}}$. We present a sampling procedure and an associated estimator for the normalized demand matrix $\underline{\bm{D}}$ that satisfies these conditions. As in \Cref{sec:norm}, we do not assume any normalization to begin with (since the authority cannot assume the market is already normalized.)

We assume that all households share a single representative utility function for goods and that the number of households exceeds $n^2$, where $n$ is the number of firms. Moreover, we assume that, for all $i$ and $j$, the entries $D_{ij}$ are uniformly bounded, i.e., $|D_{ij}| \leq A <\infty$ for some positive constant $A$. Also recall that part (a) of \Cref{Ass:Asymptotics} requires that the true diagonal entries satisfy uniform bounds: there exist universal constants $0<d_{\min}\le d_{\max}<\infty$ such that $-D_{ii} \in [d_{\min},d_{\max}]$ for all $i$. The constants $A$, $d_{\min}$, and $d_{\max}$ do not depend on $i,j$, or $n$.

For each product pair $(i,j)$, the authority samples a distinct household---call it $h(i,j)$---with the representative preferences; it performs a demand experiment to obtain unbiased estimates of the relevant entries. For off-diagonal entries $i\neq j$, we write
\[
\widehat{D}^{h(i,j)}_{ij} = D_{ij} + F_{ij}^{h(i,j)},
\]
where the $F_{ij}^{h(i,j)}$ are mean-zero, independent across unordered pairs $\{i,j\}$, uniformly bounded, and satisfy $\operatorname{Var}(F_{ij}^{h(i,j)})\le\overline{V}$ for some constant $\overline{V}$ that does not depend on $i,j$, or $n$. Dropping the superscript for notational convenience, we define the corresponding pooled off-diagonal estimators by
\[
\widehat{D}_{ij} = D_{ij} + F_{ij},\qquad i\neq j,
\]
where the $F_{ij}$ inherit the same distributional properties.

Turning to estimates of $D_{ii}$, for each ordered pair $(i,j)$ we write
\[
\widehat{D}^{h(i,j)}_{ii} = D_{ii} + G_{ij}^{h(i,j)}.
 \]  
 We assume that the $G_{ij}^{h(i,j)}$ are mean-zero, independent across $(i,j)$, uniformly bounded, and independent of $F_{i'j'}$.  We also assume that the variances of the $G_{ij}^{h(i,j)}$ are bounded by $\overline{V}$, uniformly in $i,j$, and $n$.  Averaging over $j$, we define  
\[
\widehat{D}_{ii} = \frac{1}{n} \sum_{j} \widehat{D}^{h(i,j)}_{ii}
  = D_{ii} + G_{ii}, \qquad   G_{ii} = \frac{1}{n}\sum_j G_{ij}^{h(i,j)}.
  \]  
For notational convenience, we now define the diagonal entries of the matrix $\bm{F}$ by $F_{ii}:=G_{ii}$. Thus we can write
\[
\widehat{\bm{D}} = \bm{D} + \bm{F},
\]
where $\bm{F}$ is symmetric with independent, mean-zero entries on and above the diagonal, uniformly bounded support, and
\[
\operatorname{Var}(F_{ij})\le \overline{V}\quad\text{for all }i,j,
\qquad
\operatorname{Var}(F_{ii}) \le \frac{\overline{V}}{n}.
\]
For $i>j$ we impose symmetry by setting $F_{ij}=F_{ji}$.\footnote{Since $\bm{D}$ is a symmetric matrix, we view the primitive parameters as being the entries $F_{ij}$ on and above the diagonal, which are independent.} It will be convenient below to keep using the notation $G_{ii}$ in those parts of the argument where the averaging over $j$ and the variance bound $\operatorname{Var}(G_{ii})\le \overline{V}/n$ play a direct role.

Next, as in \Cref{sec:norm}, let $\widehat{\bm{\Gamma}}$ be the diagonal matrix whose $(i,i)$ diagonal entry is $\widehat{\Gamma}_{ii}=1/\sqrt{-\widehat{D}_{ii}}$, and construct
\[
\underline{\widehat{\bm{D}}}= \widehat{\bm{\Gamma}} \widehat{\bm{D}} \widehat{\bm{\Gamma}}.
\]
Turning to true values (as opposed to hatted estimators), recall the analogous definition $\Gamma_{ii}=1/\sqrt{-D_{ii}}$ and that
\[
{\underline{\bm{D}}}= {\bm{\Gamma}} {\bm{D}} {\bm{\Gamma}}
\]
is the true normalized Slutsky matrix. Write $ \underline{\bm{E}} = \underline{\widehat{\bm{D}}} - \underline{\bm{D}}$. By independence across $j$ and the bound on $\operatorname{Var}(G_{ij}^{h(i,j)})$, we have
\[
\operatorname{Var}(G_{ii}) \leq \frac{\overline{V}}{n},
\]
a fact we will use throughout.

\begin{lem} 
Under the sampling procedure described, the associated error matrix satisfies
\[
 \|\underline{\bm{E}}\| = O_{\mathrm{p}}(n^{1/2})  
\quad \text{as } n\to\infty.
  \]  
\end{lem}

\begin{proof}
Let $\bm{\Phi}=\widehat{\bm{\Gamma}}-\bm{\Gamma}$. 
We start by rewriting the error matrix $\underline{\bm{E}}$:
\[
\underline{\bm{E}} = \underline{\widehat{\bm{D}}} - \underline{\bm{D}} = \widehat{\bm{\Gamma}}\,\widehat{\bm{D}}\,\widehat{\bm{\Gamma}} - \bm{\Gamma}\,\bm{D}\,\bm{\Gamma}.
\]
Recall from the sampling construction that $\widehat{\bm{D}} = \bm{D} + \bm{F}$, where $\bm{F}$ is symmetric with independent, mean-zero entries on and above the diagonal and diagonal entries $F_{ii}=G_{ii}$. Substituting $\widehat{\bm{\Gamma}} = \bm{\Gamma} + \bm{\Phi}$ and this decomposition of $\widehat{\bm{D}}$, we have:
\[
\underline{\bm{E}} = (\bm{\Gamma} + \bm{\Phi})(\bm{D} + \bm{F})(\bm{\Gamma} + \bm{\Phi}) - \bm{\Gamma}\,\bm{D}\,\bm{\Gamma}.
\]
Expanding the product and using the symmetry of all the matrices involved, we obtain:
\[
\begin{aligned}
\underline{\bm{E}} &= \bm{\Gamma}\,\bm{D}\,\bm{\Gamma} + \bm{\Gamma}\,\bm{D}\,\bm{\Phi} + \bm{\Gamma}\,\bm{F}\,\bm{\Gamma} + \bm{\Gamma}\,\bm{F}\,\bm{\Phi} \\
&\quad + \bm{\Phi}\,\bm{D}\,\bm{\Gamma} + \bm{\Phi}\,\bm{D}\,\bm{\Phi} + \bm{\Phi}\,\bm{F}\,\bm{\Gamma} + \bm{\Phi}\,\bm{F}\,\bm{\Phi} - \bm{\Gamma}\,\bm{D}\,\bm{\Gamma} \\
&= \bm{\Gamma}\,\bm{F}\,\bm{\Gamma} + \bm{\Gamma}\,\bm{D}\,\bm{\Phi} + \bm{\Phi}\,\bm{D}\,\bm{\Gamma} + \bm{\Gamma}\,\bm{F}\,\bm{\Phi} + \bm{\Phi}\,\bm{F}\,\bm{\Gamma} + \bm{\Phi}\,\bm{D}\,\bm{\Phi} + \bm{\Phi}\,\bm{F}\,\bm{\Phi}.
\end{aligned}
\]
Our goal is to show that the spectral norm of $\underline{\bm{E}}$ grows at most at rate $\sqrt{n}$ in probability:
\[
\|\underline{\bm{E}}\| = O_{\mathrm{p}}(\sqrt{n}).
\]
Using the triangle inequality for the spectral norm $\|\cdot\|$, we have:
\[
\|\underline{\bm{E}}\| \leq \|\bm{\Gamma}\,\bm{F}\,\bm{\Gamma}\|
 + \|\bm{\Gamma}\,\bm{D}\,\bm{\Phi}\|
 + \|\bm{\Phi}\,\bm{D}\,\bm{\Gamma}\|
 + \|\bm{\Gamma}\,\bm{F}\,\bm{\Phi}\|
 + \|\bm{\Phi}\,\bm{F}\,\bm{\Gamma}\|
 + \|\bm{\Phi}\,\bm{D}\,\bm{\Phi}\|
 + \|\bm{\Phi}\,\bm{F}\,\bm{\Phi}\|.
\]
We will bound the typical size of each term on the right-hand side. The terms $\bm{\Phi}\,\bm{D}\,\bm{\Gamma}$ and $\bm{\Gamma}\,\bm{D}\,\bm{\Phi}$ are transposes of each other, and similarly for $\bm{\Phi}\,\bm{F}\,\bm{\Gamma}$ and $\bm{\Gamma}\,\bm{F}\,\bm{\Phi}$, so their norms can be bounded in the same way; in what follows we bound one representative of each pair and absorb the resulting factor of $2$ into the constants. In this argument, we adopt the standard practice that the meaning of the constant $C$ can change from line to line, but this symbol always stands for a deterministic constant that does not depend on $i$, $j$, or $n$.

\medskip

\paragraph{First term: $\|\bm{\Gamma}\,\bm{F}\,\bm{\Gamma}\|$.} Since $\bm{\Gamma}$ is a diagonal matrix with entries $\Gamma_{ii} = 1/\sqrt{-D_{ii}}$, and the $-D_{ii}$ are bounded away from zero and infinity by assumption, there exist constants $\Gamma_{\min}, \Gamma_{\max} > 0$ such that:
\[
\Gamma_{\min} \leq \Gamma_{ii} \leq \Gamma_{\max}.
\]
Thus, $\|\bm{\Gamma}\| \leq \Gamma_{\max}$. The matrix $\bm{F}$ has independent, mean-zero entries $F_{ij}$ with variances $\operatorname{Var}(F_{ij}) \leq \overline{V}$.

Consider the matrix $\bm{K} = \bm{\Gamma}\,\bm{F}\,\bm{\Gamma}$ with entries $K_{ij} = \Gamma_{ii}F_{ij}\Gamma_{jj}$. The variances of $K_{ij}$ satisfy:
\[
\operatorname{Var}(K_{ij}) = \Gamma_{ii}^2\Gamma_{jj}^2\operatorname{Var}(F_{ij}) \leq \Gamma_{\max}^4\overline{V}.
\]
By standard results on the spectral norm of symmetric random matrices with independent, mean-zero, uniformly bounded entries (see, for example, \citealp{vershynin2018high,bandeiravanhandel2016norm}), there exists a constant $C_1>0$ such that
\[
\Pr\big(\|\bm{K}\|>C_1\sqrt{n}\big)\to 0 \quad \text{as } n\to\infty,
\]
that is, $\|\bm{K}\| = O_{\mathrm{p}}(\sqrt{n})$.

\medskip

\paragraph{Second term: $\|\bm{\Gamma}\,\bm{D}\,\bm{\Phi}\|$.} Since $\bm{\Phi} = \widehat{\bm{\Gamma}} - \bm{\Gamma}$ and $\widehat{\Gamma}_{ii} = 1/\sqrt{-\widehat{D}_{ii}}$, we use a Taylor expansion around $-D_{ii}$ to approximate\footnote{Let $f(x) = x^{-1/2}$ for $x>0$. The derivative is $f'(x) = -\frac{1}{2}x^{-3/2}$. Since $\widehat{D}_{ii}=D_{ii}+G_{ii}$, we have $-\widehat{D}_{ii}=(-D_{ii})-G_{ii}$. Expanding $f((\!-D_{ii})-G_{ii})$ around $x = -D_{ii}$ yields $f((\!-D_{ii})-G_{ii}) = f(-D_{ii}) + f'(-D_{ii})(-G_{ii}) + O(G_{ii}^2)$. Therefore, $\Phi_{ii} = \widehat{\Gamma}_{ii} - \Gamma_{ii} = f(-\widehat{D}_{ii}) - f(-D_{ii}) = \frac{1}{2}\Gamma_{ii}^3G_{ii} + O(G_{ii}^2)$, where $\Gamma_{ii} = f(-D_{ii}) = 1/\sqrt{-D_{ii}}$.} $\Phi_{ii}$:
\[
\Phi_{ii} = \widehat{\Gamma}_{ii} - \Gamma_{ii} = \frac{1}{2}\Gamma_{ii}^3G_{ii} + O(G_{ii}^2).
\]
The variance of $\Phi_{ii}$ satisfies:
\begin{equation}
\operatorname{Var}(\Phi_{ii}) \leq \left(\frac{1}{2}\Gamma_{\max}^3\right)^2\operatorname{Var}(G_{ii}) \leq \frac{C}{n} \label{eq:variance_Phi}
\end{equation}
for some constant $C > 0$. The entries of $\bm{\Gamma}\,\bm{D}\,\bm{\Phi}$ are:
\[
(\bm{\Gamma}\,\bm{D}\,\bm{\Phi})_{ij} = \Gamma_{ii}D_{ij}\Phi_{jj}.
\]
Therefore, the squared Frobenius norm is:
\[
\|\bm{\Gamma}\,\bm{D}\,\bm{\Phi}\|_F^2 = \sum_{i,j} (\Gamma_{ii}D_{ij}\Phi_{jj})^2.
\]
Taking expectations:
\[
\mathbb{E}[\|\bm{\Gamma}\,\bm{D}\,\bm{\Phi}\|_F^2] \leq \Gamma_{\max}^2A^2\sum_{i,j}\mathbb{E}[\Phi_{jj}^2] \leq Cn^2\cdot\frac{1}{n} = Cn,
\]
where in the penultimate step we have used \Cref{eq:variance_Phi}. Using the fact that the Frobenius norm bounds the spectral norm and applying Markov's inequality, we obtain, for any $K>0$,
\[
\Pr\big(\|\bm{\Gamma}\,\bm{D}\,\bm{\Phi}\| > K\sqrt{n}\big)
 \le \Pr\big(\|\bm{\Gamma}\,\bm{D}\,\bm{\Phi}\|_F > K\sqrt{n}\big)
 \le \frac{\mathbb{E}[\|\bm{\Gamma}\,\bm{D}\,\bm{\Phi}\|_F^2]}{K^2 n}
 \le \frac{C}{K^2},
\]
uniformly in $n$, which is exactly the definition of $\|\bm{\Gamma}\,\bm{D}\,\bm{\Phi}\| = O_{\mathrm{p}}(\sqrt{n})$.

\medskip

\paragraph{Third term: $\|\bm{\Gamma}\,\bm{F}\,\bm{\Phi}\|$.} The entries of $\bm{\Gamma}\,\bm{F}\,\bm{\Phi}$ are:
\[
(\bm{\Gamma}\,\bm{F}\,\bm{\Phi})_{ij} = \Gamma_{ii}F_{ij}\Phi_{jj}.
\]
The squared Frobenius norm is:
\[
\|\bm{\Gamma}\,\bm{F}\,\bm{\Phi}\|_F^2 = \sum_{i,j} (\Gamma_{ii}F_{ij}\Phi_{jj})^2.
\]
Taking expectations and using the Cauchy--Schwarz inequality:
\[
\mathbb{E}[\|\bm{\Gamma}\,\bm{F}\,\bm{\Phi}\|_F^2] \leq \Gamma_{\max}^2\overline{V}\sum_{i,j}\mathbb{E}[\Phi_{jj}^2] \leq Cn^2\cdot\frac{1}{n} = Cn.
\]
Therefore, by the same argument as for the second term (using the bound on $\mathbb{E}[\|\bm{\Gamma}\,\bm{F}\,\bm{\Phi}\|_F^2]$ and Markov's inequality), we obtain $\|\bm{\Gamma}\,\bm{F}\,\bm{\Phi}\| = O_{\mathrm{p}}(\sqrt{n})$.
\medskip

\paragraph{Fourth term: $\|\bm{\Phi}\,\bm{D}\,\bm{\Phi}\|$.} The entries of $\bm{\Phi}\,\bm{D}\,\bm{\Phi}$ are:
\[
(\bm{\Phi}\,\bm{D}\,\bm{\Phi})_{ij} = \Phi_{ii}D_{ij}\Phi_{jj}.
\]
The squared Frobenius norm is:
\[
\|\bm{\Phi}\,\bm{D}\,\bm{\Phi}\|_F^2 = \sum_{i,j} (\Phi_{ii}D_{ij}\Phi_{jj})^2.
\]
Taking expectations:
\[
\mathbb{E}[\|\bm{\Phi}\,\bm{D}\,\bm{\Phi}\|_F^2] \leq A^2\sum_{i,j}\mathbb{E}[\Phi_{ii}^2]\mathbb{E}[\Phi_{jj}^2] \leq Cn^2\left(\frac{C}{n}\right)^2 = C.
\]
By Markov's inequality, this implies $\|\bm{\Phi}\,\bm{D}\,\bm{\Phi}\| = O_{\mathrm{p}}(1)$.
\medskip

\paragraph{Fifth term: $\|\bm{\Phi}\,\bm{F}\,\bm{\Phi}\|$.} The entries of $\bm{\Phi}\,\bm{F}\,\bm{\Phi}$ are:
\[
(\bm{\Phi}\,\bm{F}\,\bm{\Phi})_{ij} = \Phi_{ii}F_{ij}\Phi_{jj}.
\]
The squared Frobenius norm is:
\[
\|\bm{\Phi}\,\bm{F}\,\bm{\Phi}\|_F^2 = \sum_{i,j} (\Phi_{ii}F_{ij}\Phi_{jj})^2.
\]
Taking expectations and using the Cauchy--Schwarz inequality,
\[
\mathbb{E}[\|\bm{\Phi}\,\bm{F}\,\bm{\Phi}\|_F^2] \leq \overline{V}\sum_{i \neq j}\mathbb{E}[\Phi_{ii}^2]\mathbb{E}[\Phi_{jj}^2] + \sum_i\mathbb{E} [\bm{\Phi}_{ii}^4]
 = \overline{V}\Big(\sum_i \mathbb{E}[\Phi_{ii}^2]\Big)^2 \leq C.
\]
By Markov's inequality, this implies $\|\bm{\Phi}\,\bm{F}\,\bm{\Phi}\| = O_{\mathrm{p}}(1)$.

\medskip

 \paragraph{Diagonal entries and the recoverable diagonal condition.} We finally verify that, in this sampling scheme,   
 the diagonal of the true Slutsky matrix is recoverable from $\widehat{\bm{D}}$.  Define the recovery map by $\Delta_n(\bm{M})_i:=M_{ii}$.  Then  
\[
 \Delta_n(\widehat{\bm{D}})_i-D_{ii}=\widehat{D}_{ii}-D_{ii}=G_{ii}  
 =\frac{1}{n}\sum_j G_{ij}^{h(i,j)}.  
  \]  
 Each $G_{ii}$ is therefore an average of independent, mean-zero, uniformly bounded random variables.   Therefore $\max_i |\Delta_n(\widehat{\bm{D}})_i-D_{ii}|=o_{\mathrm{p}}(1)$, which is exactly the recoverable diagonal condition in this setting.  
 \end{proof}

 This construction therefore provides a sampling scheme under which the associated error matrix satisfies $\|\underline{\bm{E}}\| = O_{\mathrm{p}}(\sqrt{n})$.  In particular, in any environment where $\Vert \underline{\bm{D}}\Vert$ grows faster than $n^{1/2}$ (as in the block-model example in \Cref{sec:illustration3}), condition (2) of \Cref{Def:SS} is satisfied.  Furthermore, the final paragraph of the proof shows that the diagonal of the true Slutsky matrix is consistently recoverable from $\widehat{\bm{D}}$, so the recoverable diagonal condition (\Cref{as:error_assumption}) also holds under this sampling scheme.

\subsection{Bootstrap test for recoverable structure}\label{app:split_sample_testing}

This subsection formalizes a bootstrap test for significant structure suggested by our theory. 

\paragraph{Setup and estimand.}
Fix \(n\) and a market state \(\bm{\theta}=(\bm{D},\bm{q}^0)\). Suppose the analyst can form \(T\ge 2\) independent estimates of the market state by sampling households in disjoint subsamples:
\[
\widehat{\bm{D}}^{(t)} = \bm{D} + \bm{E}^{(t)}, \qquad
\widehat{\bm{q}}^{0,(t)} = \bm{q}^0 + \bm{\varepsilon}^{(t)},
\quad t=1,\dots,T,
\]
with \((\bm{E}^{(t)},\bm{\varepsilon}^{(t)})\) independent across \(t\) conditional on \(\bm{\theta}\).
Let \(\underline{\widehat{\bm{D}}}^{(t)}\) and \(\underline{\widehat{\bm{q}}}^{0,(t)}\) denote the corresponding normalized objects (as in \Cref{sec:norm}), and define the estimated major eigenspace
\[
\widehat{\mathcal{L}}^{(t)} := \mathcal{L}\bigl(\underline{\widehat{\bm{D}}}^{(t)},\widehat{M}(n)\bigr),
\]
where \(b(n)\) is the threshold from \Cref{Def:SS} and \(\widehat{M}(n):=\tfrac{3}{4}b(n)\) (as in the proof of \Cref{Th:Main}). Let \(\underline{\widehat{\bm{\sigma}}}^{(t)}\) be the intervention implied by our rule when fed the \(t\)-th estimate: project \(\underline{\widehat{\bm{q}}}^{0,(t)}\) onto \(\widehat{\mathcal{L}}^{(t)}\), orient the sign so that spending is positive, and rescale to hit target spending \(s=1\). For \(t\neq t'\), define the directional predicted welfare effect \(\widehat{\tau}^{(t\to t')}\) by \Cref{eq:test}. Assume \(T\) is even and write \(T=2Z\). Pair the splits arbitrarily (say \(2z-1\) with \(2z\)), and for each pair \(z=1,\dots,Z\) define
\[
\widehat{\tau}^{(z)}:=\frac{1}{2}\Bigl(\widehat{\tau}^{(2z-1\to 2z)}+\widehat{\tau}^{(2z\to 2z-1)}\Bigr).
\]
(If \(T\) is odd, discard one split; this is immaterial asymptotically.) Define the estimand
\[
\mu_n(\bm{\theta}) := \mathbb{E}\!\left[\widehat{\tau}^{(1)}\mid \bm{\theta}\right],
\]
where \(\widehat{\tau}^{(1)}\) denotes the generic pair statistic.

\paragraph{Null hypothesis: significant structure.}
For a tolerance \(\epsilon\in(0,1)\), we take the null hypothesis to be that the economy has significant structure in the sense of \Cref{Def:SS}. Concretely, the null is
\[
H_{0}^{\mathrm{SS}}:\ \bm{\theta}\ \text{lies in an environment with significant structure (\Cref{Def:SS}).}
\]
Under \(H_{0}^{\mathrm{SS}}\) our main theorem (\Cref{Th:Main}) implies that the intervention constructed from one split and evaluated on an independent split has predicted welfare per dollar close to one, i.e.,
\[
\mu_n(\bm{\theta})\to 1
\quad\text{uniformly over }\bm{\theta}\in\bm{\Theta}(n),
\]
 as \(n\to\infty\). Thus, for any fixed \(\epsilon\in(0,1)\) and all sufficiently large \(n\), \(H_{0}^{\mathrm{SS}}\) implies the inequality \(\mu_n(\bm{\theta})\ge 1-\epsilon\). The bootstrap test below targets this implication.  
 \medskip

\paragraph{A scaling normalization and a denominator bound.}
 By \Cref{rem:normalization}, we may normalize the underlying environment so that \(\|\bm{q}^0\|=1\).  That, along with condition~(1) of \Cref{Def:SS}, implies \(\|\widehat{\bm{q}}^{0,(t)}\|=O_{\mathrm{p},\bm{\Theta}}(1)\).  Moreover, \Cref{Ass:Asymptotics,as:error_assumption} imply that the diagonal normalization map used to pass from \(\widehat{\bm{q}}^{0,(t)}\) to \(\underline{\widehat{\bm{q}}}^{0,(t)}\) has operator norm uniformly \(O_{\mathrm{p},\bm{\Theta}}(1)\).  Hence \(\|\underline{\widehat{\bm{q}}}^{0,(t)}\|=O_{\mathrm{p},\bm{\Theta}}(1)\).  
  Write \(\widehat{\bm v}^{(t)}:=P_{\widehat{\mathcal{L}}^{(t)}}\,\underline{\widehat{\bm{q}}}^{0,(t)}\). Under the rule, the rescaling denominator satisfies  
\[
\widehat{\bm v}^{(t)}\cdot \underline{\widehat{\bm q}}^{0,(t)}
=\underline{\widehat{\bm q}}^{0,(t)}{}^\top P_{\widehat{\mathcal{L}}^{(t)}}\,\underline{\widehat{\bm q}}^{0,(t)}
=\|\widehat{\bm v}^{(t)}\|^2.
 \]  
 Accordingly,  
 \[  
 |\widehat{\tau}^{(t\to t')}|  
\le
 \|\underline{\widehat{\bm{q}}}^{0,(t')}\|\,\|\underline{\widehat{\bm{\sigma}}}^{(t)}\|  
 =  
 \frac{\|\underline{\widehat{\bm{q}}}^{0,(t')}\|}{\|\widehat{\bm v}^{(t)}\|},  
\]
 by Cauchy--Schwarz and because \(\frac{|\widehat{\lambda}_\ell^{(t')}|}{1+|\widehat{\lambda}_\ell^{(t')}|}\le 1\) for each \(\ell\).  Therefore, whenever \(\|\widehat{\bm v}^{(t)}\|\) is bounded away from zero, the directional cross-fit statistics \(\widehat{\tau}^{(t\to t')}\) (hence also the pair statistics \(\widehat{\tau}^{(z)}\)) are \(O_{\mathrm{p},\bm{\Theta}}(1)\).  Under the significant-structure null, \Cref{lem:bound_projection_product} and the proof of \Cref{lem:control_E_dot} imply that there exists a constant \(\kappa>0\) (depending only on the environment constants in \Cref{Def:SS}, not on \(\bm{\theta}\)) such that  
\[
\inf_{\bm{\theta}\in\bm{\Theta}(n)}
\mathbb{P}\!\left(\|\widehat{\bm v}^{(t)}\|\ge \kappa\,\middle|\,\bm{\theta}\right)\to 1
\quad\text{as }n\to\infty.
  \]  

\paragraph{Bootstrap lower confidence bound and one-sided certification.}
Let \(\overline{\tau}:=\frac{1}{Z}\sum_{z=1}^Z \widehat{\tau}^{(z)}\). Define the nonparametric bootstrap as follows: conditional on the observed \(\widehat{\tau}^{1:Z}\), draw weights \((W_1,\dots,W_Z)\sim \mathrm{Multinomial}(Z;1/Z,\dots,1/Z)\) and set \(\overline{\tau}^{*}:=\frac{1}{Z}\sum_{z=1}^Z W_z\,\widehat{\tau}^{(z)}\). Let \(\widehat{Q}_{1-\alpha}\) be the \((1-\alpha)\) conditional quantile of \(\overline{\tau}^{*}-\overline{\tau}\) given \(\widehat{\tau}^{1:Z}\), and define the one-sided basic bootstrap lower confidence bound \(L_{1-\alpha}:=\overline{\tau}-\widehat{Q}_{1-\alpha}\). We use the flag \(L_{1-\alpha}<1-\epsilon\): values \(L_{1-\alpha}\ge 1-\epsilon\) certify \(\mu_n(\bm{\theta})\ge 1-\epsilon\) at one-sided confidence level \(1-\alpha\).

\begin{prop}[Conditional asymptotic coverage]\label{thm:bootstrap_split_sample}
For any fixed \(\alpha\in(0,1)\),
\[
\mathbb{P}\!\left(\mu_n(\bm{\theta})\ge L_{1-\alpha}\,\middle|\,\bm{\theta}\right)\to 1-\alpha
\qquad\text{as }Z\to\infty.
\]
Thus, \([L_{1-\alpha},\infty)\) is an asymptotically valid one-sided \((1-\alpha)\) confidence set for \(\mu_n(\bm{\theta})\). In particular, observing \(L_{1-\alpha}\ge 1-\epsilon\) provides one-sided \((1-\alpha)\) certification for \(\mu_n(\bm{\theta})\ge 1-\epsilon\).
\end{prop}

\begin{proof}[Proof sketch]
Conditional on \(\bm{\theta}\), the boundedness of all the random variables involved guarantees that a CLT is applicable to \(\overline{\tau}\). Under the same moment condition, the nonparametric bootstrap for the sample mean is consistent \citep[Chapter~23]{vandervaart1998}: conditional on the data, \(\sqrt{Z}(\overline{\tau}^*-\overline{\tau})\) converges in probability to the same Gaussian limit as \(\sqrt{Z}(\overline{\tau}-\mu_n(\bm{\theta}))\). Therefore the conditional \((1-\alpha)\)-quantile \(\widehat{Q}_{1-\alpha}\) consistently estimates the \((1-\alpha)\)-quantile of \(\overline{\tau}-\mu_n(\bm{\theta})\), yielding the stated one-sided coverage.
\end{proof}

\section{Microfoundation of block model of demand} \label{app:linear}
Consider the standard quasilinear quadratic utility function: 
$$
U(\bm{q})=
\bm{a}^{\tr} \bm{q}-\frac{1}{2}\bm{q}^\tr\bm{H}\bm{q},
$$
where 
$\bm{a}\in \mathbb{R}^n_{+}$ and $\bm{H}$ is an $n\times n$ symmetric matrix.

For a given price vector $\bm{p}\in \mathbb{R}^n_{+}$, the first-order condition for maximizing \( U(\bm{q})-\bm{p}\cdot\bm{q} \)  with respect to \( \bm{q} \) yields the demand function:
\[
\bm{q}(\bm{p}) = \bm{H}^{-1}(\bm{a} - \bm{p}).
\]
Differentiating it with respect to prices yields the Jacobian:
\[
\bm{D}=\frac{\partial \bm{q}}{\partial \bm{p}} = -\bm{H}^{-1}.
\]  
 Strict concavity of the utility of the representative consumer is equivalent to $\bm{D}$ being  
 negative definite, i.e., all the eigenvalues of $\bm{D}$ are strictly negative.  This is equivalent to all eigenvalues of $\bm{H}$ being strictly positive.  Moreover, having an eigenvalue of $\bm{D}$ that grows large (in absolute value) with the market size requires having an eigenvalue of $\bm{H}$ that goes to $0$ as the market grows large.  
 
\subsection{Imposing block structure}
Assume now that $\bm{H}$ has $K$ blocks each of size $m$ and that:
\[
H_{ij} =
\begin{cases}
z>0,         & i=j,\\[4pt]
z_{\text{In}}>0,   & i\ne j \text{ and } 
             \bigl\lfloor \tfrac{i-1}{m} \bigr\rfloor
             =\bigl\lfloor \tfrac{j-1}{m} \bigr\rfloor
             \quad\text{(same block)},\\[10pt]
z_{\text{Ex}}<0,   & \text{otherwise} \quad\text{(different blocks).}
\end{cases}
\]
Hence, the marginal utility from consuming products is decreasing within the same category and increasing across categories.
The matrix $\bm{D}=-\bm{H}^{-1}$ has three distinct eigenvalues:
\begin{eqnarray*}
\lambda_3 &=& -\frac{1}{z-z_{\text{In}}} \quad\quad \text{with multiplicity } K(m-1)\\
\lambda_2 &=& -\frac{1}{z+z_{\text{In}}(m-1) -mz_{\text{Ex}}} \quad\quad \text{with multiplicity } (K-1)\\
\lambda_1 &=& -\frac{1}{z+ z_{\text{In}}(m-1) + m(K-1)z_{\text{Ex}}} \quad\quad \text{with multiplicity } 1
\end{eqnarray*}

Note that $-\bm{H}$ is the Hessian of $U$. Hence each eigenvalue of $\bm{D}=-\bm{H}^{-1}$ is the reciprocal of an eigenvalue of the Hessian of $U$. Therefore, eigenvalue $1/\lambda_i$ measures the curvature of the utility function $U$ when we change consumption in the direction of the associated eigenvector $\bm{u}^i$ of the Hessian of $U$. Concavity requires that these eigenvalues are all negative. In particular, $\lambda_1\leq 0$ ensures that the marginal utility of increasing quantities in the all-ones direction is decreasing. The corresponding eigenvalue of $\bm{H}$ in the all-ones direction is the sum of three terms:
\begin{itemize}
\item $z>0$: The contribution to curvature of increasing consumption of product $i$.
\item $(m-1)z_{\text{In}}>0$: The contribution to curvature of increasing quantity of other products in the block.
\item $m(K-1)z_{\text{Ex}}<0$: The contribution to curvature of increasing quantity of products in other blocks.
\end{itemize}

Thus, while $\lambda_1\leq 0$ is required to keep concavity in this direction, $1/\lambda_1\to 0^-$ requires adding sufficient interblock interactions as we grow the market to flatten the utility function in this direction---so that, in the limit, the utility function scales linearly with quantity in this direction. 

\subsubsection{From $\bm{D}$ induced by the representative consumer to $\bm{D}$ specified directly in \Cref{sec:illustration3}}
In \Cref{sec:illustration3} we defined a Slutsky $\bm{D}$ associated with state $\bm{\theta}$ to be such that $D_{ii}=-1$, $D_{ij}=\alpha_{\bm{\theta}}<0$ if $i$ and $j$ belong to different product categories, and $D_{ij}=\omega_{\bm{\theta}}>0$ if $i$ and $j$ belong to the same category. Such a $\bm{D}$ has three distinct associated eigenvalues
$$
\begin{array}{ll}
\lambda_3(\bm{\theta})=-1 - \omega_{\bm{\theta}}   & \text{with multiplicity $K(m-1)$}\\
\lambda_2(\bm{\theta})=\lambda_3(\bm{\theta}) +m (\omega_{\bm{\theta}}-\alpha_{\bm{\theta}})   & \text{with multiplicity $(K-1)$}\\
\lambda_1(\bm{\theta})=\lambda_2(\bm{\theta})+mK\alpha_{\bm{\theta}}   & \text{with multiplicity $1$}
\end{array}
$$
 So the equivalence of the $\bm{D}$ generated by the representative consumer and the example outlined in \Cref{sec:illustration3} requires defining, for each state $\bm{\theta}$, a triple $(z^*(\bm{\theta}),z^*_{\text{In}}(\bm{\theta}),z^*_{\text{Ex}}(\bm{\theta}))$ so that the three distinct eigenvalues in the two approaches coincide. This is equivalent to solving, state by state, the following system:  
\[
\begin{aligned}
-1 &=\omega_{\bm{\theta}}- \frac{1}{z^*(\bm{\theta})-z^*_{\text{In}}(\bm{\theta})}
        ,\\[8pt]
 \omega_{\bm{\theta}}-\alpha_{\bm{\theta}} &= \frac{z^*_{\text{In}}(\bm{\theta})-z^*_{\text{Ex}}(\bm{\theta})}{[z^*(\bm{\theta})-z^*_{\text{In}}(\bm{\theta}) + m(z^*_{\text{In}}(\bm{\theta})-z^*_{\text{Ex}}(\bm{\theta}))](z^*(\bm{\theta})-z^*_{\text{In}}(\bm{\theta}))},\\[8pt]
 \alpha_{\bm{\theta}} &=\frac{z^*_{\text{Ex}}(\bm{\theta})}{[z^*(\bm{\theta})-z^*_{\text{In}}(\bm{\theta}) + m(z^*_{\text{In}}(\bm{\theta})-z^*_{\text{Ex}}(\bm{\theta}))][z^*(\bm{\theta})-z^*_{\text{In}}(\bm{\theta}) + m(z^*_{\text{In}}(\bm{\theta})-z^*_{\text{Ex}}(\bm{\theta}))+n z^*_{\text{Ex}}(\bm{\theta})]}.  
 \end{aligned}
\]

\subsection{Equilibrium}
 In what follows we fix the true market state $\bm{\theta}$.  We want to illustrate that limiting economies where $|\lambda_1(\bm{\theta})|$ goes to infinity are consistent with well-behaved equilibrium outcomes.  For this we consider the simple case where $a_i=a$ and $c_i=c$ for all $i$; this simplification, together with the block model structure, guarantees the existence of a symmetric equilibrium.  Recall that demand is  
\[
\bm{q} \;=\; -\bm{D}\,(\bm{a}-\bm{p}).
  \]  
Let $p_j=p$ for all $j\neq i$ and let $p_i$ be the price of firm $i$. Then, using the definition of $\bm{D}$, firm $i$'s demand is
\[
\begin{aligned}
q_i(p_i,p) &=
-\,\Bigl[
    -(a - p_i)
  + \omega_{\bm{\theta}} \!\!\!\sum_{\substack{j \,\text{same block}\\ j\neq i}}\!\!(a - p)
  + \alpha_{\bm{\theta}} \!\!\!\sum_{k \,\text{other blocks}}\!\!(a - p)
  \Bigr] \\[6pt]
&=
-\, \Bigl[
    -(a - p_i)
  + \omega_{\bm{\theta}}\,(m-1)(a - p)
  + \alpha_{\bm{\theta}}\,m(K-1)(a - p)
   \Bigr].
\end{aligned}
\]
Recall that $\lambda_1(\bm{\theta})= -1 + \omega_{\bm{\theta}}(m-1) + \alpha_{\bm{\theta}} m(K-1)$. Hence, 
\[
\begin{aligned}
q_i(p_i,p)
&= -\Bigl[-(a-p_i)+\left(\lambda_1(\bm{\theta})+1\right)(a-p)\Bigr] \\[4pt]
&= -\Bigl[ p_i-p +\lambda_1(\bm{\theta})(a-p)\Bigr] .
\end{aligned}
\]

The profit function is \(\pi_i(p_i,p) = (p_i-c)\,q_i(p_i,p)\). The first-order condition is:
\[
\frac{\partial\pi_i(p_i,p)}{\partial p_i}=0 
\;\Longrightarrow\;
q_i(p_i,p) - (p_i-c)=0 .
\]
Imposing symmetry \(p_i=p\), 
$$q_i(p,p)= |\lambda_{1}(\bm{\theta})|(a-p)$$
so equilibrium price and quantity are
\[
p \;=\; \frac{|\lambda_{1}(\bm{\theta})|}{1+|\lambda_{1}(\bm{\theta})|}a + \frac{c}{1+|\lambda_{1}(\bm{\theta})|} \quad \text{and} \quad q= \frac{|\lambda_{1}(\bm{\theta})|}{1+|\lambda_{1}(\bm{\theta})|} (a-c)
\]
which depends on $K$ only via $|\lambda_{1}(\bm{\theta})|$. Under the assumption of significant structure, $|\lambda_{1}(\bm{\theta})|$ goes to $\infty$ as $n$ increases across states and so $p$ goes to $a$ and $q$ goes to $a-c$. 

\section{Hedonic Models and Significant Structure} \label{app:hedonic}

The (un-normalized) Slutsky matrix in \citet{pelligrino2021} is $-\bm{B}^{-1}$, where $\bm{B}=\bm{I} + \alpha(\bm{\Sigma}-\bm{I})$. The matrix $\bm{\Sigma}$ is positive semidefinite because it can be written as $\bm{V}^\tr \bm{V}$, where the columns of $\bm{V}$ are the characteristic vectors of various products.

It follows from this that all eigenvalues of $\bm{B}$ are real numbers bounded below by $1-\alpha$, and all eigenvalues of $-\bm{B}^{-1}$ are at most $1/(1-\alpha)$ in magnitude. Pellegrino uses the value $\alpha=0.12$, which prevents any eigenvalue from exceeding $1.14$. Once we normalize to obtain $\underline{\bm{D}}$ the eigenvalues of the resulting matrix are somewhat different, but they can still be bounded by a constant by elaborating this argument. Numerically we see that the normalization makes little difference.

\end{document}

%% file: supporting_files/preamble.tex
\usepackage[multiple]{footmisc} 
\usepackage[normalem]{ulem}
\usepackage{setspace}
\usepackage{amsmath}
\usepackage{amsfonts}
\usepackage{amssymb}
\usepackage{bbm}  %
\usepackage{enumerate}
\usepackage{accents}
\usepackage[breaklinks]{hyperref}
\usepackage{pictex}
\usepackage{amsthm}
\usepackage{breakcites}
\usepackage{graphics,epsfig,verbatim,bm,latexsym,url,amsbsy}
\usepackage{rotating}
\usepackage[authoryear,round]{natbib}
\usepackage{mathrsfs}
\usepackage{textgreek}
\usepackage{multirow}
\usepackage{graphicx}
\usepackage{subcaption}
\usepackage{array}
\usepackage{url}

\usepackage[top=1in,bottom=1.25in,left=1in,right=1.5in,marginparwidth=1in, marginparsep=3mm]{geometry}

\usepackage{graphicx}

\usepackage{float}
\usepackage{tikz}
\usepackage{caption}

\usepackage{bm}
\usepackage{tikz-cd}

\newcommand{\tr}{\mathsf{T}}

\usepackage{pstricks, enumerate, pst-node, pst-text, pst-plot}
\usepackage{xparse}

\usepackage{graphicx}
\graphicspath{{screenshots/}{images/}} %

\theoremstyle{definition} \newtheorem{remark}{Remark}
\theoremstyle{definition}

\long\def\symbolfootnote[#1]#2{\begingroup%
\def\thefootnote{\fnsymbol{footnote}}\footnote[#1]{#2}\endgroup}

\usepackage{footnote}
\makesavenoteenv{tabular}
\usepackage{fancyhdr}
\newcommand{\documenttitle}{Thesis}

\renewcommand{\max}{\operatornamewithlimits{max}}
\renewcommand{\min}{\operatornamewithlimits{min}}

\newcommand{\be}{\begin{equation}}
\newcommand{\ee}{\end{equation}}
\newcommand{\bes}{\begin{equation*}}
\newcommand{\ees}{\end{equation*}}

\usepackage{amsmath}
\usepackage{epsfig,graphics}
\usepackage{etoolbox}

\usepackage{accents}

\renewcommand{\Pr}{\mathbb P}

\newcommand{\rdots}{\mathinner{%
  \mkern1mu\raise1pt\hbox{.}%
  \mkern2mu\raise4pt\hbox{.}%
  \mkern2mu\raise7pt\vbox{\kern7pt\hbox{.}}\mkern1mu}}

\DeclareMathOperator{\diag}{diag}

\DeclareFontFamily{U}{mathx}{\hyphenchar\font45}
\DeclareFontShape{U}{mathx}{m}{n}{
      <5> <6> <7> <8> <9> <10>
      <10.95> <12> <14.4> <17.28> <20.74> <24.88>
      mathx10
      }{}
\DeclareSymbolFont{mathx}{U}{mathx}{m}{n}
\DeclareFontSubstitution{U}{mathx}{m}{n}
\DeclareMathAccent{\widecheck}{0}{mathx}{"71}
\DeclareMathAccent{\wideparen}{0}{mathx}{"75}

\renewcommand{\Pr}{ P}

\usepackage{palatino}

\hypersetup{colorlinks=true, urlcolor=black!50!blue, citecolor=black!50!blue, linkcolor=black!50!blue}

\makeatletter
\def\@footnotecolor{purple!30!blue}
\define@key{Hyp}{footnotecolor}{%
 \HyColor@HyperrefColor{#1}\@footnotecolor%
}
\patchcmd{\@footnotemark}{\hyper@linkstart{link}}{\hyper@linkstart{footnote}}{}{}
\makeatother
\hypersetup{footnotecolor=black}

\usepackage{thmtools}
\usepackage{nameref}
\usepackage[nameinlink]{cleveref}
\crefname{lem}{Lemma}{Lemmas}
\Crefname{lem}{Lemma}{Lemmas}
\crefname{assumption}{Assumption}{Assumptions}
\Crefname{assumption}{Assumption}{Assumptions}

\newcommand\nc{\newcommand}
\nc\on{\operatorname}
\theoremstyle{definition} \newtheorem{thm}{Theorem}
\theoremstyle{definition} 

\theoremstyle{definition} \newtheorem{definition}{Definition}
\theoremstyle{definition} \theoremstyle{remark}
\theoremstyle{definition} 
\theoremstyle{definition} 
\theoremstyle{definition} \theoremstyle{plain}
\theoremstyle{definition} \newtheorem{condition}{Condition}
\theoremstyle{definition} \newtheorem{lem}{Lemma}
\theoremstyle{definition} 
\theoremstyle{definition} \newtheorem{prop}{Proposition}
\theoremstyle{definition} 
\theoremstyle{definition} \newtheorem{assumption}{Assumption}

\theoremstyle{definition} 
\theoremstyle{definition} 

\newtheorem*{PropertyA}{Property NSD}
\theoremstyle{definition} 
\theoremstyle{definition} 
\theoremstyle{definition}

\theoremstyle{definition}

\renewcommand{\paragraph}[1]{\noindent\textit{#1}}